\newif\ifsubmission
\ifsubmission
\documentclass{llncs}

\else
\documentclass[11pt]{article}
\usepackage{fullpage}
\usepackage[margin=1in]{geometry}
\fi

\usepackage[utf8]{inputenc}
\usepackage[T1]{fontenc}

\usepackage{lmodern}
\usepackage{amsmath,amssymb,amsthm,mathtools}
\usepackage{enumitem}
\usepackage{bm}
\usepackage{xcolor}
\usepackage{mathtools}
\usepackage{xparse}
\usepackage{braket}
\usepackage{graphicx} 
\usepackage{verbatim}
\usepackage{amsfonts}
\usepackage{tcolorbox}
\tcbuselibrary{skins, breakable}
\usepackage{pgfplots}
\usepgfplotslibrary{fillbetween,groupplots}
\pgfplotsset{compat=1.18}

\ifsubmission \else
\usepackage{palatino}
\fi
\usepackage[normalem]{ulem}
\usepackage{framed}
\usepackage[utf8]{inputenc}
\usepackage[pdfstartview=FitH,pdfpagemode=UseNone,colorlinks,linkcolor=blue,filecolor=blue,citecolor=Violet,urlcolor=red]{hyperref}
\usepackage{cmap}
\usepackage[T1]{fontenc}
\usepackage{multirow}
\usepackage{float}
\usepackage[section]{placeins} 
\usepackage{mathtools} 
\usepackage[dvipsnames]{xcolor}
\usepackage{amsmath,amsthm,amssymb,algpseudocode,algorithm,cryptocode}
\usepackage{mathrsfs}
\usepackage[capitalize]{cleveref}
\usepackage{dashbox}
\usepackage{physics}
\usepackage{xspace} 
\usepackage{braket}
\usepackage[normalem]{ulem}
\usepackage{soul,xcolor}
\usepackage{tikz}
\usetikzlibrary{quantikz} 
\usepackage{soul}
\usepackage{enumitem}
\usepackage{breakcites}
\newtcolorbox{protocolbox}[1]{
  enhanced, breakable,
  colback=white, colframe=black, boxrule=.8pt, arc=1mm,
  left=2mm,right=2mm,top=1.2mm,bottom=1.2mm,
  title=\textbf{#1}
}

\newcommand{\E}{\mathbb{E}}

\newcommand{\bits}{\{0,1\}}
\newcommand{\Soft}{\mathrm{Soft}}

\providecommand{\secp}{\lambda}
\providecommand{\negl}{\mathsf{negl}}
\providecommand{\Adv}{\mathsf{Adv}}
\providecommand{\bbN}{\mathbb{N}}
\providecommand{\Ext}{\mathsf{Ext}}

\newcommand{\Com}{\mathsf{Com}}

\newcommand{\state}{\mathsf{st}}

\newcommand{\Ver}{\mathsf{Ver}}

\newcommand{\Rec}{\mathsf{Rec}}

\newcommand{\poly}{\mathsf{poly}}
\newcommand{\NReps}{\mathsf{NReps}}

\usepackage{aliascnt}
\input{james-macros}

\begin{document}

\ifsubmission
\title{On the Cryptographic Structure Required for Verifying Qubits}
\author{}
\institute{}
\else
\title{On the Cryptographic Structure Required for \\ Verifying Qubits}

\renewcommand{\thefootnote}{\fnsymbol{footnote}}

\author{
James Bartusek\footnotemark[1]
\and
Itay Shalit\footnotemark[2]
}

\footnotetext[1]{Columbia University. Email: \texttt{bartusek.james@gmail.com}.}
\footnotetext[2]{Stanford University. Email: \texttt{ishalit@stanford.edu}. Supported by a Shoucheng Zhang Graduate Fellowship.}
\date{}
\maketitle

\setcounter{footnote}{0}
\renewcommand{\thefootnote}{\arabic{footnote}}
\fi

\begin{abstract}

    Classically testing for the presence of anti-commuting operators on a quantum device is a critical tool underpinning recent progress in classical verification of quantum computation. While we can base such tests on certain cryptographic assumptions, known results require highly structured assumptions, e.g.\ trapdoor claw-free functions.

    In this work, we seek to explain this state of affairs by \emph{constructing} strong cryptography from (certain forms of) classical tests of anti-commutation. In particular, we formulate the notion of a ``test of non-commutation'' (ToNC), which is an interactive protocol between a quantum prover and classical verifier where, in the final round, the prover applies one of two binary observables $P_0,P_1$ depending on the verifier's challenge bit $c$. An $(\epsilon,\delta)$-ToNC is any such protocol where there exists a quantum strategy that attains $\epsilon$ advantage in making the verifier accept, but any quantum strategy with \emph{commuting} $P_0,P_1$ can attain at most $\delta$ advantage. We then show the following results.

    \begin{itemize}
        \item $(\epsilon,\delta)$-ToNC implies (classical-communication) key agreement (KA) for any $\delta < \frac{5\epsilon-1}{4}$.
        \item $(\epsilon,\delta)$-ToNC plus one-way functions implies (classical-communication) oblivious transfer (OT) for any $\delta < \frac{\epsilon^2 + \epsilon}{2}$.
    \end{itemize}

    Along the way, we develop tools for and provide the first known results on hardness amplification for \emph{post-quantum} KA and OT, where the communication is classical but the adversary may be quantum. In particular, we prove the following results of independent interest.

    \begin{itemize}
        \item Post-quantum hard-core measure theorem: For any efficiently sampleable, high min-entropy distribution $D$ over classical $(x,b)$ such that quantum circuits have advantage at most $\delta$ in predicting $b$ given $x$, there exists a sub-distribution $M \preceq D$ of density $1-\delta$ such that $b$ is (nearly) optimally quantum-hard to predict on $(x,b) \gets M$.
        \item Post-quantum interactive XOR lemma: Given any classically-interactive protocol such that quantum adversaries have advantage at most $\delta$ in guessing a private challenger bit $b$, repeating sequentially yields a protocol where quantum adversaries have advantage at most $\delta^2 + \negl(\secp)$ in guessing the XOR of the two challenger bits $b_1 \oplus b_2$.
    \end{itemize}

\end{abstract}

\ifsubmission
\else
\newpage
\tableofcontents
\newpage 
\fi

\section{Introduction}

The qubit, being the basic unit of quantum information, lies at the foundation of quantum computation and the theory of quantum mechanics itself. While one can define a qubit as a unit vector in $\mathbb{C}^2$, it can alternatively be understood as any pair of anti-commuting binary observables \cite{vidick2020fsmp}. This idea traces its origins back to the Heisenberg view of quantum mechanics, and is justified by the fact that any pair of anti-commuting binary observables is, up to an isometry, equivalent to the Pauli observables $\sigma_X$ and $\sigma_Z$ operating on $\mathbb{C}^2$. It turns out that this definition has particularly nice operational properties, as discussed at length in \cite{vidick2020fsmp}.

\paragraph{Verifying qubits.} A natural question of importance to the theory of quantum mechanics, and especially so in the early days of quantum information-processing, is whether it is possible to \emph{classically} test for the presence of a qubit on a purported quantum device. That is, can we validate the existence of anti-commuting operators just by sending classical signals and receiving classical responses? This motivates the following definition of a ``qubit test''.

\begin{quote}
\noindent\emph{\textbf{Qubit test (informal).}}
A qubit test is a classical interactive protocol between a classical (polynomial-time) verifier and a quantum (polynomial-time) prover, whose description includes two binary observables $P_0,P_1$. In the final round, the verifier sends a challenge $c\in\{0,1\}$, and the prover returns the bit obtained by applying $P_c$ to its current state. For any prover that makes the verifier accept with probability greater than some threshold $\alpha$, it must be the case that $P_0P_1 \approx -P_1P_0$.
\end{quote}

While qubit tests are of inherent interest as a bridge between the classical and quantum worlds, in recent years they have also played a key role in the construction of several more advanced classical verification protocols. For example, methods developed to test for the presence of qubits underpin certifiable randomness protocols \cite{BCMVV}, classical verification of general BQP computation \cite{Mahadev2018ClassicalVerification,natarajan2023boundingquantumvaluecompiled}, verifiable remote state preparation (and applications) \cite{Gheorghiu2019ComputationallySecureAC,zhang:LIPIcs.ITCS.2025.96,cryptoeprint:2023/1492}, and computational self-testing (and applications) \cite{metger_et_al:LIPIcs.ITCS.2021.19,Metger:2020pem}, to name a few.

\paragraph{Cryptographic instantiations.} By now, the community has developed several constructions of qubit tests. The first result was due to \cite{BCMVV}, who instantiated them\footnote{The construction of~\cite{BCMVV} has string-valued final answers
and therefore does not literally fit the binary-output syntax above.
Appendix~\ref{sec:generalized-tonc} introduces a structured multi-output
generalization of ToNC that captures this construction. We prove that this generalized form of ToNC still implies
key agreement.} from the assumption of learning with errors (LWE), or more generally from trapdoor claw-free functions that satisfy an ``adaptive hard-core bit'' property. Followup work has diversified and broadened the set of assumptions known to imply qubit tests, yielding constructions from ``plain'' trapdoor claw-free functions \cite{kahanamokuMeyer2022, 10.1007/978-3-031-38554-4_6}, quantum fully-homomorphic encryption \cite{kalai2022quantumadvantagenonlocalgame}, cryptographic group actions \cite{10.1007/978-3-031-22318-1_10}, oblivious state preparation \cite{BK25}, and the lattice isomorphism problem (LIP) \cite{10.1007/978-3-032-12287-2_8}.

However, a common pattern can be identified behind these results. In particular, all instantiations rely on mathematically ``structured'' assumptions (LWE, group actions, LIP) that imply strong forms of asymmetric cryptography such as (classical-communication) oblivious transfer. In fact, it remains open to construct qubit tests even given some basic forms of asymmetric structure, such as a (plain) trapdoor function.

\paragraph{Comparison with proofs of quantumness.} This situation can be contrasted with what is known about the relaxed goal of (classically) verifiable \emph{proofs of quantumness}. Here, the goal is to devise a protocol where quantum provers have a strictly greater advantage than classical provers, but do not guarantee anything beyond the fact that the prover is non-classical. It is known that efficiently-verifiable proofs of quantumness exist in the random oracle model (that is, from a heuristic use of unstructured cryptography) \cite{10.1145/3658665} and \emph{inefficiently}-verifiable proofs of quantumness exist if (and only if) one-way puzzles exist~\cite{MoriameCryptoCharacterization2025}, which are believed to be weaker even than one-way functions~\cite{khurana2024commitmentsquantumonewayness, Kretschmer2021, Morimae2022}. 

While proofs of quantumness have relevance to the goal of establishing quantum advantage, they do not appear to serve as a bedrock for obtaining more advanced forms of classical verification. This state of affairs brings into focus a potential dividing line between proofs of quantumness and qubit tests, both in terms of functionality (the methods behind qubit tests appear to have much more powerful applications) and feasibility (proofs of quantumness can be instantiated with hash functions, while qubit tests are only known from structured assumptions). The goal of this work is to gain a more principled understanding of why this line exists.

\paragraph{Tests of non-commutation.} In fact, we consider a weaker primitive which merely certifies that the prover's observables do not \emph{fully} commute, rather than that they approximately anti-commute. We refer to this primitive as a test of non-commutation (ToNC). It differs from a qubit test in the soundness guarantee: For any quantum polynomial-time prover that succeeds with probability greater than $\alpha$, we require that $P_0P_1 \neq P_1P_0$ in place of the stronger conclusion that $P_0P_1 \approx -\,P_1P_0.$ We proceed to show that ToNCs imply certain strong forms of cryptography, which establishes the same for full-fledged qubit tests.

\subsection{Results}\label{subsec:results}

We give constructions of fully-secure oblivious transfer and key agreement protocols with purely classical communication, starting from (even weakly-secure) tests of non-commutation (ToNC). Along the way, we obtain the first \emph{post-quantum} hardness amplification results for bit agreement (BA) and oblivious transfer (OT), which are of independent interest. Our main novel techniques go into building weakly-secure OT, and amplifying both weakly-secure BA and OT to their fully-secure variants. In order to present these results precisely, we first fix definitions of ToNC, weak bit agreement, and OT.

\paragraph{Definitions.} We begin by formalizing our notion of ToNC. As discussed above, we consider interactive protocols where the prover's strategy consists of three parts: a quantum polynomial-time (QPT) interactive machine $P_\prep$ that interacts with a PPT verifier $V$, and two (potentially non-commuting) binary observables $P_0,P_1$. The interactive phase is denoted as \[\ket{\psi},c,a^* \gets \langle P_\prep(1^\secp),V(1^\secp)\rangle,\] where $\ket{\psi}$ is the prover's final state, $c \in \{0,1\}$ is a challenge bit sent to the prover, and $a^*$ is the ``correct'' answer bit kept private by the verifier. We then define a $(\epsilon,\delta)$-ToNC as satisfying the following properties.\footnote{The reader may notice that the verification procedure implicit in this definition is not fully general. Indeed, rather than apply an arbitrary predicate on its state and the prover's response $a$, the verifier simply checks whether $a = a^*$. In fact, in the body we define a ToNC more generally, and then show that any binary-output ToNC can be compiled into the form written here (which we call a ``normal form'' ToNC) while preserving the completeness-soundness gap.
}

\begin{itemize}
        \item \textbf{ToNC Completeness}: It holds that
        \[\Pr\left[a = a^*: \begin{array}{r}\ket{\psi},c,a^* \gets \langle P_\prep(1^\secp),V(1^\secp)\rangle \\ a \gets P_c(\ket{\psi})\end{array}\right] \geq \frac{1}{2} + \frac{\epsilon}{2}.\]
        \item \textbf{ToNC Soundness}: For any QPT strategy $\widetilde{P}_\prep, \widetilde{P}_0,\widetilde{P}_1$ such that $\widetilde{P}_0\widetilde{P}_1 = \widetilde{P}_1\widetilde{P}_0$, \[\Pr\left[a = a^* : \begin{array}{r}\ket{\psi},c,a^* \gets \langle \widetilde{P}_\prep(1^\secp),V(1^\secp)\rangle \\ a \gets \widetilde{P}_c(\ket{\psi})\end{array}\right] \leq \frac{1}{2} + \frac{\delta}{2} + \negl(\secp).\]
    \end{itemize}

It may be useful to keep a running example of a ToNC in mind. Consider the ``compiled'' CHSH game \cite{kahanamokuMeyer2022,kalai2022quantumadvantagenonlocalgame}. Here, the (honest) state $\ket{\psi}$ is a single qubit either in the standard or the Hadamard basis, and the challenge $c$ tells the (honest) prover whether to measure it in the $X+Z$ or $X-Z$ basis. In the compiled CHSH game, we have $\epsilon = 2\cos^2(\pi/8)-1-\operatorname{negl}(\lambda) \approx 0.7$, while there are several instantiations under various cryptographic assumptions attaining the soundness bound $\delta = 0.5$ (e.g. this follows from the analysis in \cite{10.1007/978-3-031-38554-4_6,BK25}).

Next, we consider game-based definitions for (weak) bit agreement and (weak) OT. Bit agreement is an interactive protocol that takes place between two parties $A$ and $B$, denoted as \[k_A,k_B,\tau \gets \langle A(1^\secp),B(1^\secp)\rangle,\] where $k_A \in \{0,1\}$ is $A$'s output, $k_B \in \{0,1\}$ is $B$'s output, and $\tau$ denotes the (classical) transcript of interaction that occurs between $A$ and $B$. We then define a $(\epsilon,\delta)$-BA as satisfying the following properties.
    \begin{itemize}
        \item \textbf{Lack of bias:} \[\Bigg|\Pr_{(k_A,k_B,\tau) \gets \langle A(1^\secp),B(1^\secp)\rangle}[k_A = 0] - \frac{1}{2} \Bigg| = \negl(\secp), \ \ \ \Bigg|\Pr_{(k_A,k_B,\tau) \gets \langle A(1^\secp),B(1^\secp)\rangle}[k_B = 0] - \frac{1}{2}\Bigg| = \negl(\secp),\] 
        \item \textbf{BA Correctness}: 
        \[\Pr_{(k_A,k_B,\tau) \gets \langle A(1^\secp),B(1^\secp)\rangle}[k_A = k_B] \geq \frac{1}{2} + \frac{\epsilon}{2}-\negl(\secp).\]
        \item \textbf{BA Security}: For any QPT adversary $\widetilde{E}$, 
        \[\Pr_{(k_A,k_B,\tau) \gets \langle A(1^\secp),B(1^\secp)\rangle}\left[\widetilde{E}(\tau) = k_A \ | \ k_A = k_B \right] \leq \frac{1}{2}+\frac{\delta}{2} + \negl(\secp).\]
        
    \end{itemize}

    Note that $(1,0)$-BA is the fully-secure notion, which we will just call bit agreement, or sometimes key agreement.

Finally, (weak) OT is an interactive protocol that takes place between a sender $S$ and receiver $R$, denoted as \[(b,r),(r_0,r_1) \gets \langle R(1^\secp),S(1^\secp)\rangle,\] where $(b,r) \in \{0,1\}^2$ is the output of $R$ and $(r_0,r_1) \in \{0,1\}^2$ is the output of $S$. We then define a $(\epsilon,\delta,\gamma)$-OT as satisfying the following properties.
    \begin{itemize}
        \item \textbf{OT Correctness}: It holds that \[\Pr[r = r_b : \begin{array}{r} (b,r),(r_0,r_1) \gets \langle R(1^\secp),S(1^\secp)\rangle \\ \end{array}] \geq \frac{1}{2} + \frac{\epsilon}{2} - \negl(\secp).\]
        \item \textbf{OT Receiver security}: For any QPT adversarial sender $\widetilde{S}$,
       \[\Pr[\widetilde{b} = b : (b,r),\widetilde{b} \gets \langle R(1^\secp),\widetilde{S}\rangle] \leq \frac{1}{2} + \frac{\gamma}{2} + \negl(\secp).\]
        \item \textbf{OT Sender security}: For any QPT adversarial receiver $\widetilde{R}$,
        \[\Pr[\widetilde{r} = r_0 \oplus r_1 : \begin{array}{r} \widetilde{r},(r_0,r_1) \gets \langle \widetilde{R},S(1^\secp)\rangle\end{array}] \leq \frac{1}{2} + \frac{\delta}{2} + \negl(\secp).\]
    \end{itemize}

    Note that $(1,0,0)$-OT is the fully-secure notion, which we just refer to as OT.

\paragraph{Main results: BA and OT from ToNC.} We are now in a position to state our main results, establishing that strong forms of cryptography are inherent to tests of non-commutation. In this section, whenever we say ToNC, we always mean \emph{normal-form} ToNC as defined above, where there is always exactly one correct answer bit (as discussed above).

Before stating the theorems, we note that $(\epsilon,\delta)$-ToNC is impossible whenever $\delta < \epsilon/2$. This is because, given any (potentially non-commuting) strategy that attains advantage $\epsilon$, one can derive a commuting strategy with advantage $\epsilon/2$, by sampling a random bit $c'$, running the $\epsilon$-good strategy when $c = c'$, and outputting a uniformly random bit when $c = 1-c'$. 

We first consider an implication to (classical-communication) key agreement, which holds unconditionally.

\begin{theorem}\label{thm:informal-KA}
    $(\epsilon,\delta)$-ToNC implies classical-communication key agreement for any constants $\epsilon,\delta$ such that
    \[\delta < \frac{5\epsilon-1}{4}.\] Moreover, $(\epsilon,\delta)$-ToNC with \emph{robust completeness}\footnote{Robust completeness (\cref{def:robust-completeness}) demands that the honest strategy succeeds with the same probability up to negligible factors for almost any sampling of the preamble, which is a natural property that is satisfied by schemes such as the compiled CHSH and magic square games.} implies classical-communication key agreement for \emph{any} non-trivial parameters, i.e., whenever $\delta < \epsilon$.
\end{theorem}

Due to the nature of our techniques, we establish the above in the presence of adversaries with non-uniform \emph{classical} advice, but not necessarily non-uniform quantum advice.

We remark that every currently known qubit-test construction falls within the scope of our key-agreement implications: the binary response constructions are handled by our main ToNC framework and \cref{thm:informal-KA} above, and in Appendix~\ref{sec:generalized-tonc} we formulate a generalized notion of ToNC with multi-bit response that abstracts the construction
of~\cite{BCMVV}, and show that it also implies full-fledged key agreement.

Next, we build (classical-communication) OT. Here, we make use of (post-quantum) one-way functions in addition to the ToNC.

\begin{theorem}\label{thm:informal-OT}
    $(\epsilon,\delta)$-ToNC plus one-way functions implies classical-communication oblivious transfer for any constants $\epsilon,\delta$ such that
    \[\delta < \frac{\epsilon^2+\epsilon}{2}.\]
\end{theorem}

Here, we actually establish the above in the presence of adversaries with non-uniform \emph{quantum} advice, but not necessarily non-uniform classical advice.\footnote{Note that the classical advice case does not necessarily follow as a special case of the quantum advice case, and indeed our particular reduction requires the use of quantum advice.}

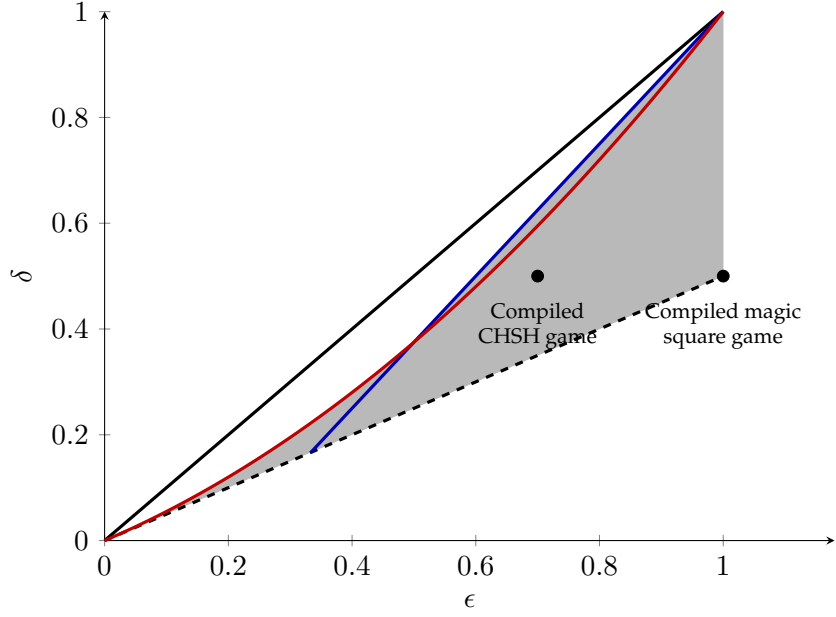
\begin{figure}[t]
\centering
\begin{tikzpicture}
\begin{axis}[
    width=0.68\textwidth,
    height=0.52\textwidth,
    xmin=0, xmax=1.18,
    ymin=0, ymax=1,
    axis lines=left,
    xlabel={$\epsilon$},
    ylabel={$\delta$},
    xtick={0,0.2,0.4,0.6,0.8,1},
    ytick={0,0.2,0.4,0.6,0.8,1},
    domain=0:1,
    samples=300,
    clip=false,
]

% Invisible paths
\addplot[name path=upper, draw=none] {x};          % delta = epsilon
\addplot[name path=lower, draw=none] {x/2};        % delta = epsilon/2
\addplot[name path=ka,    draw=none, domain=1/3:1] {(5*x - 1)/4};
\addplot[name path=ot,    draw=none, domain=0:1]   {(x^2 + x)/2};

% Dark gray KA region, restricted to plausible ToNC region
\addplot[
    fill=gray!55,
    draw=none,
    forget plot
] fill between[
    of=ka and lower,
    soft clip={domain=1/3:1}
];

% Dark gray OT region, restricted to plausible ToNC region
\addplot[
    fill=gray!55,
    draw=none,
    forget plot
] fill between[
    of=ot and lower,
    soft clip={domain=0:1}
];

% Boundary curves
\addplot[very thick, black] {x};
\addplot[very thick, black, dashed] {x/2};
\addplot[very thick, blue!70!black, domain=1/3:1] {(5*x - 1)/4};
\addplot[very thick, red!75!black] {(x^2 + x)/2};

% Compiled CHSH game point
\addplot[
    only marks,
    mark=*,
    mark size=2.2pt,
    black
] coordinates {(0.7,0.5)};
\node[
    anchor=north,
    font=\scriptsize,
    inner sep=1pt,
    align=center
] at (axis cs:0.7,0.48) {Compiled\\CHSH game};

% Compiled magic square game point.
% The bound is delta = 1 - 1/4320. Its separation from 1 is
% exaggerated here because the exact gap is invisible at this scale.
\addplot[
    only marks,
    mark=*,
    mark size=2.2pt,
    black
] coordinates {(1,0.96)};
\node[
    anchor=west,
    font=\scriptsize,
    align=center
] at (axis cs:1.03,0.96) {Compiled magic\\square game};

\end{axis}
\end{tikzpicture}
\caption{
Visualization of the parameter regimes captured by \Cref{thm:informal-KA} and \Cref{thm:informal-OT}. The solid black line is $\delta=\epsilon$, and the
dashed black line is $\delta=\epsilon/2$, so the region between is the region of parameters where (normal form) ToNC can potentially exist. The dark gray shade denotes the sub-region where we show that ToNC implies KA or OT (additionally assuming one-way functions for the latter). The blue curve is the key agreement threshold $\delta = (5\epsilon-1)/4$, and the red curve is the oblivious transfer
threshold $\delta = (\epsilon^2+\epsilon)/2$. It remains open to show that ToNC in the white region below the solid black line implies KA or OT. However, we note that ToNC with \emph{robust completeness} and any $\epsilon / 2 \leq \delta < \epsilon$ implies KA, which is not reflected in the figure above. We also note that improved analysis may show that the compiled magic square game lives significantly lower in the figure.
}
\end{figure}

Checking parameters for our running example of the compiled CHSH game where $\epsilon \approx 0.7$ and $\delta = 0.5$, we see that \[0.5 < \frac{5(0.7)-1}{4} \approx 0.63, \quad \text{and} \quad 0.5 < \frac{0.7^2 + 0.7}{2} \approx 0.6.\] As another example, one can consider the compiled \emph{magic square} game \cite{Cui2026computational}. As a direct consequence of \cite[Theorem 6.7]{Cui2026computational}, an appropriately compiled magic square game\footnote{Sample Bob's question while still in the preamble. If it is in $\{2,4\}$, then continue to the challenge phase, setting the challenge bit based on whether the question is 2 or 4. If not, then request an answer during the preamble, and if it passes, run a ``dummy'' challenge phase with a publicly-known correct answer.} yields $(\epsilon,\delta)$-ToNC with $\epsilon = 1-\negl(\secp)$ and $\delta = 1-1/4320$.\footnote{We believe that the bound on $\delta$ can be significantly improved with non-black-box use of their techniques, but we leave a formalization of this to future work.} These constant parameters satisfy both
$\delta<(5\bar{\epsilon}-1)/4$ and
$\delta<(\bar{\epsilon}^2+\bar{\epsilon})/2$,
so both of our implications (key agreement and OT) apply.

%which fits our definition of $(\epsilon,\delta)$-ToNC with $\epsilon = 1$ and $\delta = 0.5$. 

%These are strictly better parameters than the compiled CHSH game and thus also fall in the range supported by our theorem statements.

\paragraph{Post-quantum hardness amplification.} Our main results are obtained in two steps: (1) construct weak forms of BA and OT from ToNC, and (2) amplify the hardness of the BA or OT. 

Even though our protocols utilize only classical communication, the honest parties may be quantum, and the adversaries of course may be quantum as well. Thus, we need \emph{post-quantum} hardness amplification for BA and OT. While BA and OT amplification enjoys a long history of study in the purely classical setting, to the best of our knowledge there is no prior work that considers post-quantum hardness amplification for BA or OT.

We first consider post-quantum BA amplification. Here, the adversarial Eve receives a classical input (the transcript of the protocol) and applies a quantum circuit in an attempt to guess the shared bit. In the purely classical setting, Holenstein \cite{Holenstein} gave a tight amplification result, establishing that $(\epsilon,\delta)$-BA implies fully-secure key agreement whenever $\delta < \frac{2\epsilon}{1+\epsilon}$. This is tight since $(\epsilon,\delta)$-BA exists unconditionally whenever $\delta \geq \frac{2\epsilon}{1+\epsilon}$ \cite{Holenstein}. 

This result was derived from a tight hard-core set theorem, improving on that of \cite{10.5555/795662.796290}, and stated roughly as follows. Consider any function $f$ and predicate $P$ such that for any polynomial-size classical circuit, $\Pr_{x \gets \{0,1\}^n}[C(f(x)) = P(x)] \leq 1-\frac{\delta}{2}$. Then there exists a set $S \subset \{0,1\}^n$ of density roughly $\delta$ such that $C$ has success probability very close to $1/2$ on $x \gets S$.

We give a post-quantum analogue of this tight hard-core set theorem. The fact that the predictor is now a quantum predictor requires modifying the proof of the original theorem, as well as slightly modifying the statement. We discuss these changes further in \cref{subsec:BA-amp}.

\begin{theorem}(Informal: Post-quantum hard-core measure)
    Let $\{D_\secp\}_{\secp\in \mathbb{N}}$  be any family of distributions over $(x,b)$ that can be sampled in quantum polynomial time, such that $\max_{(x,b)}\{D_\secp(x,b)\}= \negl(\secp).$ If for any polynomial-size quantum circuit $Q$,
    \[\Pr_{(x,b) \gets D_\secp}[Q(x) = b] \leq \frac{1}{2} + \frac{\delta}{2},\] then there exists a measure $M_\secp \preceq D_\secp$ of density roughly $1-\delta$ such that \[\Pr_{(x,b) \gets M_\secp}[Q(x) = b] \approx \frac{1}{2}.\]
\end{theorem}

By combining our post-quantum hard-core measure theorem with the techniques of \cite{Holenstein}, we derive tight post-quantum key agreement amplification (the fact that this bound is tight is shown in \cite{Holenstein}).

\begin{theorem}
    For constants $\epsilon,\delta$, post-quantum $(\epsilon,\delta)$-BA secure against adversaries with non-uniform classical advice, implies post-quantum $(1,0)$-BA if and only if \[\delta < \frac{2\epsilon}{1+\epsilon}.\]
\end{theorem}

Next, we consider post-quantum OT amplification. Here, additional complications arise in the post-quantum setting due to the fact that the adversary is \emph{interactive}, which we discuss in depth in \cref{subsec:OT-amp-overview}.

Our main technique here is a sequential repetition theorem from which our OT amplification results can be derived, which can be seen as a post-quantum analogue of the ``interactive XOR lemma'' from \cite[Lemma 3.1]{HR08}. Informally, the theorem can be stated as follows.

\begin{theorem}(Informal: Post-quantum sequential XOR lemma)
    Consider any interactive protocol between a (potentially quantum) adversary and a classical challenger who obtains a private bit $b$ at the end of interaction. Suppose that any QPT adversary has advantage at most $\delta$ in guessing $b$. Then, if we repeat the protocol $\ell$ times sequentially, no QPT adversary has advantage better than $\delta^\ell + \negl(\secp)$ in guessing $\bigoplus_{i \in [\ell]}b_i$, where $b_i$ is the challenger's (private) bit obtained in the $i$'th repetition.
\end{theorem}

Even given the recent progress on parallel repetition in the (post-)quantum setting \cite{10.1145/3618260.3649603,cryptoeprint:2025/1027}, the above theorem does not follow directly from prior work (see more discussion in \cref{subsec:related}). To prove it, we use Marriott-Watrous \cite{marriott2005} style rewinding arguments to establish a post-quantum analogue of Levin's isolation lemma \cite{levinIsolation85,GoldreichOnYao2011}.

Finally, we use our sequential XOR lemma to amplify weak OT to fully-secure OT. For technical reasons, the first half of this amplification process only applies to a specialized form of OT we define called ``committed-bit OT'' (\cref{def:committed-bit}), which we build from ToNC. This technique amplifies $(\epsilon,\delta,0)$ committed-bit OT to $(1,1/2,0)$ committed-bit OT for any $\delta < \epsilon^2$. The next part follows essentially directly from the sequential XOR lemma, and applies to any OT satisfying the standard game-based definition. That is, we prove the following theorem.

\begin{theorem}
    For any $\delta < 1$, post-quantum $(1,\delta,0)$-OT implies post-quantum $(1,0,0)$-OT.
\end{theorem}

\paragraph{Discussion.} Our work is motivated by the apparent gap between the type of cryptography known to be sufficient for qubit tests (and related forms of classical verification) and the type of cryptography sufficient for weaker primitives such as tests of quantumness.

Starting with the work of \cite{IR89}, cryptographers have studied an often elucidating characterization of cryptographic primitives into two groups: (1) those that exist unconditionally in the random oracle model (sometimes called ``minicrypt'' primitives), and (2) those that \emph{provably do not} exist unconditionally in the random oracle model (sometimes called ``cryptomania'' primitives). While the group (1) primitives can be heuristically instantiated with unstructured cryptography such as hash functions, group (2) primitives seem to require highly structured, e.g. algebraic, forms of cryptographic hardness.

While tests of quantumness are known to exist unconditionally in the (quantum) random oracle model \cite{10.1145/3658665}, and thus exist in group (1), we so far do not know such a result for qubit tests. Our work shows that (at least, in the parameter ranges covered by our theorem statements above) establishing qubit tests as a group (1) or ``minicrypt'' primitive would \emph{also} establish classical-communication key agreement and oblivious transfer as minicrypt primitives. This would be a highly surprising result, for the following reasons. First of all, even though the study of (classical-communication) key agreement goes back to the earliest days of modern cryptography \cite{diffie1976new}, there have been no candidates proposed from unstructured cryptography. This is partially explained by the result of \cite{IR89}, who showed that there is no black-box construction from one-way functions assuming that the honest parties and adversaries make only \emph{classical} queries to the one-way function. Extending this result to the setting of quantum queries remains an open question, though there has been recent progress towards this goal \cite{10.1007/978-3-031-15979-4_6,cryptoeprint:2025/639}.

In fact, current results suggest that qubit tests may require \emph{more} structure than even classical-communication key agreement. For example, consider that injective trapdoor functions, which are known to imply key agreement, are not known to imply qubit tests. Instead, current constructions of qubit tests utilize \emph{two-to-one} trapdoor functions, at the very least.

This situation can be partially explained by our result establishing that tests of non-commutation plus one-way functions imply oblivious transfer. Indeed, while injective trapdoor functions imply key agreement, they are black-box separated from oblivious transfer \cite{10.5555/795666.796557}.\footnote{Again, note that this result only considers honest parties and adversaries that make classical queries to the primitives, and the quantum-query case remains open} Thus, if one could build a test of non-commutation from injective trapdoor functions (which trivially also imply one-way functions), this would yield classical-communication OT from injective trapdoor functions, which would be considered a very surprising result.

We note that these implications hold with respect to the game-based notions of OT and key agreement which are explored in our work. They could potentially be extended to different notions as well, such as OT with alternative security definitions.

\subsection{Related work}\label{subsec:related}

\paragraph{Proofs of quantum memory \cite{hhan2025proofsquantummemory}.} Recently, \cite{hhan2025proofsquantummemory} introduced the notion of a proof of quantum memory (PoQM), which is a classical-communication protocol that occurs in two stages. After the first stage, there is a one-message challenge from the verifier that is answered by the prover. It has completeness $\alpha$ if there exists a prover that can cause the verifier to accept with probability $\alpha$, and it has soundness $\beta$ if any prover that maintains only \emph{classical} memory between stages can cause the verifier to accept with probability at most $\beta$.\footnote{They also consider a more general variant, where there is a number of qubits of memory that suffice for completeness, and a fewer number of qubits for which the soundness bound holds.}

PoQM is, in general, a less structured primitive than ToNC. For example, it is not hard to see that (normal form) $(\epsilon,\delta)$-ToNC implies PoQM with completeness $\frac{1}{2} + \frac{\epsilon}{2}$ and soundness $\frac{1}{2} + \frac{\delta}{2}$, as well as ``extraction probability'' $\frac{1}{2} + \frac{\epsilon}{2}$, which refers to the probability that the verifier can predict the prover's final message. Moreover, \cite{hhan2025proofsquantummemory} establish that PoQM with (1) extraction probability $1-1/\poly$ and, (2) a $1/\poly$ gap between completeness and soundness, implies (fully-secure) classical-communication key agreement. Thus, our result that $(1,\delta)$-ToNC implies KA for any $\delta < 1$ can also be achieved by an alternate route that goes through PoQM. 

However, the majority of our results do not follow from \cite{hhan2025proofsquantummemory}. This includes our constructions of KA from $(\epsilon,\delta)$-ToNC for any setting of $\epsilon < 1$ (covering important examples such as the compiled CHSH game), as well as all of our constructions of OT from ToNC (plus one-way functions), explaining the lack of approaches to building ToNC from weaker forms of asymmetric cryptography such as plain trapdoor functions.

\paragraph{Oblivious state preparation \cite{BK25}.} Recently, \cite{BK25} introduced the notion of oblivious state preparation (OSP), which is a classical-communication protocol that allows a client to prepare a BB84 quantum state in the server's memory in such a manner that the server has negligible advantage in determining whether it is in the standard or Hadamard basis. They show that OSP implies both (classical-communication) key agreement and OT.\footnote{In fact, they require one-way functions as well to obtain the notion of OT that we consider in this paper.} 

OSP is a stronger primitive than ToNC: While OSP is known to imply $(2\cos^2(\pi/8)-1,\,0.5)$-ToNC \cite{BK25} by compiling the CHSH game, the reverse direction is not known. Whether we can obtain (some form of) OSP from ToNC is an interesting question, though appears somewhat difficult.

\paragraph{Foundations of interactive quantum advantage \cite{TZ25}.} In another recent work, \cite{TZ25} study the cryptographic implications of verifiable quantum advantage. Their result most relevant to this work establishes that there is no construction of constant-round interactive quantum advantage from the black-box use of indistinguishability obfuscation and one-way permutations. Since constant-round $(\epsilon,\delta)$-ToNC immediately implies constant-round quantum advantage for any $\delta < \epsilon$, this establishes the same for any non-trivial constant-round ToNC. Our results differ from theirs in the following ways:
\begin{itemize}
    \item Our cryptographic implications of ToNC apply to ToNC with any number of rounds, and are not limited to black-box constructions of ToNC.
    \item Our results give evidence that ToNC cannot be constructed from random oracles, while \cite{TZ25}'s results do not rule out the possibility that ToNC can be constructed from collision-resistant hashing (which exists in the random oracle model).
\end{itemize}

\paragraph{(Post-)quantum parallel repetition \cite{10.1145/3618260.3649603,cryptoeprint:2025/1027}.} A couple of recent works have ushered in the study of hardness amplification in the post-quantum and quantum settings, in particular when protocols are \emph{parallel} repeated. However, as we discuss further in \cref{sec:tech-overview}, these results do not have implications to post-quantum key agreement and oblivious transfer amplification, for the following reasons. \cite{cryptoeprint:2025/1027} focus primarily on \emph{public-coin arguments}, where the job of the adversary is to \emph{convince the verifier to accept}, as opposed to guess a hidden bit known to a private-coin challenger. While they do have a private-coin result, it is limited to three-message protocols. \cite{10.1145/3618260.3649603} focus on \emph{three-message, quantum-communication} protocols (they also have a round-compression theorem for quantum-communication arguments, but this does not apply to classical-communication protocols). While they do derive a quantum analogue of Yao's XOR lemma, it does not apply to the classically-interactive private-coin setting that we require for OT amplification.

\subsection{Open problems}

Our work motivates several directions for future research, which we outline here.

\begin{enumerate}
    \item Can we tighten our implications from ToNC to key agreement and OT? For example, there exists a region of $\delta < \epsilon$ for which we don't know how to derive OT from $(\epsilon,\delta)$-ToNC. Moreover, can we hope to remove the reliance on one-way functions and establish OT from ToNC unconditionally?

    \item As mentioned earlier, our BA amplification techniques only apply to adversaries with non-uniform classical advice, while our OT amplification techniques only apply to adversaries with non-uniform quantum advice. Can we give post-quantum BA and OT amplification for \emph{both} types of advice? Can we give post-quantum BA and OT amplification in the \emph{uniform} setting?

    \item Can we continue to improve our understanding of the relationships between PoQM, ToNC, OSP, and OT? For example, does some variant of ToNC imply OSP? Does OT imply ToNC, or is there a separation? One current gap between what is known about OT vs. ToNC is the fact that trapdoor permutations are known to imply OT, but not known to imply ToNC.
    
    \item In Appendix~\ref{sec:generalized-tonc}, we introduce a structured
multi-bit-answer generalization of ToNC that captures the original qubit test
of~\cite{BCMVV}, and show that it implies key agreement. Does this generalized
notion also imply oblivious transfer? More generally, can its structural
conditions be weakened while retaining a key-agreement implication?

    \item As we are the first to study post-quantum hardness amplification for key agreement and OT, several questions remain. For example, can we derive more general amplification results for OT? Such results would presumably extend the parameter regime in which our construction of OT from ToNC remains valid. Additionally, are there other applications of our post-quantum hard-core measure theorem or post-quantum interactive XOR lemma?
\end{enumerate}
\section{Technical overview}\label{sec:tech-overview}

In this overview, we will first cover our constructions of weak bit agreement (BA) and weak oblivious transfer (OT) from tests of non-commutation (ToNC), and proceed to discuss our post-quantum hardness amplification techniques. 

\paragraph{Normal-form ToNC.} Recall from the previous section that we work with a notion that we call ``normal-form'' ToNC. In such a ToNC, the interaction between quantum prover and classical verifier always reaches a final stage where the prover receives a challenge bit $c$ and the verifier obtains a (private) answer bit $a^*$. The prover then applies one of two binary observables $P_0$ or $P_1$ to its state $\ket{\psi}$, depending on the choice of $c$, to obtain an answer $a$. We say that the prover has $\epsilon$ advantage if $\Pr[a = a^*] = \frac{1}{2} + \frac{\epsilon}{2}$. 

In \Cref{subsec:normal-form}, we show that a more-general variant of ToNC (where the verifier may sometimes accept both or neither possible answers) can be compiled into normal-form ToNC with the same completeness-soundness gap. Hence, in the remainder of this section, we will focus on building cryptography from normal-form ToNC.

\subsection{Weak BA from ToNC}

It turns out that (normal-form) ToNC can essentially already be seen as a weak bit agreement between prover and verifier. Indeed, suppose that the prover computes $a$ but keeps it private. Then we have that the prover and verifier agree on a bit $a = a^*$ with advantage $\epsilon$.

Now consider an eavesdropper $E$ who may observe the (classical) transcript $\tau$ of interaction between the prover and verifier, and suppose that they are able to guess the value of $a^*$ with advantage $\gamma$ (meaning probability $\frac{1}{2} + \frac{\gamma}{2}$). Since the transcript $\tau$ is disjoint from the prover's state $\ket{\psi}$, such an $E$ can be used to derive a \emph{commuting} strategy with advantage $\frac{\epsilon + \gamma}{2}$. In particular, consider a prover that, prior to seeing the final challenge bit $c$, flips a random coin $b \gets \{0,1\}$. If $b = c$, it runs the honest prover strategy on its final state $\ket{\psi}$ while if $b = 1-c$, it runs $E$ on the classical transcript of the protocol. Such a prover is commuting by inspection and has advantage exactly $\frac{\epsilon + \gamma}{2}$.

Now, suppose we start with a $(\epsilon,\delta)$-ToNC according to the definition provided in the previous section (and formally in \cref{def:ToNC} and \cref{def:normal-form}). Then we obtain a bit agreement protocol with correctness advantage $\epsilon$, and where any $E$ has advantage at most $2\delta - \epsilon$ in guessing the verifier's bit (call her Alice). Due to technical reasons in the ensuing bit agreement amplification, we care in particular about $E$'s advantage \emph{conditioned} on Alice and Bob outputting the same bit. Thus, the best \emph{conditional} advantage bound we can derive comes out to \[\delta' = \frac{1+4\delta-3\epsilon}{1+\epsilon},\] yielding a $(\epsilon,\delta')$-BA protocol for such $\delta'$ according to the definition given in the previous section (and formally in \cref{def:bit-agreement}). The formal protocol and analysis are given in \cref{subsec:WBA}.

\subsection{Weak OT from ToNC}

Next, we discuss our construction of weak OT from ToNC. Building OT presents several complications that did not arise in the construction of BA. For instance, the adversary \emph{actively} participates (potentially maliciously) in the protocol. Moreover, to prepare for our OT amplification step, we will actually build a stronger notion of OT than what was presented in the previous section, which we call committed-bit OT. 

\paragraph{Skeleton protocol.} In this overview, we give a sense of the core of our idea by presenting a bare-bones OT protocol that works only if several conditions are satisfied. In the body, we utilize techniques such as cryptographic commitments (which introduces an additional assumption of one-way functions), cut-and-choose, and random-termination to ``compile'' the skeleton protocol into a full-fledged committed-bit OT. However, the main intuition can be appreciated with the following skeleton, which uses a (normal-form) $(\epsilon,\delta)$-ToNC.

\begin{itemize}
    \item Let the OT sender be the verifier and the OT receiver be the prover in two runs of the ToNC. That is, run \[\ket{\psi_0},c_0,a^*_0 \gets \langle P_\prep(1^\secp),V_\prep(1^\secp)\rangle,\quad \ket{\psi_1},c_1,a^*_1 \gets \langle P_\prep(1^\secp),V_\prep(1^\secp)\rangle.\] However, the sender (verifier) does not send $(c_0,c_1)$ to the receiver (prover) yet.
    \item The receiver samples $c_{R,0},c_{R,1} \gets \{0,1\}$ and computes answers \[a_{R,0} \gets P_{c_{R,0}}(\ket{\psi_0}), \quad a_{R,1} \gets P_{c_{R,1}}(\ket{\psi_1}).\]
    %\itaycomment{I think the condition below should be $c_{R,0}\oplus c_{R,1}\neq c_0\oplus c_1$.}
    \item The receiver sends $c_{R,0} \oplus c_{R,1}$ and the protocol only continues if $c_{R,0} \oplus c_{R,1} \neq c_0 \oplus c_1$. In this case, the sender sends $c_0,c_1$ to the receiver and outputs $r_0 = a_0^*$ and $r_1 = a_1^*$ as its OT bits.
    \item The receiver sets $b$ to be the unique bit such that $c_{R,b} = c_b$, and sets $r = a_{R,b}$.
\end{itemize}

First, note that if the receiver uses the honest prover strategy, then $r = r_b$ with advantage $\epsilon$, as in this case the challenge $c_{R,b}$ it used to produce $a_{R,b}$ was the same as the verifier's challenge $c_b$. 

Next, note that receiver security is optimal: The sender learns nothing about the receiver's bit $b$ since it learns $c_{R,0} \oplus c_{R,1}$ but neither bit individually.

Finally, it remains to prove a bound on the advantage that the (potentially adversarial) receiver has in guessing $r_{1-b}$.

\paragraph{Enforcing good behavior.} We say that an adversarial receiver is ``well-behaved'' if (1) its advantage in guessing $a_{R,b} = a^*_b$ is indeed (roughly) $\epsilon$, (2) $c_{R,0},c_{R,1}$ are chosen uniformly at random, and (3) it honestly reports the value of $c_{R,0} \oplus c_{R,1}$. In the full protocol, we enforce that the receiver is well-behaved by using a coin-flipping protocol to sample $c_{R,0},c_{R,1}$, having the receiver commit to its challenges $c_{R,0},c_{R,1},c_{R,0} \oplus c_{R,1}$ and answers $a_{R,0},a_{R,1}$, and running a cut-and-choose plus random-termination strategy that repeats the protocol and requires them to open all but one of these sets of commitments.

\paragraph{Reduction to ToNC.} Now, we argue that any well-behaved receiver with advantage $\gamma$ in guessing $r_{1-b}$ implies a \emph{commuting} prover strategy in the ToNC with advantage roughly $\frac{\epsilon + \gamma}{2}$. To see this, consider a prover $\widetilde{P}$ that utilizes the well-behaved receiver as follows. $\widetilde{P}$ simulates the receiver's interaction with the sender in two ToNC executions except that it performs the verifier itself in one of them, and forwards messages to its own verifier in the other. Which one is forwarded to the ``real'' verifier will be chosen uniformly at random, but for the sake of discussion let's say protocol 0 is forwarded and protocol 1 is simulated in $\widetilde{P}$'s head.

%\itaycomment{Here too, I think the condition should be $c_{R,0}\oplus c_{R,1}\neq c_0\oplus c_1$.}

Condition on $c_{R,0} \oplus c_{R,1} \neq c_0 \oplus c_1$ and let $b$ be such that $c_{R,b} = c_b$. If $b = 0$, then we know the receiver's answer $a_{R,b}$ will yield advantage $\epsilon$ in predicting the verifier's bit $a^*$, so $\widetilde{P}$ will return $a_{R,b}$ to its verifier. Otherwise, if $b = 1$, then we know that the receiver's guess for $r_{b-1} = r_0$ has advantage $\gamma$ in predicting the verifier's bit $a^*$, so $\widetilde{P}$ will return its guess for $r_{b-1}$ instead. Since $c_{R,0},c_{R,1}$ are random, and we randomize the choice of which of the two protocols are forwarded to the verifier, both cases happen with probability exactly 1/2, yielding an overall advantage of $\frac{\epsilon+\gamma}{2}$. By soundness of the ToNC, we have that $\frac{\epsilon + \gamma}{2} < \delta$, and thus $\gamma < 2\delta - \epsilon$. 

\paragraph{Using non-uniformity.} This is the main idea that underlies our result that $(\epsilon,\delta)$-ToNC implies $(\epsilon,2\delta - \epsilon,0)$ (committed-bit) OT. However, before moving on, we point out one complication that arises in the above strategy. On the one hand, $\widetilde{P}$ actually has to \emph{learn} the values $c_{R,0},c_{R,1},a_{R,0},a_{R,1}$ used by the adversarial receiver in the OT protocol, but on the other hand, these cannot be revealed in the transcript as they may leak the receiver's choice bit. To remedy this, we have the receiver \emph{commit} to these values and then have $\widetilde{P}$ \emph{extract} them in the reduction. As we use only one-way functions for the commitment, this extraction procedure is not efficient. However, it turns out that we can push all of the inefficiency required to \emph{before} the beginning of the ToNC protocol, and have $\widetilde{P}$ make use of non-uniformity to perform an inefficient operation before the beginning of the protocol. Thus, our final result is actually that $(\epsilon,\delta)$-ToNC secure against adversaries with non-uniform (potentially quantum) advice implies  $(\epsilon,2\delta - \epsilon,0)$ (committed-bit) OT. The formal protocol and analysis are given in \cref{subsec:WOT}.

\subsection{BA amplification}\label{subsec:BA-amp}

In the remainder of the paper, we prove amplification theorems for BA and OT in the post-quantum setting. That is, we address the setting where adversaries and the communicating parties are quantum (rather than classical) polynomial-time machines. Due to structural differences between BA and OT, we were compelled to take different amplification routes for each. We begin by describing at a high level how we amplify weak BA into full-fledged key agreement, presented in \cref{subsec:hardcore} and \cref{subsec:KA-amp}.

\paragraph{The classical setting.} Our starting point is Holenstein's tight classical amplification theorem~\cite{Holenstein}, which characterizes exactly when weakly-secure BA can be amplified to fully-secure key agreement. We briefly recall the classical approach and then explain the modifications needed in our (post-quantum) setting.

\paragraph{Secure random variables.} Following
Holenstein, we consider the notion of an \emph{information-theoretically} secure weak BA, where a single execution produces correlated bits
$(X,Y)$ for Alice and Bob, together with side information $Z$ held by Eve -- namely, the transcript.
One says that $(X,Y,Z)$ is an $(\epsilon,\delta)$-\emph{secure random variable} if (i) $X$ and $Y$
are unbiased, (ii) $\Pr[X=Y]\ge \tfrac12+\tfrac\epsilon2$, and (iii) there exists an event $E$ implying $X=Y$ such that
$\Pr[E\mid X=Y]\geq\delta$ and, conditioned on $E$, the shared
bit $X$ is uniform and independent of $Z$. Definition \ref{def:secure-random-variable} formalizes this notion.

\paragraph{An information-theoretic amplification protocol.} Given $\ell$ independent $(\epsilon,\delta)$-secure random variables, Holenstein gives an efficient
protocol that compiles them into a perfect (up to negligible factors) secure random variable -- and hence fully-secure key agreement, assuming the following condition on $(\epsilon, \delta)$ holds:
\[\delta>\frac{1-\epsilon}{1+\epsilon}.\]

It is now left to show that applying such a compilation protocol to multiple independent \emph{computationally-secure} weak BA instances results in a fully-secure BA instance. To do so, Holenstein proves a variant of Impagliazzo's hard-core set lemma against non-uniform classical circuits, with improved parameters~\cite[Lemma~2.1]{Holenstein}. Let us state it in terms of bit agreement for the sake of discussion.

\paragraph{Hard-core set lemma.} Fix some arbitrarily small constant $\gamma>0$, and suppose that any adversarial circuit of size $s$ given the communication transcript can succeed in guessing the output bit (conditioned on the parties agreeing on the same bit) with probability at most $1-\delta/2$. Then, there exists a set $S$ of transcripts whose total mass depends on $\delta,s$, such that when sampling a weak BA instance from $S$, all circuits of size somewhat smaller than $s$ succeed in guessing the agreed bit (conditioned on agreement) with probability at most $\frac{1+\gamma}{2}$.

\paragraph{Amplification in the (classical) computational security setting.} The amplified protocol works exactly as in the information-theoretic case. 

\begin{enumerate}
    \item The parties perform $\ell(\secp)=\poly(\secp)$ independent instances of the weak computationally-secure BA protocol, producing pairs of bits $(x_i,y_i)_{i=1,\ldots,\ell}$.
    \item They run the compilation protocol on $(x_i,y_i)_{i=1,\ldots,\ell}$, producing final output bits $k_A,k_B$. 
\end{enumerate}

Using the hard-core set lemma, one proves that any potential adversary attempting to break the amplified protocol over weak computationally-secure BA instances, cannot do significantly better than it would on a series of information theoretic $(\epsilon,\delta)$-secure random variables. 

Indeed, given a hard-core set $S$, define the event $E$ to be "the transcript corresponds to $x\in S$ and the parties output the same bit". Conditioned on $E$, the weak BA transcript $z$ carries essentially no efficiently-computable information about the shared bit $x=y$, so for the computationally-bounded adversary, the triplet $(x,y,z)$ is indistinguishable from an information-theoretically secure $(\epsilon,\delta_{\sec})$-secure random variable, with $\delta_{\sec}\approx \Pr[E\mid X=Y]$. The proof is completed using an appropriate hybrid argument.

%A hybrid argument shows that any adversary that predicts the final key with noticeable advantage in the amplified protocol based on weak BA instances, would yield a circuit that predicts one of the underlying $\ell$ weak BA output bits noticeably better than allowed, contradicting weak BA security.

\paragraph{Moving to the quantum setting.} It turns out that given an appropriate post-quantum analogue of the hard-core set lemma, the BA amplification works using an argument similar to that described above. So we will focus on the hard-core set lemma itself. In particular, our goal is to re-prove it in the setting where the adversary can apply a quantum circuit rather than a classical circuit.

The approach laid out by \cite{Holenstein,HolensteinPuzzles} for proving the hard-core set lemma in the classical setting, consists of two main steps.

\begin{itemize}
\item \textbf{A hard set for each predictor.}
Assume that no efficient adversary can guess the output bit of the protocol from the communication transcript
with probability better than $1-\delta/2$, conditioned on agreement.
Then for any fixed predictor $Q$ and any slack $\gamma\in(0,\delta)$, there is a set of transcripts of mass at least $\delta-\gamma$ on which $Q$ has
essentially no advantage, i.e.\ its success is at most $1/2+\gamma/2$.

\item \textbf{Per-set predictors $\Rightarrow$ one universal predictor.}
Now suppose that for every large set there exists
\emph{some} adversary of size $s$ that computes the output bit with noticeable advantage $\gamma$ on that measure.
Then one can ``fix'' the predictor: There exists a \emph{single} predictor (with a moderate size blowup)
that achieves advantage better than $\gamma/2$ \emph{simultaneously on all} sets of mass at least $\delta$. This contradicts the previous item.
\end{itemize}

Together, these two components imply the existence of a hard-core set on which \emph{all} predictors of size $s$ have success probability close to $1/2$. The key component in the proof of the second step is von Neumann's minimax theorem, which translates to the setting of quantum circuits without requiring significant adaptations. In what follows, we describe how we adapt the proof of the first step in some more detail.

\paragraph{A hard \emph{measure} for any quantum predictor.} Fix a deterministic predicate $P:\{0,1\}^n\rightarrow \bits$ and any randomized classical circuit  $C(x,b)$ that is intended to output $1$ when $b=P(x)$ and $0$ otherwise. Theorem~3 in \cite{HolensteinPuzzles} (which is similar in spirit to Lemma 2.4 from \cite{Holenstein}) shows that one can
\emph{decompose} the input space relative to this particular $C$ into a large
``hard'' subset of the domain $S\subseteq \bits^n$ of probability mass $\delta$ and its complement. That is, on $S$, $C$'s acceptance probability is
essentially \emph{insensitive} to whether $b=P(x)$ or $b=1-P(x)$, meaning that flipping
$P$ on $S$ is nearly indistinguishable to $C$. Moreover,
one can convert $C$ into an explicit predictor $Q$ that recovers $P(x)$ with
an overall success at least $1-\delta/2$ on the entire domain. 

In the classical hard-core set proof, the first component is a statement of the form:
\emph{given a fixed predictor $A$, there exists a large region on which $A$ has only tiny advantage.}
The above decomposition can be used as an alternative route to such a statement, as follows.
Define a distinguisher $C_A(x,b)$ that outputs $1$ iff $A(x)=b$. Applying the decomposition to $C_A$
yields a set $S$ where $C_A$ cannot reliably distinguish $b=P(x)$ from $b=1-P(x)$, which exactly means
that on $x\leftarrow S$, the predictor $A$ has success close to $1/2$.
Moreover, under the \emph{assumed global hardness} parameter $\delta$ (namely, no predictor can achieve
success $\ge 1-\delta/2$ overall), the theorem's predictor $Q$ forces the extracted set $S$
to have density at least about $\delta$ -- otherwise $Q$ would contradict hardness.

Our Theorem \ref{thm:Q3} is a generalization of Theorem~3 in \cite{HolensteinPuzzles}, which differs in several respects:
\begin{itemize}
    \item \textbf{Quantum circuits.} We allow the distinguisher $C$ to be a \emph{quantum} circuit, and thus the resulting predictor $Q$ is quantum.
    \item \textbf{Arbitrary input distributions.} Rather than working only with
    the uniform distribution over $\bits^n$, we formulate the statement with
    respect to an arbitrary distribution $\mathcal{D}$ over inputs. This enables applying the theorem to the case of bit agreement protocols, where the distribution over communication transcripts is typically not uniform. We note that in the classical case, one can often utilize the uniform distribution over the participant's random coins, thus avoiding the necessity of generalizing the distribution. However, this is not possible if the participants are quantum.
    \item \textbf{A hard measure rather than a hard set.} As a consequence of allowing a quantum distinguisher $C$, we construct a quantum predictor which maps a pair $(x,P(x))$ \emph{probabilistically} to a bit. Since the predictor is non-deterministic, it does not immediately give rise to a hard subset. Instead, its prediction statistics constitute a 
    \emph{sub-distribution measure}  of $\mathcal{D}$
    that designates the portion of the domain on which $C$ is
    insensitive to flipping $P$. The measure is explicitly constructed in Corollary \ref{cor:soft-hardcore-random-labels}.
\end{itemize}

The measure vs. set discrepancy comes from the fact that the proof in \cite{HolensteinPuzzles} relies on derandomizing the distinguisher $C$ through fixing a set $R$ of random seeds. Evaluations of $C(x,r)$ on any $r\in R$ are used to compute a criterion which determines whether $x\in S$---thus $S$ becomes a deterministic set. This perspective breaks down in the quantum setting, where the distinguisher's behavior is inherently probabilistic and cannot be derandomized in the same manner. Instead, the hard region is represented by a weight function $w(x) \in [0,1]$ derived from the prediction statistics of the quantum distinguisher, yielding a sub-distribution $M \preceq D$. Consequently, the remainder of the amplification argument is reformulated so that Holenstein's hard-core-set reasoning continues to hold when the hard object is a weighted measure rather than a deterministic subset.\footnote{As mentioned before, \cite{Holenstein}'s proof also relied on a randomness fixing argument, though we found \cite{HolensteinPuzzles}'s approach easier to adapt to the quantum setting while maintaining optimal parameters.} 

A further subtlety is that the weak-BA distribution is not necessarily a deterministic predicate. After conditioning on the two parties agreeing on the same bit, we obtain a distribution over transcript/bit pairs $(Z,X)$. To apply a predicate-style hard-core theorem, we lift the domain to $(Z,b_0)$, where $b_0$ is the agreed-upon bit, and define the deterministic predicate $P(Z,b_0)=b_0$. Given a transcript predictor $C(Z)$, we form a distinguisher $C((Z,b_0),b)$ that ignores $b_0$ and checks whether $C(Z)=b$. The obliviousness of the predictor returned by our decomposition theorem is essential for using weak BA security to prove a hard-core measure: if the hard-core measure has too little mass, the resulting predictor on the lifted domain collapses back to a predictor of $X$ from $Z$, contradicting weak-BA security.

Lemma \ref{lem:minimax-fix} establishes the minimax fixing step needed for the hard-core measure theorem, which appears as Theorem \ref{thm:hardcore-Q3}. Theorem \ref{thm:wba-to-ka} establishes the generic amplification statement with a hybrid argument, and Theorem \ref{thm:ba-to-ka-params} states that the amplifiable parameter range in the quantum setting is identical to Holenstein's tight result in the classical setting.

\subsection{OT amplification}\label{subsec:OT-amp-overview}

Finally, we discuss our approach to post-quantum OT amplification, which requires a new set of techniques, presented in \cref{subsec:XOR-lemma} and \cref{subsec:OT-amp}.

\paragraph{New challenges.} The fundamental differences with BA amplification stem from the fact that the adversary participates in the protocol itself. That is, one must consider the ``input'' of the weakly hard distinguishing problem to be the \emph{quantum state} produced at the end of the protocol, as opposed to just the classical transcript. Moreover, as we allow the adversary to be malicious,\footnote{One might wonder about the possibility of post-quantum \emph{semi-honest} OT amplification, where both the honest and adversarial participants are quantum but the communication is classical. This appears to be a bit tricky to reason about; in particular, the accepted model of quantum semi-honest adversaries allow them to \emph{purify} their actions in the protocol \cite{10.1007/978-3-642-14623-7_37}, resulting in an interaction that now involves quantum communication. We believe that several of the non-trivialities involved in the setting of malicious post-quantum amplification would also arise in any reasonable semi-honest quantum model, and thus we focus on the malicious case.} it may influence the generation of this (transcript, state) pair.

While hard-core set theorems have been applied successfully to obtain (semi-honest)\footnote{In the classical setting, semi-honest amplification is sufficient, since maliciously-secure OT follows generically from semi-honest OT \cite{10.1145/28395.28420,10.1145/62212.62215}. No such semi-honest to malicious compiler is known in the quantum setting.} OT amplification in the purely classical setting \cite{10.1007/978-3-540-72540-4_32}, the fact that we must consider \emph{quantum} inputs now prevents us from applying the techniques from last section.\footnote{One might wonder about the possibility of obtaining some analogue of a hard-core measure for distributions over quantum states, but we leave this possibility to future work (and indeed, are not convinced that reasonable formulations of this statement are possible).} Instead, we prove a novel post-quantum XOR lemma that applies to classical-communication interactive protocols, and put it to work to obtain our OT amplification results. 

\paragraph{Interactive XOR lemma.} The interactive XOR lemma can be motivated by considering the following desired form of OT amplification. Suppose we have a $(1,\delta,0)$-OT protocol (according to the definition given in \cref{subsec:results} and formally in \cref{def:standard-OT}) for some constant $\delta < 1$, and our goal is to obtain fully-secure OT, i.e. (1,0,0)-OT. A natural idea is to repeat the protocol $\secp$ times sequentially, instructing the sender to XOR together their bits across runs. That is, letting $\{r_{i,0},r_{i,1}\}_{i \in [\secp]}$ be the set of sender bits obtained from the $\secp$ repetitions, we know that any adversarial receiver has $\delta$ advantage in guessing each $r_{i,0} \oplus r_{i,1}$, and it is natural to conjecture that any adversarial receiver thus has $\delta^\secp = \negl(\secp)$ advantage in guessing $\bigoplus_{i \in [\secp]} r_{i,0} \oplus r_{i,1}$. Crucially, note that we are not asking that the receiver make a guess for each $r_{i,0} \oplus r_{i,1}$ directly after the $i$'th repetition of the protocol (in which case it would be trivial to show that all guesses are correct with probability at most $(\frac{1+\delta}{2})^\secp$), but are rather asking that they guess a single bit at the end of all $\secp$ repetitions.

A useful abstraction of this setting can be phrased as follows. Suppose we have a protocol $\Gamma$ that takes place between a (potentially quantum) adversary $A$ and a classical verifier $V$, who obtains an output $b \in \{0,1\}$ at the end of the protocol. Let $\delta$ be the maximum advantage that any $A$ can obtain in guessing the verifier's output $b$. Then, if we sequentially repeat the protocol twice, letting $b_1$ and $b_2$ denote the two verifier outputs, the maximum advantage that any $A$ can obtain in guessing $b_1 \oplus b_2$ is $\delta^2 + \negl(\secp)$. Note that in the full statement (\cref{thm:sequential-rep}), we allow for different protocols $\Gamma_1,\Gamma_2$ with different advantage bounds $\delta_1,\delta_2$.

In fact, this is precisely the post-quantum analogue of a classical claim due to Halevi and Rabin \cite[Lemma~3.1]{HR08}, which itself is an interactive analogue of Levin's Isolation Lemma \cite{levinIsolation85}, as stated in \cite{GoldreichOnYao2011}. Unfortunately, the classical proof makes heavy use of rewinding, which is non-trivial and sometimes impossible in the quantum setting. However, the core rewinding idea \emph{has} been adapted to the quantum setting in a recent work of \cite{10.1145/3618260.3649603}, though we note that their theorems do not directly apply to our setting as they are stated for three-message parallel-repeated quantum-communication protocols. We proceed to give an overview of our proof technique, which is somewhat related to \cite{10.1145/3618260.3649603}, but ultimately different and self-contained. 

\paragraph{Proof technique.} For convenience, let's move to the $\{+1,-1\}$ basis for classical bits, and suppose that the adversary's output is $b_A \in \{+1,-1\}$, while the two verifier outputs are $b_1,b_2 \in \{+1,-1\}$. Then, for contradiction, we assume the existence of an adversary $A$ such that \[\E[b_A \cdot b_1 \cdot b_2] > \delta^2.\] Rearranging, we have that \[\E[b_1 \cdot \E[b_A \cdot b_2 \ \  | \ \ket{\psi}]] > \delta^2 \quad \Rightarrow \quad \E\left[b_1 \cdot \E\left[\frac{b_A \cdot b_2}{\delta} \ \ \bigg| \  \ket{\psi}\right]\right] > \delta,\] where $\ket{\psi}$ is the intermediate state that $A$ has between the executions of $\Gamma_1$ and $\Gamma_2$. By security of $\Gamma_2$ we know that for any such $\ket{\psi}$, it holds that \[-\delta \leq \E[b_A \cdot b_2 \ \ | \ \ket{\psi}] \leq \delta.\] This suggests the following strategy for breaking security of the \emph{first} protocol:
\begin{enumerate}
    \item Run the first stage of $A$ until it finishes interacting with $\Gamma_1$ and produces state $\ket{\psi}$.
    \item Use $\ket{\psi}$ to obtain a real number $\alpha \in [-\delta,\delta]$ whose expectation is distributed identically to $\E[b_A \cdot b_2]$.
    \item Output $b_A\in\{1,-1\}$ sampled as a Bernoulli random variable which takes the value $1$ with probability $\frac{1}{2}(\frac{\alpha}{\delta}+1)$.
\end{enumerate}

The only step that is potentially problematic in the quantum setting is (2). Indeed, in the classical setting, this is accomplished by repeatedly simulating executions with $\Gamma_2$ and estimating the probability of success. Note that it does \emph{not} suffice to run a single execution, as this would output a value in $\{+1,-1\}$ which has the correct expectation but is not within the range $[-\delta,\delta]$ required to get step (3) to work.

Fortunately, Marriott-Watrous \cite{marriott2005} style rewinding applied to $\ket{\psi}$ allows us to implement step (2). Indeed, alternating the projector $\Pi_0$ corresponding to initializing the verifier for $\Gamma_2$ and the projector $\Pi_1$ corresponding to running the interaction and checking whether $b_A = b_2$ allows us to (approximately) sample from a distribution with the correct expectation \emph{and} the correct support. However, to prove that the distribution is supported on $[-\delta,\delta]$, we must actually assume that $\Gamma_2$ is secure against adversaries with non-uniform \emph{quantum} advice. In particular, this guarantees that the eigenvalues $p_j$ of
$\Pi_0\Pi_1\Pi_0$ restricted to $\operatorname{im}(\Pi_0)$ lie in
$[(1-\delta)/2,(1+\delta)/2]$.
The Marriott--Watrous procedure approximately samples these
eigenvalues. Using the corresponding correlation values
$2p_j-1\in[-\delta,\delta]$, we obtain an output
$\alpha\in[-\delta,\delta]$ (with high probability).

\paragraph{Completing OT amplification.} Recursively applying the interactive XOR lemma  allows us to prove our motivating example, that $(1,\delta,0)$-OT implies $(1,0,0)$-OT for any $\delta < 1$.\footnote{In fact this also establishes that $(1,0,\delta)$-OT implies $(1,0,0)$-OT by first applying a standard OT reversal.} While this result is already of inherent interest, it does not suffice to complete our desired implications from ToNC. In particular, if the ToNC does not have completeness 1 (e.g.\ the compiled CHSH game), then we obtain some $(\epsilon,\delta,0)$ for both $\epsilon < 1$ and $\delta < 1$.

To boost correctness of $(\epsilon,\delta,0)$-OT, we adapt an approach of \cite{10.1007/978-3-540-72540-4_32} establishing that for any $\delta < \epsilon^2$, \emph{classical semi-honest} $(\epsilon,\delta,0)$-OT implies \emph{classical semi-honest} $(1,\delta',0)$-OT for some $\delta' < 1$. The approach involves two steps:
\begin{enumerate}
    \item Boost security by applying a sequential XOR lemma (which also brings correctness down with it). We already have the tool for analyzing this.
    \item Boost correctness by repeating the protocol several times with the \emph{same} input bits from each party.
\end{enumerate}

In step (2), it is imperative for security that the receiver indeed uses the same choice bit $b$ in each execution. While this holds by definition in the classical semi-honest model, we have to somehow enforce that our quantum malicious receiver exhibits this behavior. To do so, we rely on another layer of cut-and-choose, utilizing the fact that our weak OT protocol is in fact a ``committed-bit'' OT. This means that after the protocol has completed, the receiver can \emph{open} the interaction to its choice bit $b$, convincing the sender that it indeed used this $b$ during the interaction. This allows us to implement step (2) securely, and we thus establish that $(\epsilon,\delta,0)$ committed-bit OT implies $(1,1/2,0)$-OT for any $\delta < \epsilon^2$.

\ifsubmission\else\section{Preliminaries}

Let $\secp$ be the security parameter. By default, we consider non-uniform families of quantum adversaries, but the \emph{type} of non-uniformity will vary. In particular, we will sometimes consider adversaries with non-uniform \emph{classical} advice, given by $\{Q_\secp\}_{\secp \in \bbN}$, and sometimes consider adversaries with non-uniform \emph{quantum} advice, given by $\{\rho_\secp,Q_\secp\}_{\secp \in \bbN}$. In both cases, we will often suppress the parameterization over $\secp$ and write just $Q$ or $\rho,Q$.

\begin{definition}[Sub-distribution measure]\label{def:measure}
Let $\Omega$ be a finite set. A \emph{sub-distribution} (equivalently, a \emph{sub-probability measure}) on $\Omega$
is a function $M:\Omega\to[0,1]$ such that its \emph{total mass}
\[
\mu(M)\;:=\;\sum_{\omega\in\Omega} M(\omega)
\]
satisfies $\mu(M)\le 1$. If $\mu(M)>0$, we write $\omega\leftarrow M$ to denote sampling from the normalized
distribution $M/\mu(M)$, i.e.,
\[
\Pr[\omega=\omega_0]\;=\;\frac{M(\omega_0)}{\mu(M)}\qquad(\omega_0\in\Omega).
\]
Given a distribution $D$ on $\Omega$, we write $M\preceq D$ if $M(\omega)\le D(\omega)$ for all $\omega\in\Omega$.
\end{definition}

\begin{theorem}[von Neumann's minimax theorem {\cite{vonNeumann1928}}]
Let $\mathcal{A}$ be a finite set, let $Y \subseteq \R^d$ be a nonempty compact convex set,
and let $u : \mathcal{A}\times Y \to \R$ be such that $u(a,\cdot)$ is affine 
on $Y$ for every $a\in\mathcal{A}$.
Write $\Delta(\mathcal{A})$ for the simplex of distributions over $\mathcal{A}$, and define
$\bar u(p,y) := \E_{a\sim p}[u(a,y)]$.
Then
\[
\min_{y\in Y}\ \max_{a\in\mathcal{A}} u(a,y)
\;=\;
\max_{p\in \Delta(\mathcal{A})}\ \min_{y\in Y} \bar u(p,y).
\]
In particular, there exists $p^\star\in\Delta(\mathcal{A})$ such that for all $y\in Y$,
$\bar u(p^\star,y)\ \ge\ \min_{y'\in Y}\max_{a\in\mathcal{A}}u(a,y')$.
\end{theorem}

\subsection{Cryptographic primitives}
\begin{definition}[Inefficiently-extractable commitment]\label{def:com}
    An inefficiently-extractable commitment between a classical committer and classical receiver consists of an interaction \[\state_\Com,\tau \gets \langle \Com(1^\secp,b),\Rec(1^\secp)\rangle,\] where $\state_\Com$ is the private state of the committer, and $\tau$ is the public transcript of interaction produced by the protocol, along with a PPT algorithm 
    \[\Ver(\state_\Com,\tau,b) \to \{\top, \bot\}.\] It should satisfy the following properties.

    \begin{itemize}
        \item \textbf{Correctness}: for any $b \in \{0,1\}$,
        \[\Pr\left[\Ver(\state_\Com,\tau,b) = \top : \begin{array}{r}\state_\Com,\tau \gets \langle \Com(1^\secp,b),\Rec(1^\secp)\rangle \end{array}\right] = 1-\negl(\secp).\]
        \item \textbf{Hiding}: For any QPT adversarial receiver $\rho,\widetilde{\Rec}$ outputting a bit $\widetilde{b}$,
        \begin{align*}
        &\bigg|\Pr\left[\widetilde{b} = 0 : \widetilde{b} \gets \langle \Com(1^\secp,0),\widetilde{\Rec}(\rho)\rangle\right] - \Pr\left[\widetilde{b} = 0 : \widetilde{b} \gets \langle \Com(1^\secp,1),\widetilde{\Rec}(\rho)\rangle\right]\bigg| = \negl(\secp).
        \end{align*}
        \item \textbf{Binding:} There exists a (potentially inefficient) function $\Ext$ such that for any unbounded adversarial committer $\widetilde{\Com}$,
    \begin{align*}
        \Pr\left[\Ver(\widetilde{\state}_\Com,\tau,1-b) = \top : \begin{array}{r} \widetilde{\state}_\Com,\tau \gets \langle \widetilde{\Com},\Rec(1^\secp)\rangle \\ b \coloneqq \Ext(\tau) \end{array}\right] = \negl(\secp).
    \end{align*}
    \end{itemize}

\end{definition}

\begin{remark}
    We remark that Naor's commitment \cite{10.5555/646754.705040} based on any (post-quantum) pseudo-random generator satisfies the above definition.
\end{remark}

\begin{definition}[Bit Agreement]\label{def:bit-agreement}
    An $(\epsilon,\delta)$-weak bit agreement (BA) protocol with classical communication consists of an interaction between polynomial-time (potentially quantum) parties $A$ and $B$, denoted by
    \[(k_A,k_B,\tau) \gets \langle A(1^\secp),B(1^\secp)\rangle ,\]

    where $k_A \in \{0,1\}$ is $A$'s output, $k_B \in \{0,1\}$ is $B$'s output, and $\tau$ denotes the (classical) transcript of interaction that occurs between $A$ and $B$. The protocol satisfies the following properties.
    \begin{itemize}
        \item \textbf{Lack of bias.} \[\Bigg|\Pr_{(k_A,k_B,\tau) \gets \langle A(1^\secp),B(1^\secp)\rangle}[k_A = 0] - \frac{1}{2} \Bigg| = \negl(\secp), \ \ \ \Bigg|\Pr_{(k_A,k_B,\tau) \gets \langle A(1^\secp),B(1^\secp)\rangle}[k_B = 0] - \frac{1}{2}\Bigg| = \negl(\secp),\] 
        \item \textbf{Correctness.}
        \[\Pr_{(k_A,k_B,\tau) \gets \langle A(1^\secp),B(1^\secp)\rangle}[k_A = k_B] \geq \frac{1}{2} + \frac{\epsilon}{2}-\negl(\secp).\]
        \item \textbf{Security.} For any QPT adversary $\rho,\widetilde{E}$, 
        \[\Pr_{(k_A,k_B,\tau) \gets \langle A(1^\secp),B(1^\secp)\rangle}\left[\widetilde{E}(\rho,\tau) = k_A \ | \ k_A = k_B \right] \leq \frac{1}{2}+\frac{\delta}{2} + \negl(\secp).\]
    \end{itemize}
\end{definition}

We introduce two definitions of oblivious transfer (OT). First, we give a standard game-based definition, specified as follows.

\begin{definition}[Oblivious transfer]\label{def:standard-OT}
    An $(\epsilon,\delta,\gamma)$-weak oblivious transfer (OT) with classical communication is a protocol that takes place between a polynomial-time (potentially quantum) sender $S$ and a polynomial-time (potentially quantum) receiver $R$, denoted by 
    \[(b,r),(r_0,r_1),\tau \gets \langle R(1^\secp),S(1^\secp)\rangle,\] where $(b,r)$ is the output of $R$, $(r_0,r_1)$ is the output of $S$, and $\tau$ is the (classical) transcript of interaction produced by the protocol. It should satisfy the following properties.
    \begin{itemize}
        \item \textbf{Correctness}: It holds that \[\Pr[r = r_b : \begin{array}{r} (b,r),(r_0,r_1),\tau \gets \langle R(1^\secp),S(1^\secp)\rangle \\ \end{array}] \geq \frac{1}{2} + \frac{\epsilon}{2} - \negl(\secp).\]
        \item \textbf{Receiver security}: For any QPT adversarial sender $\rho,\widetilde{S}$,
       \[\Pr[\widetilde{b} = b : (b,r),\widetilde{b} \gets \langle R(1^\secp),\widetilde{S}(\rho)\rangle] \leq \frac{1}{2} + \frac{\gamma}{2} + \negl(\secp).\]
        \item \textbf{Sender security}: For any QPT adversarial receiver $\rho,\widetilde{R}$,
        \[\Pr[\widetilde{r} = r_0 \oplus r_1 : \begin{array}{r} \widetilde{r},(r_0,r_1) \gets \langle \widetilde{R}(\rho),S(1^\secp)\rangle\end{array}] \leq \frac{1}{2} + \frac{\delta}{2} + \negl(\secp).\]

        %\item \textbf{Sender security}: There exists a (potentially unbounded) function $\Ext$ such that for any QPT adversarial receiver $\widetilde{R}$,
        %\[\Pr[\widetilde{r} = r_{1-b} : \begin{array}{r} \widetilde{r},(r_0,r_1),\tau \gets \langle \widetilde{R},S(1^\secp)\rangle \\ b \coloneqq \Ext(\tau)\end{array}] \leq \frac{1}{2} + \negl(\secp).\]
    \end{itemize}
\end{definition}

\begin{remark}
    We refer to $(1,0,0)$-weak OT as ``standard OT'', or just ``OT''\footnote{by "standard", we mean fully-secure under our choice of game-based definition}.
\end{remark}

 Next, we give a specialized variant of OT that we call \emph{committed-bit} OT.

\begin{definition}[Committed-bit OT]\label{def:committed-bit}

    An $(\epsilon,\delta,\gamma)$-weak committed-bit OT with classical communication is a protocol that takes place between a polynomial-time (potentially quantum) sender $S$ and a polynomial-time (potentially quantum) receiver $R$. The protocol begins with a setup phase, denoted by \[\init_R,\tau_\Setup \gets \Setup\langle R(1^\secp),S(1^\secp)\rangle,\] where $\init_R$ is $R$'s private state and $\tau_\Setup$ is the public (classical) transcript of interaction. Next, for any $b \in \{0,1\}$, the parties run \[(r,\state_R),(r_0,r_1,\state_S),\tau \gets \OT\langle R(b,\init_R),S(\tau_\Setup)\rangle.\] In addition, there exists a (potentially inefficient) function $\TrapGen$, a polynomial-time function $\Ext$, and a verification protocol \[\{\top,\bot\} \gets \Ver\langle R(\state_R),S(b,\state_S)\rangle,\]

    such that the following properties hold.

    \begin{itemize}
        \item \textbf{Correctness}: For any $b \in \{0,1\}$, \[\Pr[r = r_b : \begin{array}{r}\init_R,\tau_\Setup \gets \Setup\langle R(1^\secp),S(1^\secp)\rangle \\ (r,\state_R),(r_0,r_1,\state_S),\tau \gets \OT\langle R(b,\init_R),S(\tau_\Setup)\rangle \end{array}] \geq \frac{1}{2} + \frac{\epsilon}{2} - \negl(\secp).\]
        \item \textbf{Completeness}: For any $b \in \{0,1\}$, \[\Pr\left[\Ver\langle R(\state_R),S(b,\state_S)\rangle = \top : \begin{array}{r}\init_R,\tau_\Setup \gets \Setup\langle R(1^\secp),S(1^\secp)\rangle \\ (r,\state_R),(r_0,r_1,\state_S),\tau \gets \OT\langle R(b,\init_R),S(\tau_\Setup)\rangle \end{array}\right] = 1-\negl(\secp).\]
        
        \item \textbf{Receiver security}: For any QPT adversarial sender $\rho,\widetilde{S}$,
        \begin{align*}&\Bigg|\Pr[\widetilde{b} = 0 : \begin{array}{r} \init_R,\init_S,\tau_\Setup \gets \Setup\langle R(1^\secp),\widetilde{S}(\rho)\rangle \\ (r,\state_R),(\widetilde{b},\state_S),\tau \gets \OT\langle R(0,\init_R),\widetilde{S}(\init_S)\rangle \end{array}] \\ &- \Pr[\widetilde{b} = 0 : \begin{array}{r} \init_R,\init_S,\tau_\Setup \gets \Setup\langle R(1^\secp),\widetilde{S}(\rho)\rangle \\ (r,\state_R),(\widetilde{b},\state_S),\tau \gets \OT\langle R(1,\init_R),\widetilde{S}(\init_S)\rangle \end{array}]\Bigg| \leq \gamma + {\negl(\secp)}.\end{align*}
        \item \textbf{Sender security}: For any unbounded algorithm $\widetilde{R}_\init$, with probability $1-\negl(\secp)$ over \[\init_R,\tau_\Setup \gets \Setup\langle \widetilde{R}_\init,S(1^\secp)\rangle, \quad \td \coloneqq \TrapGen(\tau_\Setup),\] it holds that for any QPT adversarial receiver $\widetilde{R}$,
        \begin{align*}\Pr\left[\widetilde{R}(\state_R,r_b) = r_{1-b} : \begin{array}{r} \state_R,(r_0,r_1,\state_S),\tau \gets \OT\langle \widetilde{R}(\init_R),S(\tau_\Setup)\rangle \\ b \coloneqq \Ext(\td,\tau) \end{array}\right] \leq \frac{1}{2} + \frac{\delta}{2} + \negl(\secp),\end{align*}
        and, \[\Pr\left[\Ver\langle \widetilde{R}(\state_R),S(1-b,\state_S)\rangle = \top : \begin{array}{r} \state_R,(r_0,r_1,\state_S),\tau \gets \OT\langle \widetilde{R}(\init_R),S(\tau_\Setup)\rangle \\ b \coloneqq \Ext(\td,\tau) \end{array}\right] = \negl(\secp).\]
    \end{itemize}

\end{definition}

\begin{remark}
    We note that any $(\epsilon,\delta,\gamma)$-weak committed-bit OT implies $(\epsilon,\delta,\gamma)$-weak OT secure against adversaries with non-uniform quantum advice. Correctness and receiver security are immediate, while sender security follows due to the fact that if $\widetilde{R}$ can compute $r_0 \oplus r_1$ with some advantage, then, given $r_b$, it can output $r_{1-b}$ with (at least) the same advantage.
\end{remark}

\subsection{Alternating projectors}
Marriott and Watrous introduce algorithms based on alternating projectors, which they used for witness-preserving QMA amplification \cite{marriott2005}. Chiesa et al.\ present tools for analyzing these algorithms \cite{ChiesaRewinding}; we now review these tools, which will later be used for proving Theorem \ref{thm:sequential-rep}. 

\paragraph{Jordan decomposition.}
By Jordan's lemma applied to $(\Pi_A,\Pi_B)$, $\mathcal{H}$ decomposes orthogonally as
\[
\mathcal{H} \;=\; \bigoplus_j S_j,
\]
where each $S_j$ has dimension at most $2$, and is preserved by both projections.  For a $2$-dimensional block $S_j$ choose
orthonormal pairs $\{\ket{v^{A}_{j,1}},\ket{v^{A}_{j,0}}\}$ and
$\{\ket{v^{B}_{j,1}},\ket{v^{B}_{j,0}}\}$ such that
$\ket{v^{A}_{j,1}}\in \mathrm{im}(\Pi_A)$, $\ket{v^{A}_{j,0}}\in \ker(\Pi_A)$ and similarly for $B$.
Define the overlap parameter
\[
p_j := \bigl|\bracket{v^{A}_{j,1}}{v^{B}_{j,1}}\bigr|^2 \in [0,1].
\]

\noindent By convention, phases may be fixed so that within $S_j$, the bases are related by a planar rotation:
\[
\ket{v^{A}_{j,1}} = \sqrt{p_j}\,\ket{v^{B}_{j,1}} + \sqrt{1-p_j}\,\ket{v^{B}_{j,0}},
\qquad
\ket{v^{B}_{j,1}} = \sqrt{p_j}\,\ket{v^{A}_{j,1}} + \sqrt{1-p_j}\,\ket{v^{A}_{j,0}}.
\]

% --- Corollary (of the Jordan decomposition paragraph above), with minimal repetition ---

\begin{corollary}
\label{cor:jordan-param-eig}
In the Jordan decomposition notation above, fix a $2$-dimensional block $S_j$ and let
$p_j := |\langle v^A_{j,1}\mid v^B_{j,1}\rangle|^2$, where the bases are chosen so that
\begin{equation}
\label{eq:rotation-relation}
|v^A_{j,1}\rangle \;=\; \sqrt{p_j}\,|v^B_{j,1}\rangle + \sqrt{1-p_j}\,|v^B_{j,0}\rangle.
\end{equation}
Then
\[
(\Pi_A\Pi_B\Pi_A)\,|v^A_{j,1}\rangle \;=\; p_j\,|v^A_{j,1}\rangle.
\]
In particular, each Jordan parameter $p_j$ is an eigenvalue of the restriction
$\left.\Pi_A\Pi_B\Pi_A\right|_{\mathrm{im}(\Pi_A)}$. Consequently, if
\[
\mathrm{spec}\!\left(\left.\Pi_A\Pi_B\Pi_A\right|_{\mathrm{im}(\Pi_A)}\right)\subseteq [a,b],
\]
then every Jordan parameter $p_j$ lies in $[a,b]$.
\end{corollary}

\begin{proof}
Fix $j$ and abbreviate $|v^A\rangle:=|v^A_{j,1}\rangle$ and $|v^B\rangle:=|v^B_{j,1}\rangle$.
Since $|v^A\rangle\in\mathrm{im}(\Pi_A)$ we have $\Pi_A|v^A\rangle=|v^A\rangle$.
Using the rotation relation \eqref{eq:rotation-relation} and that $\Pi_B|v^B_{j,1}\rangle=|v^B_{j,1}\rangle$ and
$\Pi_B|v^B_{j,0}\rangle=0$, we get
\[
\Pi_B|v^A\rangle \;=\; \sqrt{p_j}\,|v^B\rangle.
\]
Applying $\Pi_A$ again and using the second rotation relation
$|v^B\rangle=\sqrt{p_j}\,|v^A\rangle+\sqrt{1-p_j}\,|v^A_{j,0}\rangle$ together with
$\Pi_A|v^A_{j,0}\rangle=0$, we obtain
\[
\Pi_A|v^B\rangle \;=\; \sqrt{p_j}\,|v^A\rangle.
\]
Combining,
\[
\Pi_A\Pi_B\Pi_A|v^A\rangle
\;=\;
\Pi_A\Pi_B|v^A\rangle
\;=\;
\sqrt{p_j}\,\Pi_A|v^B\rangle
\;=\;
p_j\,|v^A\rangle,
\]
as claimed. The final implication follows since $p_j$ is an eigenvalue of
$\left.\Pi_A\Pi_B\Pi_A\right|_{\mathrm{im}(\Pi_A)}$.
\end{proof}

The following analysis is concerned with algorithms that alternately apply two binary
projective measurements $A=(\Pi_A, I-\Pi_A)$ and $B=(\Pi_B, I-\Pi_B)$ on a Hilbert space
$\mathcal{H}$, stopping after a prescribed number of steps or upon observing a desired outcome. 

\paragraph{A classical distribution for alternating outcomes.}
Fix $T\in\mathbb{N}$ and $p\in[0,1]$. Define a distribution $\mathrm{MWDist}(T,p)$ on bit strings
$(b_1,\dots,b_T)$ by:
\begin{enumerate}
  \item Sample independent $a_1,\dots,a_T\in\{0,1\}$ with $\Pr[a_i=1]=p$.
  \item Set $b_0:=1$ and update $b_i := b_{i-1} \oplus (1-a_i)$ for $i=1,\dots,T$.
  \item Output $(b_1,\dots,b_T)$.
\end{enumerate}

\begin{lemma}[Lemma~4.4 in~\cite{ChiesaRewinding}]
If the initial post-measurement state is $\ket{v^{B}_{j,1}}$ and we apply $T$ measurements
alternating $A,B,A,B,\ldots$, then the resulting outcome bits $(b_1,\ldots,b_T)$ are distributed as
$\mathrm{MWDist}(T,p_j)$.
\end{lemma}

\begin{lemma}[Lemma~4.5 in~\cite{ChiesaRewinding}]
If the initial state lies in $\mathrm{im}(\Pi_B)$, it must take the form 
$\ket{\psi}=\sum_j \alpha_j \ket{v^{B}_{j,1}}$. Then the outcome distribution equals:
sample an index $j$ with probability $|\alpha_j|^2$, then sample from $\mathrm{MWDist}(T,p_j)$.
\end{lemma}

\noindent Let $\tilde b=(b_0,b_1,\ldots,b_n)\in\{0,1\}^{n+1}$. Define the normalized count of consecutive repeats
\[
\mathrm{NReps}(\tilde b) \;:=\; \frac{1}{n}\bigl|\{\,j\in\{1,\ldots,n\}: b_{j}=b_{j-1}\,\}\bigr|.
\]

\begin{proposition}[Proposition~4.7 in~\cite{ChiesaRewinding}] \label{prop:mwDist}
If $(b_1,\ldots,b_T)\sim \mathrm{MWDist}(T,p)$ , then
$\mathrm{NReps}(1,b_1,\ldots,b_T)$ has the same distribution as $\mathrm{Bin}(T,p)/T$. In particular, $\mathbb{E}[\NReps(1,b_1,\dots,b_T)]=p$.

\end{proposition}

Hence, applying $T$ alternating measurements and outputting the repeat fraction
provides an estimator whose expectation (in block $j$) is $p_j$. 

\begin{proposition}[Additive Chernoff bound, Proposition~3.1 in~\cite{ChiesaRewinding}]
Let $X\sim \mathrm{Bin}(n,p)$. For $\varepsilon,\eta>0$, if

\[
    n \geq \frac{\ln(2/\eta)}{2\varepsilon^2}.
\]

then
\[
\Pr\Big[\; \big|X/n - p\big|\le \varepsilon \;\Big]\;\ge\; 1-\eta.
\]
\end{proposition}

Using Chernoff's bound, we have
$\Pr[|\mathrm{NReps}(\mathrm{MWDist}(T,p_j))-p_j|\le \varepsilon]\ge 1-\eta$ whenever
$T \geq \ln(2/\eta)/(2\varepsilon^2)$.

\fi
\section{Tests of non-commutation}

\subsection{Definitions}

\begin{definition}[Test of Non-Commutation]\label{def:ToNC}
    An $(\alpha, \beta)$ classical test of non-commutation (ToNC) is a two-stage interactive protocol between a QPT (uniform) prover $P$ and a PPT (uniform) verifier $V$. The first stage is defined by interactive strategies $P_\prep$ and $V_\prep$, which we denote by 
    \[\left(\ket{\psi},(c,\state)\right) \gets \langle P_\prep(1^\secp),V_\prep(1^\secp)\rangle,\] where $\ket{\psi}$ is $P$'s output, $(c,\state)$ is $V$'s output, and $c \in \{0,1\}, \state \in \{0,1\}^*$. The second stage is defined by two binary-outcome observables $P_0, P_1$, where $P_c$ is applied in the case that the first stage outputs $c$, and a verification predicate $V_\out(\state,a) \to \{0,1\}$. It should satisfy the following properties.
    \begin{itemize}
        \item \textbf{Completeness}: It holds that\footnote{In an abuse of notation, we write $P(\ket{\psi})$ to denote the distribution induced by measuring state $\ket{\psi}$ with projective measurement $\{P,I-P\}$.} 
        \[\Pr\left[V_\out(\state,a) = 1 : \begin{array}{r}\ket{\psi},(c,\state) \gets \langle P_\prep(1^\secp),V_\prep(1^\secp)\rangle \\ a \gets P_c(\ket{\psi})\end{array}\right] \geq \beta.\]
        \item \textbf{Soundness}: For any non-uniform QPT strategy $\rho,\widetilde{P}_\prep, \widetilde{P}_0,\widetilde{P}_1$ such that $\widetilde{P}_0\widetilde{P}_1 = \widetilde{P}_1\widetilde{P}_0$, \[\Pr\left[V_\out(\state,a) = 1 : \begin{array}{r}\ket{\psi},(c,\state) \gets \langle \widetilde{P}_\prep(1^\secp,\rho),V_\prep(1^\secp)\rangle \\ a \gets \widetilde{P}_c(\ket{\psi})\end{array}\right] \leq \alpha + \negl(\secp).\]
    \end{itemize}
\end{definition}

\begin{definition}[Normal form test of non-commutation]\label{def:normal-form}
    We say that the test of non-commutation (\cref{def:ToNC}) is in \emph{normal form} if:
\begin{itemize}
    \item $\state = a^* \in \{0,1\}$, and 
    \item $V_\out(a^*,a) = 1$ iff $a = a^*$.
\end{itemize}
That is, there is always exactly one answer bit that the verifier accepts.
\end{definition}

\begin{remark}\label{remark:advantage}
    In any normal-form ToNC, both $\beta,\alpha \geq 1/2$ since there always exists a trivial strategy which outputs a uniformly random bit. Thus, it is often convenient to write the parameters of a normal-form ToNC using the \emph{advantage} of the prover as opposed to the raw probability of success. That is, we will often refer to a normal form $(\alpha,\beta)$-ToNC where $\alpha, \beta \in [1/2,1]$ as an $(\epsilon,\delta)$-ToNC, where $\epsilon \coloneqq 2\beta - 1, \delta \coloneqq 2\alpha - 1 \in [0,1]$. We will also assume without loss of generality that $\delta \geq \epsilon/2$, as justified in our discussion in \cref{subsec:results}.
\end{remark}

Next, we specify a stronger completeness guarantee stating that the honest prover’s advantage is negligibly close to \(\epsilon\) for all but a negligible fraction of honest preambles. Note that this is satisfied by the compiled CHSH game (with $\epsilon = 2\cos^2(\pi/8) - 1-\negl(\secp)$) and the compiled magic square game (with $\epsilon = 1-\negl(\secp)$).

\begin{definition}[Test of non-commutation with robust completeness]\label{def:robust-completeness}
    We say that a normal-form $(\epsilon,\delta)$-ToNC has
\emph{robust completeness} if there exists a negligible function
$\nu(\lambda)$ such that, with probability at least $1-\nu(\lambda)$
over
\[
|\psi\rangle,(c,\mathsf{st})
\leftarrow
\langle P_{\mathsf{Prep}}(1^\lambda),
        V_{\mathsf{Prep}}(1^\lambda)\rangle,
\]
we have
\[
\left|
\Pr[V_{\mathsf{Ver}}(\mathsf{st},a)=1 :
    a\leftarrow P_c(|\psi\rangle)]
-\left(\frac12+\frac{\epsilon}{2}\right)
\right|
\le \nu(\lambda).
\]
\end{definition}

% Normal-form reduction, streamlined version.
%\ifsubmission\section{Normal form compiler}\else

\subsection{Normal form compiler}\label{subsec:normal-form}

In this subsection, we present a procedure that transforms a less structured form of ToNC into a normal-form ToNC as in Definition~\ref{def:normal-form}. This transformation incurs an inverse-polynomial parameter loss. Whereas in the final round of a normal-form ToNC the verifier accepts exactly one of the two possible prover answers, in the less structured form, the verifier may alternatively accept both or neither. Since our constructions of oblivious transfer and key
agreement use normal-form ToNCs, the compilation procedure below shows that the
same implications hold for a broader class of ToNC protocols. Moreover, this compilation procedure suggests that normal-form ToNC is not merely a convenient technical restriction, but a natural abstraction: it captures, up to bounded parameter loss, a broad class of ToNC protocols.

\begin{theorem}[Normal form compiler]
\label{thm:tonc-normal-form-succinct}
Let $I$ be two-stage protocol satisfying Definition~\ref{def:ToNC}, except that the second-stage predicate
$V_\out(\state,\cdot)$ may (for some values of $\state$) accept both $a\in\{0,1\}$ or accept neither. Then for any inverse-polynomial function $\varepsilon(\secp)$, there exists a \emph{normal-form} test $I^{\mathsf{nf}}$ with parameters $(\alpha',\beta')$ such that
\[
\beta' - \alpha' \;\ge\; \beta - \alpha - O(\varepsilon).
\]

Here, the verifier of $I_{\mathrm{nf}}$ may perform uniform quantum
polynomial-time preprocessing on input $1^\lambda$, producing only
classical information. After preprocessing, the verifier is classical
probabilistic polynomial-time, and all communication is classical.
\end{theorem}

\begin{remark}
    As stated above, the verifier in the constructed normal-form test $I^{\mathsf{nf}}$ may include a \emph{quantum} polynomial-time preprocessing phase that takes as input the security parameter and outputs some classical information. After this, its interaction with the prover is entirely classical. While this does not exactly fit the syntax of \cref{def:ToNC}, we remark that all the implications of (normal-form) ToNC that we show (i.e. to classical-communication KA and OT) follow in exactly the same manner if we start with a ToNC with quantum verifier preprocessing.
\end{remark}

\begin{proof}

By the completeness of $I$, there exists a uniform QPT strategy $S=(P_\prep,P_0,P_1)$ that succeeds with probability at least $\beta$ for any value of the security parameter $\secp$. Let $(c,\state)$ be the verifier output in the preparation stage when interacting with $P_\prep$. For any such output, define the acceptance set
\[
A(c,\state) := \{a\in\{0,1\} : V_\out(c,\state,a)=1\},
\]
so that $|A(c,\state)|\in\{0,1,2\}$. Define
\[
\ell(\secp) := \Pr[|A(c,\state)|=2],\qquad
\mu(\secp) := \Pr[|A(c,\state)|=0],\qquad
\gamma(\secp) := \Pr[|A(c,\state)|=1], 
\]
where the probabilities are taken over random interactions between $S(\secp)$ and the verifier $V(\secp)$. 

Given the security parameter $\secp$, the verifier in our normal-form test $I^{\mathsf{nf}}$ must determine the numbers $\ell(\secp), \mu(\secp), \gamma(\secp)$. For all non-commutation tests we are aware of, these numbers are in fact constants independent of the security parameter, and thus they can be hard-coded into the verifier. However, in general, this may not be the case, so the verifier will \emph{estimate} these values by running a quantum pre-processing step in $I^{\mathsf{ns}}$ (Step 1 in the protocol described below).

Now, let $a$ be the answer obtained by applying $P_c$ to the prover state, and define the conditional success on
\emph{unique-accepting} transcripts
\[
\xi(\secp) := \Pr[a=a^\star(c,\state)\mid |A(c,\state)|=1],
\]

\noindent where if $|A(c,\state)|=1$, $a^\star(c,\state)$ is the unique accepted answer. Then the overall success of $S$ decomposes as
\begin{equation}
\label{eq:succ-decomp-succinct}
\Pr[I\text{ accepts with }S] \;=\; \ell + \gamma\,\xi \;\ge\; \beta.
\end{equation}

Consider the following strategy in $I$. It executes $S$ in the first stage of the protocol, and in the second stage it sends the verifier a uniformly random bit. This strategy succeeds in the protocol with probability $\ell+\frac{\gamma}{2}$, and applies perfectly commuting observables in the second stage. Therefore, the soundness of $I$ implies that $\ell+\frac{\gamma}{2}\le \alpha$. On the other hand, $\ell+\gamma\xi\ge\beta$. Thus, 

\[\beta-\alpha\le (\xi-\frac{1}{2})\gamma\rightarrow \gamma\ge\frac{\beta-\alpha}{\xi-\frac{1}{2}} \qquad \text{(and indeed}\; \xi>\frac{1}{2} \;\text{because}\; \beta>\alpha\text{)}.\]

\noindent We deduce that $\gamma>0$. Moreover, since $\alpha<\beta$ are constants, we get  $\gamma(\secp)\ge\frac{1}{\poly(\secp)}$.

\medskip
\noindent\textbf{The normal-form protocol $I^{\mathsf{nf}}$.}
Fix functions $\varepsilon(\secp)=\frac{1}{\poly(\secp)}$ such that $\varepsilon<\frac{\gamma}{4}$, and $\eta>0$, where $\eta(\lambda)=\negl(\secp)$, and let $r, v$ be two polynomials. The verifier operates as follows:
\begin{enumerate}
    %\item Runs $T(\secp)$ and obtains $\ell:=\ell(\secp),\mu:=\mu(\secp),\gamma:=\gamma(\secp)$.
    \item Obtain estimates of $\ell:=\ell(\secp),\mu:=\mu(\secp),\gamma:=\gamma(\secp)$, either as hard-coded values or by simulating interaction with the prover strategy $S$ several times and taking the averages. By running enough times, the verifier can obtain estimates that, with overwhelming probability, are arbitrarily inverse-polynomial close to the real values. We suppress mention of this approximation for the rest of the proof, and assume that the verifier obtains the exact values $\ell:=\ell(\secp),\mu:=\mu(\secp),\gamma:=\gamma(\secp)$.
    \item Sets a variable $\sf FLAG:=0$, draws $w\in[v(\secp)]$ uniformly at random. For $i=1,\ldots,w$:
        \begin{enumerate}
        \item Runs $r$ copies of $I$
    sequentially with the prover, using fresh verifier randomness at each copy. The verifier completes only the first step of the protocol at each copy, producing outputs
    $(c_i,\state_i)$.
        \item Computes $|A(c_i,\state_i)|$ by executing $V_\out(\state_i,0)$ and $V_\out(\state_i,1)$,
    and sets
    \[
    \widehat{\ell} := \frac{1}{r}|\{i:|A(c_i,\state_i)|=2\}|,\quad
    \widehat{\mu} := \frac{1}{r}|\{i:|A(c_i,\state_i)|=0\}|,\quad
    \widehat{\gamma} := \frac{1}{r}|\{i:|A(c_i,\state_i)|=1\}|.\]
    \item If $\widehat{\gamma}=0$, or if any of $|\widehat{\ell}-\ell|,|\widehat{\mu}-\mu|,|\widehat{\gamma}-\gamma|$ exceeds
$\varepsilon$, the verifier sets $\sf FLAG:=1$ and exits the loop.
 \end{enumerate}
    \item The verifier then runs the following steps.
        \begin{enumerate}
        \item Runs step 2 (a) once, and thus produces outputs $(c_i,\state_i)$.
        \item Samples $i\gets [r]$ uniformly at random, sets $c:=c_i$ and sends $i$ to the prover. If $|A(c_i,\state_i)|\ne 1$, sets $\sf FLAG =1$.
        \item If $\sf FLAG =1$, the verifier draws a uniformly random bit $\state\gets \{0,1\}$, sends $c$ to the prover and receives an answer bit $a$, and then accepts if and only if $\state=a$. 
        \item Otherwise,  the verifier sets $\state:=a^\star(c_i,\state_i)$, where $a^\star(c_i,\state_i)$ is the unique answer accepted by the verifier in the $i$-th instance. It sends $c$ to the prover. The prover sends an answer bit $a$, and the verifier accepts if and only if $\state=a$.
        \end{enumerate}
\end{enumerate}

\textbf{Completeness.} 
We consider the prover strategy $\widetilde{S}$ for $I^{\mathsf{nf}}$, which executes $S$ independently in each copy of $I$. By standard concentration bounds and a union bound over the three empirical frequencies, for large enough polynomials $r,v$ we obtain that the probability that $\sf{FLAG}=1$ at the end of stage $2$ of the protocol is at most $\eta(\secp)$.

Conditioned on the event that $\sf{FLAG}=0$ at the end of step 2, with probability $\gamma$, the verifier does not set $\sf FLAG=1$ in stage 3(b), and in that case the acceptance probability is $\xi=\frac{\beta-\ell}{\gamma}$. $\xi$ is the acceptance probability of a prover running the strategy $S$ in the original protocol $I$, conditioned on $|A(c,\sf{st})|=1$. If the verifier does set $\sf FLAG=1$ in stage 3(b), the acceptance probability is $\frac{1}{2}$.

Overall,

\begin{equation}
\label{eq:beta-prime-completeness}
\Pr[I^{\mathsf{nf}}\text{ accepts with }\widetilde{S}] \;\ge\; \xi \gamma + \frac{1}{2}(1-\gamma)-\eta,\; \text{where}\; \xi = \frac{\beta-\ell}{\gamma}.
\end{equation}

\textbf{Soundness.}
Let
\[
\widetilde S=(\widetilde P_\prep,\widetilde P_0,\widetilde P_1)
\]
be an arbitrary QPT strategy for \(I^{\mathsf{nf}}\) whose final observables commute:
\[
\widetilde P_0\widetilde P_1=\widetilde P_1\widetilde P_0 .
\]
Let \(s\) denote the success probability of \(\widetilde S\) in \(I^{\mathsf{nf}}\).

Recall that a block consists of \(r\) first-stage executions of \(I\). We call a block
\emph{good} if its empirical frequencies satisfy
\[
|\widehat{\ell}-\ell|\le \varepsilon,\qquad
|\widehat{\mu}-\mu|\le \varepsilon,\qquad
|\widehat{\gamma}-\gamma|\le \varepsilon .
\]
Since we chose \(\varepsilon<\gamma/4\), every good block also satisfies
\(\widehat{\gamma}>0\). Let \(F\) be the event that \(\sf FLAG=0\) at the end of stage \(2\), and write
\[
u:=\Pr[F].
\]
Thus, \(F\) is the event that every one of the \(w\) cut-and-choose blocks in stage \(2\)
is good.

We first record the following random-stopping fact. This is the point
where we use the fact that \(w\gets[v]\) is hidden from the prover.

\begin{claim}
\label{claim:random-stopping}
For every strategy \(\widetilde S\),
\[
\Pr[F \wedge \text{the stage-3 block is bad}]\le \frac{1}{v}.
\]
Consequently, if \(u>0\), then
\[
\Pr[\text{the stage-3 block is bad}\mid F]\le \frac{1}{uv}.
\]
\end{claim}

\begin{proof}
Consider the coupled execution in which the verifier keeps generating blocks of \(r\)
first-stage executions using the same interface that the prover sees before the index \(i\)
is sent. Let \(T\) be the index of the first bad block in this coupled execution, with
\(T=\infty\) if no bad block ever occurs.

The stopping time \(w\gets[v]\) is sampled independently of the prover's behavior and is
not revealed before the stage-3 block has already been generated. Then, consider the event
\[
F \wedge \text{the stage-3 block is bad}.
\]

This is exactly the events that $T=w+1$. Since \(w\) is uniform in \([v]\), this event
has probability at most \(1/v\). The conditional bound follows by dividing by
\(\Pr[F]=u\).
\end{proof}

\noindent We now split into two cases.

\medskip
\noindent\emph{Case 1: \(u\le \varepsilon\).}
Conditioned on \(\neg F\), the verifier has already set \(\sf FLAG=1\), and therefore the final
target bit \(\state\) is uniform and independent of the prover's final answer. Thus the prover
succeeds with probability exactly \(1/2\) conditioned on \(\neg F\). Conditioned on \(F\), its
success probability is at most \(1\). Therefore
\[
s\le \frac{1-u}{2}+u
   = \frac12+\frac{u}{2}
   \le \frac12+\frac{\varepsilon}{2}.
\]
On the other hand, the strategy for \(I\) that runs the honest first-stage strategy \(S\) and
then answers with a uniformly random bit has commuting final observables and succeeds with
probability \(\ell+\gamma/2\). By the soundness of \(I\),
\[
\ell+\frac{\gamma}{2}\le \alpha+\negl(\secp).
\]
Hence, in this case,
\[
s
\le
\frac12+\frac{\varepsilon}{2}
\le
\frac12+\alpha-\ell-\frac{\gamma}{2}
+O(\varepsilon)+\negl(\secp).
\]

\medskip
\noindent\emph{Case 2: \(u>\varepsilon\).}
Choose \(v\) large enough so that \(v\ge 1/\varepsilon^2\). By
Claim~\ref{claim:random-stopping},
\[
\Pr[\text{the stage-3 block is bad}\mid F]
\le \frac{1}{uv}
\le \varepsilon .
\]

Let \(q\gets[r]\) be the uniformly random coordinate selected by the verifier in stage \(3\).
Conditioned on \(F\), define
\[
\overline{\ell}
:=
\Pr[|A(c_q,\state_q)|=2\mid F],
\qquad
\overline{\gamma}
:=
\Pr[|A(c_q,\state_q)|=1\mid F].
\]
Since, conditioned on \(F\), the stage-3 block is good except with probability at most
\(\varepsilon\), and since \(q\) is uniform in the stage-3 block, we have
\begin{equation}
\label{eq:bar-frequencies-close}
|\overline{\ell}-\ell|\le 2\varepsilon,
\qquad
|\overline{\gamma}-\gamma|\le 2\varepsilon .
\end{equation}

\noindent Next define
\[
\overline{\chi}
:=
\Pr[a=a^\star(c_q,\state_q)
\mid
F\wedge |A(c_q,\state_q)|=1],
\]
where \(a\) is the answer output by \(\widetilde S\) after receiving the selected index \(q\) and
the challenge \(c_q\). If the conditioning event has probability \(0\), set
\(\overline{\chi}=1/2\).

We now construct a commuting strategy \(K\) for the original protocol \(I\). The strategy \(K\)
internally runs \(\widetilde S\) and simulates the verifier of \(I^{\mathsf{nf}}\). It first simulates
stages \(1\)--\(2\) of \(I^{\mathsf{nf}}\). If \(\sf FLAG=1\) at the end of stage \(2\), it restarts this
simulation with a fresh copy of \(\widetilde S\). Since \(u>\varepsilon\), after
\[
z=\poly(1/\varepsilon,\secp)
\]
repetitions, \(K\) obtains an execution satisfying \(F\), except with negligible probability.
If all repetitions fail, \(K\) answers randomly.

Once an execution satisfying \(F\) has been obtained, \(K\) generates the stage-3 block. It
samples \(q\gets[r]\) uniformly at random and embeds its external interaction with the verifier
of \(I\) in the \(q\)-th coordinate of this block, while simulating all other coordinates honestly.
After the \(r\) first-stage executions have been completed, \(K\) sends the index \(q\) to
\(\widetilde S\). When the external verifier of \(I\) sends its challenge bit \(c\), the strategy \(K\)
forwards this challenge to \(\widetilde S\), obtains an answer \(a\), and returns \(a\) to the
external verifier.

The final observables of \(K\) are precisely the final observables induced by
\(\widetilde S\) after the simulated transcript and the classical index \(q\) have been fixed.
Therefore they commute, because \(\widetilde S\)'s final observables commute.

By construction, up to the negligible probability that \(K\) fails to obtain an execution
satisfying \(F\), the external instance of \(I\) is distributed as the selected coordinate of the
stage-3 block conditioned on \(F\). Therefore the success probability of \(K\) in \(I\) is
\[
\overline{\ell}+\overline{\gamma}\,\overline{\chi}-\negl(\secp).
\]
Indeed, if \(|A(c_q,\state_q)|=2\), every answer is accepted; if
\(|A(c_q,\state_q)|=0\), no answer is accepted; and if
\(|A(c_q,\state_q)|=1\), the success probability is exactly \(\overline{\chi}\). By the soundness of the original protocol \(I\), we obtain
\begin{equation}
\label{eq:reduction-soundness}
\overline{\ell}+\overline{\gamma}\,\overline{\chi}
\le
\alpha+\negl(\secp).
\end{equation}

We now relate this to the success probability \(s\) of \(\widetilde S\) in \(I^{\mathsf{nf}}\). Conditioned
on \(\neg F\), the verifier uses a uniformly random target bit, so the success probability is
\(1/2\). Conditioned on \(F\), the verifier uses the real unique accepting bit only if the selected
coordinate is unique-accepting; otherwise it sets \(\sf FLAG=1\) and uses a uniformly random
target bit. Hence
\[
s
=
\frac{1-u}{2}
+
u\left(
\overline{\gamma}\,\overline{\chi}
+
\frac{1-\overline{\gamma}}{2}
\right)
=
\frac12
+
u\,\overline{\gamma}\left(\overline{\chi}-\frac12\right).
\]
If \(\overline{\chi}\le 1/2\), then \(s\le 1/2\), and using
\(\ell+\gamma/2\le \alpha+\negl(\secp)\), we again get
\[
s
\le
\frac12+\alpha-\ell-\frac{\gamma}{2}
+\negl(\secp).
\]
Otherwise, \(\overline{\chi}>1/2\). Since \(u\le 1\),
\[
s
\le
\frac12+\overline{\gamma}\left(\overline{\chi}-\frac12\right)
=
\frac12+\overline{\gamma}\,\overline{\chi}
-\frac{\overline{\gamma}}{2}.
\]
Using \eqref{eq:reduction-soundness}, this gives
\[
s
\le
\frac12+\alpha-\overline{\ell}
-\frac{\overline{\gamma}}{2}
+\negl(\secp).
\]
Finally, by \eqref{eq:bar-frequencies-close},
\[
s
\le
\frac12+\alpha-\ell-\frac{\gamma}{2}
+O(\varepsilon)+\negl(\secp).
\]
Since \(\widetilde S\) was arbitrary, the soundness parameter \(\alpha'\) of
\(I^{\mathsf{nf}}\) satisfies
\begin{equation}
\label{eq:alpha-prime-final}
\alpha'
\le
\frac12+\alpha-\ell-\frac{\gamma}{2}
+O(\varepsilon)+\negl(\secp).
\end{equation}

\medskip
\paragraph{Gap analysis.}
From the completeness analysis, we have
\[
\beta'
\ge
\frac12+\gamma\xi-\frac{\gamma}{2}-\eta .
\]
Using the decomposition
\[
\ell+\gamma\xi\ge \beta,
\]
we get
\begin{equation}
\label{eq:beta-prime-final}
\beta'
\ge
\frac12+\beta-\ell-\frac{\gamma}{2}-\eta .
\end{equation}
Combining \eqref{eq:alpha-prime-final} and \eqref{eq:beta-prime-final},
\[
\beta'-\alpha'
\ge
\beta-\alpha
-O(\varepsilon)-\eta-\negl(\secp).
\]
Choosing the completeness parameters so that \(\eta\le O(\varepsilon)\), and absorbing the
negligible term into the \(O(\varepsilon)\) loss, we conclude that
\[
\beta'-\alpha'
\ge
\beta-\alpha-O(\varepsilon),
\]
as required.

Recall that $I^{\mathsf{nf}}$ is constructed with a choice of $\varepsilon$, and it remains efficient for any $\varepsilon(\secp)=\frac{1}{\poly(\secp)}$. Thus, the above gap is attainable for any such $\varepsilon$. This completes the proof.
\end{proof}

\section{Weak cryptography from tests of non-commutation} 
\label{sec:OT_from_anti_comm_test}
\subsection{Weak bit agreement}\label{subsec:WBA}

Weak bit agreement is a protocol run between two potentially quantum parties $A$ and $B$ that communicate classically. The protocol is parameterized by $\epsilon$, which determines the probability that $A$ and $B$ agree on a bit, and $\delta$, which determines how much information an eavesdropper can obtain about the shared bit.

\begin{theorem}\label{thm:wba_from_test}
    Given any normal-form $(\epsilon,\delta)$ test of non-commutation secure against adversaries with non-uniform classical (resp. quantum) advice (\cref{def:ToNC}, \cref{remark:advantage}), there exists a $\left(\epsilon, \delta'\right)$-BA  secure against adversaries with non-uniform classical (resp. quantum) advice (\cref{def:bit-agreement}) for any \[\delta' \geq \frac{1+4\delta-3\epsilon}{1+\epsilon}.\] 
    
    Moreover, given any normal-form $(\epsilon,\delta)$ test of non-commutation  secure against adversaries with non-uniform classical (resp. quantum) advice with \emph{robust} completeness (\cref{def:robust-completeness}), there exists a $(\epsilon,\delta')$-BA secure against adversaries with non-uniform classical (resp. quantum) advice for any \[\delta' \geq 2\delta - \epsilon.\] 

%These results hold for non-uniform adversaries with either classical or quantum non-uniform advice
\end{theorem}

\begin{proof}
    The protocol is given in Protocol~\ref{prot:WBA_from_anti_comm_test}. The theorem follows by combining \cref{lem:wba_correctness} and \cref{lem:wba_security}.
\end{proof}

\begin{protocolbox}{Protocol: Weak bit agreement}
\label{prot:WBA_from_anti_comm_test}
\textbf{Parties.} Alice $A$ and Bob $B$ (classical communication; parties may be quantum).\\
\textbf{Ingredient.} Normal-form $(\epsilon,\delta)$ test of non-commutation $P_\prep,V_\prep,P_0,P_1,V_\Ver$ (\cref{def:ToNC}).

\medskip
\begin{enumerate}[leftmargin=*,label=\textbf{Step \arabic*:}]
  \item \textbf{Non-commutation test.}
  The parties run \[\ket{\psi},\left(c,a^*\right) \gets \langle P_\prep(1^\secp),V_\prep(1^\secp)\rangle,\] with $A$ playing the role of verifier and $B$ playing the role of prover. At its conclusion, $A$ obtains the ``correct'' answer bit $a^*$ and sends the challenge bit $c$ to $B$.

  \item \textbf{Random mask.} $A$ samples $r\gets\{0,1\}$ uniformly at random and sends it to $B$. 

  \item \textbf{Outputs.}
  $A$ outputs the bit $k_A \coloneqq a^* \oplus r$. Upon receiving $(c,r)$, $B$ applies $a \gets P_c(\ket{\psi})$ and outputs $k_B \coloneqq a \oplus r$.
\end{enumerate}
\end{protocolbox}

\begin{lemma}[Correctness]\label{lem:wba_correctness}
    Protocol~\ref{prot:WBA_from_anti_comm_test} satisfies Lack of bias and Correctness according to \cref{def:bit-agreement} with the same $\epsilon$ from the ToNC parameters.
\end{lemma}

\begin{proof}
The event $k_A=k_B$ is exactly the event that $a = a^*$,
and hence occurs with advantage $\epsilon$ due to the completeness of the ToNC. Moreover, due to the random masking, 
\[
\Pr[k_A=0]=\Pr[k_B=0]=\tfrac12.
\]
\end{proof}

\begin{lemma}[Security]
\label{lem:wba_security}
     Protocol~\ref{prot:WBA_from_anti_comm_test} satisfies Security according to \cref{def:bit-agreement} with $\delta' = \frac{1+4\delta-3\epsilon}{1+\epsilon}$. In the case that the ToNC has robust completeness, it satisfies Security with $\delta' = 2\delta - \epsilon$.

\end{lemma}

\begin{proof}
Fix a QPT adversary $\widetilde E$ and let $\gamma$ be such that $\Pr[\widetilde E(\tau)=k_A] = \frac{1}{2} + \frac{\gamma}{2}$ in an execution with honest $A,B$, where $\tau$ is the transcript produced by the interaction. By completeness of the ToNC, fix a QPT prover strategy $P_\prep,P_0,P_1$ that accepts with probability at least $\frac{1}{2} + \frac{\epsilon}{2}$. $E$ may or may not have non-uniform quantum advice $\rho$, which determines whether our new prover strategy does. We suppress mention of $\rho$ in what follows.

We build a new prover strategy $\widetilde P_\prep,\widetilde{P}_0,\widetilde{P}_1$ with commuting final-round observables, as follows.
$\widetilde P_\prep$ interacts with the verifier through the preparation stage by executing $P_\prep$. It also samples an independent uniform bit $r\leftarrow\{0,1\}$ as the random mask. For $t\in\{0,1\}$, let $\tau_t$
denote the preparation stage transcript followed by the challenge
$t$ and the random mask $r$. After receiving the final-round challenge $c\in\{0,1\}$,
it samples a fresh uniform bit $b\gets\{0,1\}$ and proceeds as follows:
\begin{itemize}
  \item If $c=b$, it answers using $P_c$.
  \item If $c\neq b$, it answers with
$\widetilde{E}(\tau_c)\oplus r$.
\end{itemize}

Equivalently, $\widetilde P_\prep$ can execute a \emph{joint} measurement producing both candidate answers $(\hat a_0,\hat a_1)$:
sample $b\leftarrow\{0,1\}$; measure using $P_b$ to obtain $\hat a_b$; compute $\widehat{a}_{1-b}\leftarrow
\widetilde{E}(\tau_{1-b})\oplus r$. Then, on challenge $c$, output $\hat a_c$.
Since this procedure produces both answers in one joint measurement, it is clear that the induced final-round observables commute.

Conditioned on $b=c$ (probability $1/2$), $\widetilde P$ behaves like $P$ on the final round, and thus has advantage at least $\epsilon$. Conditioned on $b\ne c$ (probability $1/2$), $\widetilde P$ has advantage $\gamma$, by definition. Hence, the advantage of $\widetilde P$ is at least $\frac{\epsilon+\gamma}{2}$.

By soundness of the ToNC against commuting strategies,

\[
\frac{\epsilon+\gamma}{2}\le \delta \quad\Rightarrow\quad
\gamma \le 2\delta-\epsilon.
\]

Since $\Pr[k_A=k_B]\ge \frac{1}{2}+\frac{\epsilon}{2}$, we have
\begin{align*}
2\Pr[\widetilde E(\tau)=k_A \mid k_A=k_B] - 1
&\le
2\left(\frac{\Pr[\widetilde E(\tau)=k_A]}{\Pr[k_A=k_B]}\right) - 1
\\ & \le
2\left(\frac{\frac{1}{2} + \delta - \frac{\epsilon}{2}}{\frac{1}{2} + \frac{\epsilon}{2}}\right) - 1  \\ &= \frac{2 + 4\delta - 2\epsilon}{1+\epsilon} -1 \\ &=\frac{1+4\delta - 3\epsilon}{1+\epsilon},
\end{align*}
which completes the proof.

Now suppose that the ToNC satisfies robust completeness.
Conditioned on the full honest preparation history and the public
mask, the events $\widetilde{E}(\tau)=k_A$ and $k_A=k_B$ are
independent. By robust completeness, the conditional probability
of agreement is negligibly close to $(1+\epsilon)/2$, except on
a negligible fraction of preparations. Consequently,
\[
\Pr[\widetilde{E}(\tau)=k_A \;\wedge\; k_A=k_B]
=
\left(\frac12+\frac{\epsilon}{2}\right)
\Pr[\widetilde{E}(\tau)=k_A]
\pm \operatorname{negl}(\lambda),
\]
and
$\Pr[k_A=k_B]=(1+\epsilon)/2\pm\operatorname{negl}(\lambda)$.
Dividing by this agreement probability gives
\[
\begin{aligned}
2\Pr[\widetilde{E}(\tau)=k_A \mid k_A=k_B]-1
&= \gamma \pm \operatorname{negl}(\lambda)\\
&\le 2\delta-\epsilon+\operatorname{negl}(\lambda),
\end{aligned}
\]
% as required.

% Now, in the case that the ToNC satisfies robust completeness, we have that the event $\widetilde{E}(\tau) = k_A$ is \emph{independent} of the event $k_A = k_B$, and thus \begin{align*}
% 2\Pr[\widetilde E(\tau)=k_A \mid k_A=k_B] - 1
% =
% 2\Pr[\widetilde E(\tau)=k_A ] - 1 =\gamma \leq 2\delta - \epsilon.
% \end{align*}
\end{proof}

\subsection{Weak (committed-bit) oblivious transfer}\label{subsec:WOT}

\begin{theorem}
     Assume the existence of an inefficiently-extractable commitment (\cref{def:com}) and a normal form $(\epsilon,\delta)$ test of non-commutation secure against adversaries with non-uniform quantum advice (\cref{def:ToNC}, \cref{remark:advantage}). Then, for constant parameters $\epsilon,\delta$ and any constant
$\delta' > 2\delta-\epsilon$, there exists an
$(\epsilon,\delta',0)$-weak committed-bit OT (Definition \ref{def:committed-bit})
\end{theorem}

\begin{proof}
    The protocol is given in Protocol \ref{prot:OT_from_anti_comm_test_classical_sender}. The theorem follows by combining \cref{lemma:correctness}, \cref{lemma:completeness}, \cref{lemma:rec-sec}, and \cref{lemma:sen-sec}. 
\end{proof}

\begin{protocolbox}{Protocol: Weak committed-bit OT}\label{prot:OT_from_anti_comm_test_classical_sender}
\textbf{Parties.} Sender $S$ and receiver $R$ (classical communication; parties may be quantum).\\
\textbf{Ingredients.} Inefficiently-extractable commitment $(\Com,\Rec,\Ver_\Com,\Ext_\Com)$ (\cref{def:com}) and normal-form $(\epsilon,\delta)$ test of non-commutation $(P_\prep,V_\prep,P_0,P_1,V_\out)$ (\cref{def:ToNC}).\\
%\textbf{Parameters.} Security parameter $\lambda$ and slack parameter $s = s(\lambda) = \secp^{1/3}$. \\
%\textbf{Notation.} Let $[r]=\{1,\dots,r\}$. For each instance $i\in[r]$, let $q_i\in\bits$ be $S$'s final-round challenge bit.\\
%\textbf{Abort.} Whenever $S$ aborts, it samples and outputs uniformly random $r_0,r_1 \gets \{0,1\}^2$.
\medskip

\noindent $\Setup\langle R(1^\secp),S(1^\secp)\rangle$:

\begin{itemize}
    \item \textbf{Challenge commitments.} For each $k \in [\secp]$, $i \in [2\secp]$, $R$ samples $x^{k}_i \gets \{0,1\}$ and the parties run \[\state^{k}_{\Com,i},\tau^{k}_i \gets \left\langle\Com\left(1^\secp,x^{k}_i\right),\Rec\left(1^\secp\right)\right\rangle,\] with $R$ playing the role of $\Com$ and $S$ playing the role of $\Rec$. 
    \item $R$ sets $\init_R \coloneqq \left\{\state^{k}_{\Com,i}\right\}_{k \in [\secp], i \in [2\secp]}$.
\end{itemize}

\noindent $\OT\langle R(b,\init_R),S(\init_S)\rangle$:

\begin{itemize}
    \item \textbf{Answer commitments.} For each $k \in [\secp]$, $i \in [2\secp]$, $R$ samples $y^{k}_i \gets \{0,1\}$ and the parties run \[{{\state'}^{k}_{\Com,i}},{{\tau'}^{k}_i} \gets \left\langle\Com\left(1^\secp,y^{k}_i\right),\Rec\left(1^\secp\right)\right\rangle,\] with $R$ playing the role of $\Com$ and $S$ playing the role of $\Rec$. 
    \item \textbf{Pair commitments.} For each $k \in [\secp]$ and unordered pair $\{i,j\}\subseteq[2\secp]$, $R$ sets $x^{k}_{i,j} \coloneqq x^{k}_i \oplus x^{k}_j$, and the parties run \[\state^{k}_{\Com,\{i,j\}},\tau^{k}_{\{i,j\}} \gets \left\langle\Com\left(1^\secp,x^{k}_{i,j}\right),\Rec\left(1^\secp\right)\right\rangle,\] with $R$ playing the role of $\Com$ and $S$ playing the role of $\Rec$.

    \item \textbf{Main loop.} $S$ samples $k^* \gets [\secp]$. For $k = 1$ to $k^*$:

    \begin{itemize}

  \item \textbf{Non-commutation tests.} For each $i \in [2\secp]$, the parties run \[\ket*{\psi^{k}_i},\left(c^{k}_i,a^{k}_i\right) \gets \langle P_\prep(1^\secp),V_\prep(1^\secp)\rangle,\] with $R$ playing the role of the prover and $S$ playing the role of the verifier. For each $i \in [2\secp]$, $S$ samples $z^{k}_i \gets \{0,1\}$ and sends $z^{k}_i$ to $R$. For each $i \in [2\secp]$, $R$ sets $c^{k}_{R,i}:=x^{k}_i\oplus z^{k}_i$, computes \[a^{k}_{R,i} \gets P_{c^{k}_{R,i}}\left(\ket*{\psi^{k}_i}\right)\] and sends $\widehat{a}^{k}_i \coloneqq a^{k}_{R,i} \oplus y^{k}_i$ to $S$.

  \item \textbf{Pairing and checks.}
  $S$ partitions $[2\secp]$ uniformly at random into a disjoint set of $\secp$ pairs $\mathcal{P}$, and sends $\mathcal{P}$ to $R$. $R$ checks that $\mathcal{P}$ is a partition of $[2\secp]$
into $\secp$ disjoint two-element sets; otherwise, it aborts. For each $\{i,j\} \in \mathcal{P}$, $R$ sends $\left(x^{k}_{i,j},\state^{k}_{\Com,\{i,j\}}\right)$ to $S$, who aborts if \[\Ver_\Com\left(\state^{k}_{\Com,\{i,j\}},\tau^{k}_{\{i,j\}},x^{k}_{i,j}\right) = \bot.\]
  $S$ chooses a pair $P^* = \{i^*,j^*\} \in \mathcal{P}$ uniformly at random conditioned on $x^{k}_{i^*,j^*} \oplus z^{k}_{i^*} \oplus z^{k}_{j^*} \neq c^{k}_{i^*} \oplus c^{k}_{j^*}$ (or aborts if no such pair exists), and sends $P^*$ to $R$. $R$ checks that $P^* \in \mathcal{P}$; otherwise, it aborts. For each $\{i,j\}\in \mathcal{P} \setminus \{P^*\}$:
  \begin{itemize}[leftmargin=1.2em]
    \item $R$ sends $x^{k}_i,\state^{k}_{\Com,i},x^{k}_j,\state^{k}_{\Com,j},y^{k}_i,{{\state'}^{k}_{\Com,i}},y^{k}_j,{{\state'}^{k}_{\Com,j}}$ to $S$, who aborts if any of the respective $\Ver_\Com$ checks fail. 
    \item $S$ aborts if $x^{k}_{i,j} \neq x^{k}_i \oplus x^{k}_j.$ Otherwise, $S$ sets \[c^{k}_{R,i} = x^{k}_i \oplus z^{k}_i, c^{k}_{R,j} = x^{k}_j \oplus z^{k}_j,  a^{k}_{R,i} = \widehat{a}^{k}_i \oplus y^{k}_i,  a^{k}_{R,j} = \widehat{a}^{k}_j \oplus y^{k}_j.\]  

\end{itemize}
 $S$ aborts if:
    \begin{itemize}[leftmargin=1.2em]
        \item For less than  $1/4$ fraction of $i \in [2\secp] \setminus P^*$, $c^{k}_{R,i} = c^{k}_i$.
        \item For less than $\frac{1}{2} + \frac{\epsilon}{2} - \secp^{-1/3}$ fraction of $i \in [2\secp] \setminus P^*$ such that $c^{k}_{R,i} = c^{k}_i$, $a^{k}_{R,i} = a^{k}_i$.
    \end{itemize}

  \end{itemize}

  \item \textbf{Final transfer.} Re-define 
  \begin{align*}
    &c_0 = c^{k^*}_{i^*}, c_1 = c^{k^*}_{j^*}, \quad c_{R,0} = c^{k^*}_{R,i^*}, c_{R,1} = c^{k^*}_{R,j^*} \\
    &a_0 = a^{k^*}_{i^*}, a_1 = a^{k^*}_{j^*}, \quad  a_{R,0} = a^{k^*}_{R,i^*}, a_{R,1} = a^{k^*}_{R,j^*} \\
    &x_0 = x^{k^*}_{i^*}, x_1 = x^{k^*}_{j*}, \quad  z_0 = z^{k^*}_{i^*},z_1 = z^{k^*}_{j^*} \\
    &\state_{\Com,0} = \state^{k^*}_{\Com,i^*},\state_{\Com,1} = \state^{k^*}_{\Com,j^*}, \quad \tau_0 = \tau^{k^*}_{i^*}, \tau_1 = \tau^{k^*}_{j^*}.
  \end{align*}
  
  \begin{itemize}[leftmargin=1.2em]
      \item $S$ sends $\{c_0,c_1\}$ to $R$.
      \item $R$ checks that $(c_0,c_1) \in \{0,1\}^2$ and
      $c_0 \oplus c_1 \neq c_{R,0} \oplus c_{R,1}$;
      otherwise, it aborts.
      \item Let $\widehat{b} \in \{0,1\}$ be the unique bit such that $c_{R,\widehat{b}} = c_{\widehat{b}}$. $R$ sends $b' = b \oplus \widehat{b}$ to $S$.
      \item $S$ outputs $(r_0 \coloneqq a_{b'},r_1 \coloneqq a_{1-b'})$ and $\state_S = (b',\tau_0,\tau_1,c_0,c_1,z_{0},z_{1})$.
      \item $R$ outputs $r \coloneqq  a_{R,\widehat{b}}$ and $\state_R = (x_0,x_1,\state_{\Com,0},\state_{\Com,1})$.
  \end{itemize}

\end{itemize}

$\Ver\langle R(\state_R),S(b,\state_S)\rangle$:
\begin{itemize}
    \item If $\Ver_\Com(\state_{S,0},\state_{R,0},x_0) = \bot$ or $\Ver_\Com(\state_{S,1},\state_{R,1},x_1) = \bot$, output $\bot$.
    \item Let $c_{R,0} = x_0 \oplus z_0$, $c_{R,1} = x_1 \oplus z_1$. If $c_0 \oplus c_1 = c_{R,0} \oplus c_{R,1}$, output $\bot$.
    \item Let $\widehat{b}$ be the unique bit such that $c_{R,\widehat{b}} = c_{\widehat{b}}$. Output $\top$ iff $b = \widehat{b} \oplus b'$.
\end{itemize}

$\TrapGen(\tau_\Setup)$:
\begin{itemize}
    \item For each $k \in [\secp], i \in [2\secp]$, run $x^k_i \coloneqq \Ext_\Com\left(\tau^k_i\right)$. Set $\td \coloneqq \left\{x^k_i\right\}_{i \in [2\secp]}$.

\end{itemize}

$\Ext(\td,\tau)$:
\begin{itemize}
    \item $\tau$ includes $z_0,z_1,c_0,c_1,b'$ and $\td$ includes $x_0 = x_{i^*}$ and $x_1 = x_{j^*}$.
    \item Let $c_{R,0} = x_0 \oplus z_0$, $c_{R,1} = x_1 \oplus z_1$. If $c_0 \oplus c_1 = c_{R,0} \oplus c_{R,1}$, output $0$.
    \item Let $\widehat{b}$ be the unique bit such that $c_{R,\widehat{b}} = c_{\widehat{b}}$. Output $\widehat{b} \oplus b'$.
\end{itemize}

\end{protocolbox}

\paragraph{Abort convention.}
If the $\OT$ phase in Protocol \ref{prot:OT_from_anti_comm_test_classical_sender} aborts, an abort flag is recorded in the transcript.
The honest sender outputs fresh independent uniform bits $(r_0,r_1)$,
independent of the preceding execution, and sets $\mathsf{st}_S=\bot$.
The honest receiver outputs an independent uniform bit $r$ and sets
$\mathsf{st}_R=\bot$.
On an aborted transcript, $\mathsf{Ext}$ returns $0$ without parsing
any remaining fields. The verification protocol $\mathsf{Ver}$
immediately rejects if $\mathsf{st}_S=\bot$.

\begin{lemma}[Correctness]\label{lemma:correctness}
    Protocol \ref{prot:OT_from_anti_comm_test_classical_sender} satisfies Correctness according to \cref{def:committed-bit} with $\epsilon = 2\beta - 1$.
\end{lemma}

\begin{proof}
    By inspection, and the completeness of the ToNC (\cref{def:ToNC}), we have that $r = r_b$ with probability $\beta = \frac{1}{2} + \frac{\epsilon}{2}$, \emph{conditioned on} $S$ not aborting during one of the iterations of the main loop. So, it remains to show that $S$ aborts with $\negl(\secp)$ probability. We will show that $S$ aborts in an iteration for any $k \in [\secp]$ with $\negl(\secp)$ probability, which establishes the above by a union bound.
    
    First, consider the check that for $1/4$ fraction of $i \in [2\secp] \setminus P^*$, $c^k_{R,i} = c^k_i$. Since the bits $c^k_{R,i}$ are masked by the $z^k_i$, which are sampled uniform and independent of the rest of the protocol, this check fails with probability bounded by the following Hoeffding inequality: \[\Pr[\frac{\#\{i: c^k_{R,i} \neq c^k_i\}_{i \in [2\secp] \setminus P^*}}{2\secp-2} < \frac{1}{4}] \leq \exp\left(-2\frac{2\secp-2}{16}\right) = \negl(\secp).\]

    Next, consider the check that for less than $\frac{1}{2} + \frac{\epsilon}{2} - \secp^{-1/3}$ fraction of $i \in [2\secp] \setminus P^*$ such that $c^k_{R,i} = c^k_i$, $a^k_{R,i} = a^k_i$. Conditioned on the above, we have that the number of such $i$ is at least $(2\secp-2)/4$. Due to the completeness of ToNC, each event  $a^k_{R,i} = a^k_i$ is an independent Bernoulli random variable with success $\frac{1}{2} + \frac{\epsilon}{2}$. Thus, this check fails with probability bounded by the following Hoeffding inequality:
    \[\Pr[\frac{\#\{i: c^k_{R,i} = c^k_i \wedge a^k_{R,i} = a^k_i\}}{\#\{i: c^k_{R,i} = c^k_i\}} < \frac{1}{2} + \frac{\epsilon}{2} - \secp^{-1/3}] \leq \exp\left(-2\frac{2\secp-2}{4\secp^{2/3}}\right) = \negl(\secp).\]

    It is easy to check that all other abort conditions happen with $\negl(\secp)$ probability, so we can conclude that $r = r_b$ with probability at least $\frac{1}{2} + \frac{\epsilon}{2} - \negl(\secp)$.
\end{proof}

\begin{lemma}[Completeness]\label{lemma:completeness}
    Protocol \ref{prot:OT_from_anti_comm_test_classical_sender} satisfies Completeness according to \cref{def:committed-bit}.
\end{lemma}

\begin{proof}
    This follows directly from the correctness of the commitment scheme (\cref{def:com}).
\end{proof}

\begin{lemma}[Receiver security]\label{lemma:rec-sec}
    Protocol \ref{prot:OT_from_anti_comm_test_classical_sender} satisfies Receiver security according to \cref{def:committed-bit} with $\gamma = 0$ .
\end{lemma}

\begin{proof}
    This follows from the hiding of the commitment scheme (\cref{def:com}). Indeed, the only information that $S$ receives about $b$ is masked by $\widehat{b}$, which is determined by $x_{i^*},x_{j^*}$ as well as public information in the transcript. $S$ receives $x_{i^*} \oplus x_{j^*}$, which fixes $x_{i^*},x_{j^*}$ to either $\{00,11\}$ or $\{01,10\}$. In either case, $\widehat{b}$ (along with other public information) is determined by which of the two strings was committed to.   
    By a standard hybrid argument, commitments to 00 and 11 can only be distinguished with $\negl(\secp)$ advantage and likewise for commitments to 01 and 10. This completes the proof. 
\end{proof}

\begin{lemma}[Sender security]\label{lemma:sen-sec}
    Protocol \ref{prot:OT_from_anti_comm_test_classical_sender} satisfies Sender security according to \cref{def:committed-bit} for any constant $\delta' > 2\delta - \epsilon$.
\end{lemma}

\begin{proof}

    First, we prove the ``binding'' statement. Fix an unbounded setup adversary $\widetilde R_{\mathrm{init}}$. We show that with probability $1-\operatorname{negl}(\lambda)$ over
$\tau_{\mathrm{Setup}}$, no continuation can open a challenge commitment
to a bit different from its extracted value with more than negligible
probability. To obtain this conditional statement, maximize the
probability of an incorrect opening for each commitment transcript,
apply binding against unbounded committers, and then apply Markov's
inequality and a union bound over the setup commitments.
Sender security for committed-bit OT has two requirements. For these setup transcripts, the second sender-security requirement
follows from binding, so it remains to prove the first requirement.

Next, we prove the ``privacy'' statement. Fix such a transcript $\tau_{\mathrm{Setup}}$ and a receiver state
$\mathrm{init}_R$ after setup. All probabilities below refer to the
protocol starting from these fixed values; the sender's subsequent
random choices are sampled as in the protocol. Suppose that a QPT
receiver $\widetilde R$ has guessing advantage at least $\delta'$ in
this execution. We construct a non-uniform commuting ToNC prover
with advantage at least $\frac12(\epsilon+\delta')-o(1)$.
$\tau_{\mathrm{Setup}}$ and
$\mathrm{init}_R$  may depend on $\lambda$. Thus, if such a
receiver and fixed values existed for infinitely many $\lambda$,
the inequality $\delta'>2\delta-\epsilon$ would contradict ToNC
soundness against a receiver with non-uniform quantum advice.

    % Sender security for committed-bit OT has two requirements. The second follows directly from the binding of the commitment scheme (\cref{def:com}), so it remains to establish the first.

    % We exhibit a reduction that takes any $(\widetilde{R}_\init,\widetilde{R})$ that has $\delta'$ advantage in the first part of sender security and produces a non-uniform commuting prover that has $\frac{1}{2}(\epsilon + \delta') - o(1)$ advantage in the soundness game of the ToNC. When $\delta' = 2\delta - \epsilon + \Omega(1)$, this yields an adversary with advantage $\delta + \Omega(1)$, a contradiction.
    
    Our non-uniform ToNC prover $(P_\init,P_\prep,P_0,P_1)$ is specified as follows, where $P_\init$ is a potentially-inefficient operation that runs prior to the beginning of the interaction with the verifier, and outputs a quantum state. Note that this model of non-uniform adversary with quantum advice is identical to the model where we allow $P_\prep$ to be initialized with a (potentially inefficiently-computable) mixed state $\rho$.  \\

    \noindent \underline{$P_\init$}:
    \begin{itemize}
        \item Run the following loop until it succeeds:
        \begin{itemize}
            \item Initialize a fresh copy of $\widetilde R$ in the fixed state
$\mathrm{init}_R$, and initialize the sender with the fixed transcript
$\tau_{\mathrm{Setup}}$. Set
$\mathrm{td}:=\mathrm{TrapGen}(\tau_{\mathrm{Setup}})$ and parse
$\mathrm{td}$ as $\{x_i^k\}_{k\in[\lambda],\,i\in[2\lambda]}$.
Fresh copies of the fixed receiver state are used only in this
potentially inefficient initialization; its final output state is
supplied as quantum advice before the external ToNC interaction.
            \item Interact with $\widetilde{R}$ during the ``answer commitment'' phase of $\OT$.
            \item For each $k \in [\secp], i \in [2\secp]$, run $y^k_i \coloneqq \Ext_\Com(\tau_i')$.
            \item Interact with $\widetilde{R}$ during the
            ``pair commitment'', and ``main loop'' phases of $\OT$ up until the beginning of the $k^*$'th loop.
            \item If there has not been an abort yet, then declare success and move on, otherwise try the loop again.
        \end{itemize}

    \end{itemize}

    \noindent \underline{$P_\prep$}:
    \begin{itemize}
        \item Sample $i^* \gets [2\secp]$. For each $i \in [2\secp]$, interact with $\widetilde{R}$ as $V_\prep$ in the protocol \[\ket{\psi_i},\left(c^{k^*}_i,a^{k^*}_i\right) \gets \langle \widetilde{R},V_\prep(1^\secp)\rangle,\] except that for the $i^*$'th instance, $P$ forwards the message to its own external challenger. Note that $P$ knows $\left\{c^{k^*}_i,a^{k^*}_i\right\}_{i \neq i^*}$ but not $c^{k^*}_{i^*}, a^{k^*}_{i^*}$.
        \item Sample $z^{k^*} \gets \{0,1\}^{2\secp}$, send $z^{k^*}$ to $\widetilde{R}$, and receive $\left\{\widehat{a}^{k^*}_{i}\right\}_{i \in [2\secp]}$. Set $c^{k^*}_{R,i^*} \coloneqq x^{k^*}_{i^*} \oplus z^{k^*}_{i^*}$ and $a^{k^*}_{R,i^*} \coloneqq \widehat{a}^{k^*}_{i^*} \oplus y^{k^*}_{i^*}$.
    \end{itemize}

    \noindent \underline{$P_{c_{R,i^*}}$}
    \begin{itemize}
        \item In this case, the external challenge $c = c^{k^*}_{R,i^*}$. 
        \item Output $a' \coloneqq a^{k^*}_{R,i^*}$.
    \end{itemize}

    \noindent \underline{$P_{1-c_{R,i^*}}$}
    \begin{itemize}
        \item In this case, the external challenge $c = 1-c^{k^*}_{R,i^*}$. 
        \item In what follows, if there is an abort, then output a uniformly random bit.

        \item Set $c^{k^*}_{i^*}:=c$ and, for every $i\in[2\lambda]$,
set $c^{k^*}_{R,i}:=x^{k^*}_i\oplus z^{k^*}_i$. Let
\[
    U:=\{i\in[2\lambda]:c^{k^*}_{R,i}\neq c^{k^*}_i\}
\]
be the set of mismatching indices; in this branch, $i^*\in U$.

Sample a uniformly random pairing $\mathcal Q$ of $[2\lambda]$.
Among its pairs with exactly one endpoint in $U$, choose
$\{u,v\}$ uniformly, labeling the endpoints so that
$u\in U$ and $v\notin U$. Abort if no such pair exists.

Let $\pi$ swap $u$ and $i^*$ and fix all other indices
(or be the identity if $u=i^*$). Set
\[
    \mathcal P
    :=\bigl\{\{\pi(i),\pi(j)\}:\{i,j\}\in\mathcal Q\bigr\},
    \qquad
    P^*:=\{i^*,v\},
    \qquad
    j^*:=v.
\]
Only the pairing is relabeled: all stored bits, transcripts,
and quantum registers retain their original indices.

\item Send the full pairing $\mathcal P$ to $\widetilde R$.
Receive and verify the pair-commitment openings for every
pair in $\mathcal P$, including $P^*$, exactly as $S$ does.
Then send $P^*$ to $\widetilde R$, receive the individual
challenge- and answer-commitment openings for every pair
in $\mathcal P\setminus\{P^*\}$, and perform all remaining
checks prescribed by $S$. Abort if any check fails.

        % \item Set $c^{k^*}_{i^*} \coloneqq c$ and for $i \in [2\secp]$, set $c^{k^*}_{R,i} \coloneqq x^{k^*}_i \oplus z^{k^*}_i$. Sample $j^*$ uniformly at random such that $c^{k^*}_{R,i^*} \oplus c^{k^*}_{R,j^*} \neq c^{k^*}_{i^*} \oplus c^{k^*}_{j^*}$, meaning $c^{k^*}_{R,j^*} = c^{k^*}_{j^*}$, and set $P^* \coloneqq \{i^*,j^*\}$. Partition the rest of $[2\secp] \setminus P^*$ uniformly at random to form $\mathcal{P}$.
        % \item Send $\mathcal{P}$ to $\widetilde{R}$, receive $\left\{x^{k^*}_{i,j}\right\}_{\{i,j\}\in \mathcal{P}}$, and abort if \[\Ver_\Com\left(\state^{k^*}_{\Com,\{i,j\}},\tau^{k^*}_{\{i,j\}},x_{i,j}\right) = \bot\] for any $\{i,j\} \in \mathcal{P}$. Send $P^*$ to $\widetilde{R}$ and receive \[x^{k^*}_i,\state^{k^*}_{\Com,i},x^{k^*}_j,\state^{k^*}_{\Com,j},y^{k^*}_i,{\state'}^{k^*}_{\Com,i},y^{k^*}_j,{\state'}^{k^*}_{\Com,j}\] for each $\{i,j\} \in \mathcal{P} \setminus P^*$. Perform the remainder of the abort checks implemented by $S$.
        \item Send $\left\{c^{k^*}_{i^*},c^{k^*}_{j^*}\right\}$ and $r_{b} \coloneqq a^{k^*}_{j^*}$ to $\widetilde{R}$, and let $r_{1-b}$ be $\widetilde{R}'s$ final output. Output $a' \coloneqq r_{1-b}$. 
    \end{itemize}

    Note that the operations $P_0$ and $P_1$ are commuting by definition, since one of them directly outputs a classical bit stored in memory. We show that this prover guesses $a^*$ with advantage at least
$\frac12(\epsilon+\delta')-o(1)$.

\begin{itemize}
\item Let $\mathcal S$ denote the event that no abort occurs
before iteration $k^*$, and let $q:=\Pr[\mathcal S]$.
Since aborted executions contribute zero advantage to the receiver, the receiver's advantage
from the fixed setup, before conditioning on $\mathcal S$, equals
$q$ times its advantage conditioned on $\mathcal S$. Hence $q\geq\delta'>0$, and
its conditional advantage is at least $\delta'$.
The initialization procedure samples conditioned on $S$.

If the main loop is continued until its first abort or until
all $\lambda$ iterations finish, its first abort iteration
is independent of the hidden uniform choice $k^*$.
It therefore equals $k^*$ with probability at most $1/\lambda$,
giving
\[
    \Pr[\text{abort in iteration }k^*\mid\mathcal S] =\frac{\Pr[\text{abort in iteration }k^*]}{\Pr[S]}
    \leq \frac{1}{\lambda q}=o(1).
\]
Throughout the remaining analysis, we condition on $\mathcal S$
and suppress the superscript $k^*$.

\item First consider a modified real experiment that selects
$P^*$ using the parities $x_i\oplus x_j$ of the extracted
individual bits, rather than the committed pair parities,
while keeping all opening checks unchanged.
Thus a pair is eligible precisely when it has one matching
and one mismatching endpoint.
Its selection may be made privately before the pair commitment openings.

Let $H$ be the complete classical record before choosing a pairing,
including the verifier's private bits but excluding $i^*$,
and write
\[
    U:=\{i:c_{R,i}\neq c_i\},
    \qquad
    B:=\{c\neq c_{R,i^*}\}.
\]
Before conditioning on $B$, the uniform index $i^*$
is independent of $H$ and the receiver's state.
Conditioned on $H$ and $B$, it is uniform in $U$.

For fixed $H$, the law of a uniform pairing with a uniform
eligible selected pair is invariant under permutations of $U$.
The swap sampler preserves this law after averaging over $i^*$:
for any successful output $(\mathcal P,P^*)$ with mismatching
endpoint $w$, we must have $i^*=w$, contributing probability
$1/|U|$. There are exactly $|U|$ possible pre-swap outcomes,
one for each original endpoint $u\in U$, and each has the
same probability as $(\mathcal P,P^*)$ under the modified
real sampler. These factors cancel. The probability of
having no eligible pair is unchanged as well.

It remains only to account for conditioning on $B$.
The fresh masks give $\Pr[B]=1/2$ and
$|U|\sim\operatorname{Bin}(2\lambda,1/2)$.
Since $\Pr[B\mid H]=|U|/(2\lambda)$, conditioning on $B$
reweights each record $H$ by $|U|/\lambda$.
The resulting statistical distance is
\[
    \frac12\,\mathbb E\left[
        \left|\frac{|U|}{\lambda}-1\right|
    \right]
    \leq
    \frac{\sqrt{\lambda/2}}{2\lambda}
    =
    \frac{1}{2\sqrt{2\lambda}}.
\]
For the same record and selected pairing, the receiver's
conditional state and subsequent interaction are identical.
On non-aborted runs, $i^*$ is the selected pair's mismatching
endpoint. Thus the terminal guessing probability differs
from that in the modified real experiment by only
$O(\lambda^{-1/2})$, with aborted runs assigned a fair guess.

\item We now compare the modified real experiment with the
actual one. Call a pair inconsistent if its extracted
committed parity differs from $x_i\oplus x_j$; these extracted
parities are used only in the analysis and are fixed before
the fresh masks.

Except with negligible binding error, two or more inconsistent
pairs in a pairing force both experiments to abort, since
every pair other than $P^*$ is checked. With no inconsistent
pair, the experiments agree. With exactly one, their eligible
sets differ only in that pair. Before opening checks, the
fresh masks make each of the other $\lambda-1$ pairs eligible
independently with probability $1/2$. Thus there are
$\Omega(\lambda)$ other eligible pairs except with negligible
probability, and adding or removing one candidate changes
the uniform selection by only $O(1/\lambda)$.
Consequently, the terminal guessing probabilities differ by
$O(1/\lambda)+\operatorname{negl}(\lambda)$.

\item Each challenge branch has probability $1/2$.
In the matching branch, imagine continuing with the real
sender's pairing and checks. This hypothetical continuation
aborts with probability at most $2/(\lambda q)=o(1)$.
For a fixed pre-pairing record, $i^*$ is uniform among the
matching indices, and this continuation does not depend on
$i^*$. On passing the checks, at least a
$1/2+\epsilon/2-\lambda^{-1/3}$ fraction of the checked
matching indices have correct answers. There are
$\Omega(\lambda)$ such indices, so including the at most
two unchecked indices changes this fraction by only
$O(1/\lambda)$. Binding therefore gives advantage
$\epsilon-o(1)$ in this branch.

In the mismatching branch, the preceding comparisons transfer
the receiver's advantage conditioned on $\mathcal S$ with
only an $o(1)$ loss. The supplied bit is $a_{j^*}$, and the
bit to be guessed is $a_{i^*}=a^*$. This branch therefore
has advantage at least $\delta'-o(1)$.

Averaging the two branches gives
\[
    \operatorname{Adv}(P)
    \geq \frac12(\epsilon+\delta')-o(1),
\]
as required.
\end{itemize}

\end{proof}

\ifsubmission\else\section{Hardness amplification}

\subsection{Post-quantum hard-core measure theorem}\label{subsec:hardcore}

Fix a deterministic predicate $P:\{0,1\}^n\rightarrow \bits$ and any randomized classical circuit $C(x,b)$ that is intended to output $1$ when $b=P(x)$ and $0$ otherwise. Theorem~3 in \cite{HolensteinPuzzles} (which is similar in spirit to Lemma 2.4 from \cite{Holenstein}) shows that one can
\emph{decompose} the input space relative to this particular $C$ into a large
``hard'' subset of the domain $S\subseteq \bits^n$ of probability mass $\delta$ and its complement. That is, on $S$, $C$'s acceptance probability is
essentially \emph{insensitive} to whether $b=P(x)$ or $b=1-P(x)$, meaning that flipping
$P$ on $S$ is nearly indistinguishable to $C$. Moreover,
one can convert $C$ into an explicit predictor $Q$ that recovers $P(x)$ with
an overall success at least $1-\delta/2$ on the entire domain. 

We prove a variation of \cite[Theorem 3]{HolensteinPuzzles} which differs in the following respects.
\begin{itemize}
    \item \textbf{Post-quantum setting.} We allow the distinguisher $C$ to be a \emph{quantum} circuit, and thus the resulting predictor $Q$ is quantum.
    \item \textbf{Arbitrary input distributions.} Rather than working only with
    the uniform distribution over $\bits^n$, we formulate the statement with
    respect to an arbitrary distribution $\mathcal{D}$ over inputs.

    \item \textbf{A measure rather than a hard set.} As a consequence of allowing a quantum distinguisher $C$, we construct a quantum predictor which maps a pair $(x,P(x))$ \emph{probabilistically} to a bit. Since the predictor is non-deterministic, it does not immediately give rise to a hard subset. Instead, its prediction statistics constitute a 
    \emph{sub-distribution measure} of $\mathcal{D}$ (as in Definition \ref{def:measure})
    that designates the portion of the domain on which $C$ is
    insensitive to flipping $P$. The measure is explicitly constructed in Corollary \ref{cor:soft-hardcore-random-labels}.
\end{itemize}

The measure vs. set discrepancy comes from the fact that the proof in \cite{HolensteinPuzzles} relies on derandomizing the distinguisher $C$ through fixing a set $R$ of random seeds. Evaluations of $C(x,r)$ on any $r\in R$ are used to compute a criterion which determines whether $x\in S$---thus $S$ becomes a deterministic set. Such derandomization is impossible in our setting, where $C$ is no longer classical. However, shifting to working with a measure enables recovering analogous results in this setting as well. 

Next, we define a function that will be used in the Theorem's statement.

\paragraph{Definition.}
For parameters $\tau\in\mathbb R$ and $\xi>0$, define
\[
\Soft_{\tau,\xi}(t)\ :=\
\begin{cases}
1, & t\le \tau-\xi,\\[1mm]
\frac{\tau-t}{\xi}, & \tau-\xi < t < \tau,\\[1mm]
0, & t\ge \tau.
\end{cases}
\]
$\Soft_{\tau,\xi}(t)\in[0,1]$ and it is $1/\xi$-Lipschitz in $t$ and also $1/\xi$-Lipschitz in $\tau$
(for each fixed $t$).

$\Soft_{\tau,\xi}(t)$ can be thought of as a ``soft indicator'' of being in the hard region: it is $1$ when $t$ is clearly
below $\tau$, it is $0$ when $t$ is above $\tau$, and it interpolates linearly in a narrow boundary band.

\begin{theorem}[A generalization of Theorem 3 in \cite{HolensteinPuzzles}]\label{thm:Q3}
Let there be:
\begin{itemize}
    \item A function $n:\N\rightarrow \N$, such that $n(\secp)=\poly(\secp)$.
    \item A family of predicates $\{P_\secp\}_{\secp\in\N}$, where $P_\secp:\bits^{n(\secp)}\to\bits$.
    \item A family of distributions $\{\mathcal D_\secp\}_{\secp \in \N}$, where $\mathcal D_\secp$ is over $\bits^{n(\secp)}$, such that the following min-entropy condition holds:
    for $m(\secp):=\max\{D_\secp(z) \mid z\in \bits^{n(\secp)}\}$, it holds that $m(\secp)\le \negl(\secp)$.
    
    \item A family of samplers $\{\mathsf{Samp}_\secp\}_{\secp \in \N}$ such that $\mathsf{Samp}_\secp$ outputs $(x,P_\secp(x))$ where $x$ is sampled from $D_\secp$. $\mathsf{Samp}$ runs in time $\poly(\secp)$.
\end{itemize}

Consider any function $\varepsilon:\mathbb{N}\to(0,1)$ satisfying
$\varepsilon(\secp)=1/\poly(\secp)$.
Then there exists a \emph{uniform quantum polynomial-time oracle algorithm}
$\mathsf{Gen}$ such that for every security parameter value $\secp$ and every quantum circuit
\[
\mathcal C:\{0,1\}^{\,n(\secp)+1}\to\{0,1\},
\]
on input $(\varepsilon(\secp),1^\secp)$ and with oracle access to $\mathcal C$,
$\mathsf{Gen}$ outputs:

 \begin{itemize}
\item A parameter $\xi>0$ and a threshold $\tau\in[0,1+\xi]$.
\item Quantum circuits $S,Q$ with oracle access to $\mathcal C$,
\item a parameter $\delta\in[0,1]$,
\item A sign $\sigma\in\{-1,+1\}$.
\end{itemize}
such that except with probability $\negl(\secp)$ over measurements of $\mathsf{Gen}^{\mathcal C}$,
the following properties hold\footnote{Fix $\secp$ and suppress the subscript $\secp$.}.

\noindent For $x\in\bits^{n(\secp)}$ and $b\in\bits$ define the \emph{pointwise gaps}
\[
p_x(b)\ :=\ \Pr[\mathcal C(x,b)=1],
\qquad
\Delta_x\ :=\ p_x(P(x)) - p_x(1-P(x)),
\]
and define the \emph{global gap}
\[
\Delta\ :=\ \Pr_{x\leftarrow \mathcal D}\bigl[\mathcal C(x,P(x))=1\bigr]\ -\ \Pr_{x\leftarrow \mathcal D}\bigl[\mathcal C(x,1-P(x))=1\bigr]
\ =\ \E_{x\leftarrow\mathcal D}[\Delta_x].
\]

\noindent Additionally, denote
\[
    w_\tau(x):=\operatorname{Soft}_{\tau,\xi}(\sigma\Delta_x),
    \qquad
    \delta_\tau:=\mathbb E_{x\gets\mathcal D}[w_\tau(x)].
\]

\begin{enumerate}
\item \textbf{(Correct mass).}
For every fixed $x$,
\[
\Pr[S(x,P(x))=1]\ =\ w_\tau(x)\ \pm\ (\varepsilon/200+ \negl(\secp)).
\]

Moreover, $\mathsf{Gen}$'s output $\delta$ satisfies
$\delta_\tau - \epsilon/200 \leq \delta\ \le\ \delta_\tau.$

\item \textbf{(Indistinguishability).}
Let $P'(x)$ be the randomized predicate $P'(x):=P(x)\oplus S(x,P(x))$, then
\[
\left|
\Pr_{x\leftarrow\mathcal D}\bigl[\mathcal C(x,P(x))=1\bigr]
-
\Pr_{x\leftarrow\mathcal D}\bigl[\mathcal C(x,P'(x))=1\bigr]
\right|
\ \le\ \varepsilon.
\]

\item \textbf{(Predictability).}
The circuit $Q$ satisfies
\[
\Pr_{x\leftarrow \mathcal D}\bigl[Q(x)=P(x)\bigr]\ \ge\ 1-\frac{\delta}{2}-\negl(\secp).
\]

\item \textbf{(Obliviousness).}

The circuit $Q(x)$ can be written as $Q^{C(x,\cdot),C(\cdot,\cdot)}$. 
That is, the underlying oracle circuit $Q$ is fixed and independent of $x$, and it utilizes $x$ only within a fixed, $x$-independent procedure of oracle queries to $C(x,\cdot)$. 
In this sense, $Q$ is oblivious to $x$.

% It forwards its input $x$ only to its first oracle (via queries of the form $(x,b)$ for $b \in \{0,1\}$); all other oracle queries, as well as the circuit structure and query schedule, are independent of $x$. Only the oracle responses—and thus the resulting wire values—depend on $x$.

% The circuit $Q(x)$ can be written as $Q^{C(x,\cdot),C(\cdot,\cdot)}$. That is, $Q$ is oblivious to its input $x$--

% $Q$ uses $x$ only by forwarding it unchanged to its oracle gates: there exists a fixed (i.e., $x$-independent) oracle-query procedure such that on input $x$, $Q$ only ever queries $\mathcal C$ on inputs of the form $(x,0)$ and $(x,1)$ with a query schedule independent of $x$, and the remaining computation is independent of $x$.\james{I think I know what you mean, but this could also be better clarified}

\end{enumerate}

There exist fixed polynomials $q(\cdot,\cdot)$ and $r(\cdot,\cdot)$ such that the oracle circuits $S,Q$:
\begin{itemize}
  \item make at most $q(\secp,1/\varepsilon)$ oracle calls to $\mathcal C$, and
  \item the remaining circuit has size at most $r(\secp,1/\varepsilon)$.
\end{itemize}
Consequently, after hardwiring $\mathcal C$ and replacing each oracle gate by a copy of $\mathcal C$,
the compiled circuit $\widehat Q$ satisfies
\[
|\widehat Q|\ \le\ q(\secp,1/\varepsilon)\cdot|\mathcal C| \;+\; r(\secp,1/\varepsilon),
\]
and similarly for $\widehat S$. We choose $q,r$ large enough that the same bounds also cover
running $\operatorname{Gen}$ and then evaluating its output
circuit $Q$.

\end{theorem}

\begin{proof}
Fix $\secp$. All probabilities over $\mathcal C,S,Q$ are over their internal measurements. Fix $\nu \le \negl(\secp)$.

Given $(x,b)$ and accuracy $\alpha>0$, using Chernoff's bound, one can estimate $p_x(b)$ to within $\pm \alpha$ with failure probability at most $\nu$ by running $\mathcal C(x,b)$ independently
\[
T=\Theta(\log(1/\nu)/\alpha^2)
\]
times and taking the empirical mean. The same holds for the gap
\[
\Delta_x^{(b)}:=p_x(b)-p_x(1-b),
\]
up to constant factors.

\paragraph{Normalizing the global gap.}
Initialize $\sigma:=1$ and include it in every output of
$\mathsf{Gen}$. Using $\poly(\secp,1/\varepsilon)$ samples $(x,P(x))\leftarrow \mathsf{Samp}$ and a constant number of calls to $\mathcal C$ per sample, estimate $\Delta$ to within $\pm \varepsilon/20$ except with probability $\negl(\secp)$. If the estimate $\widehat\Delta$ satisfies $|\widehat\Delta|\le \varepsilon/10$, output $\delta:=1$, $\xi:=\varepsilon^2/1000$, $\tau:=1+\xi$, let $S$ always output $1$, and let $Q$ output a uniform bit. All four conclusions are immediate. Otherwise, if $\widehat{\Delta}<-\varepsilon/10$, set
$\sigma:=-1$ and replace $\mathcal C$ by $1-\mathcal C$.
In the rest of the proof, $\mathcal C$ denotes the circuit after
this possible replacement. Henceforth assume
\[
\Delta\ge \varepsilon/20.
\]

\paragraph{Defining a function $F$.}
Fix
\[
\xi := \varepsilon^2/1000.
\]
For $\tau\in\mathbb R$, recall
\[
w_\tau(x):=\Soft_{\tau,\xi}(\Delta_x)
\]
and define the associated function
\[
F(\tau)
:= \mathbb E_{x\leftarrow\mathcal D}\bigl[\Delta_x\cdot w_\tau(x)\bigr].
\]

\noindent We will apply the intermediate value theorem on the interval $[0,1+\xi]$.
First,
\[F(0)\le 0,\]
since for any $x$ such that $w_0(x)\ne 0$, necessarily $\Delta_x<0$. Second, for all $\Delta_x\in[-1,1]$ we have $\Soft_{1+\xi,\xi}(\Delta_x)=1$, hence
\[
F(1+\xi)=\mathbb E_{x\leftarrow\mathcal D}[\Delta_x]=\Delta.
\]
The map $\tau\mapsto F(\tau)$ is continuous, and by Step~0 we have $\Delta\ge \varepsilon/20$. Therefore there exists $\tau^\star\in[0,1+\xi]$ such that
\begin{equation}\label{eq:target_window}
F(\tau^\star)=\frac{\varepsilon}{25}.
\end{equation}

\paragraph{A mean-correct estimator for $F(\tau)$.}
Fix
\[
\alpha:=\xi\varepsilon/800.
\]
For each sampled $(x,P(x))$, compute two \emph{independent} estimates of $\Delta_x$:
\[
\widehat\Delta_x^{(A)}
\quad\text{and}\quad
\widehat\Delta_x^{(B)},
\]
each satisfying $|\widehat\Delta_x^{(\cdot)}-\Delta_x|\le \alpha$, except with probability at most $\nu$.
Define
\[
Y_\tau(x)
:= \widehat\Delta_x^{(B)}\cdot \Soft_{\tau,\xi}\bigl(\widehat\Delta_x^{(A)}\bigr)
\in [-1,1].
\]
Let $X\sim\mathcal D$ and define $\mu_\tau:=\mathbb E[Y_\tau(X)]$.

\begin{lemma}[Mean correctness]\label{lem:mean_correct}
For every $\tau\in\mathbb R$,
\[
|\mu_\tau - F(\tau)| \le \alpha/\xi + \nu.
\]
\end{lemma}

\begin{proof}
Condition on $X=x$. Because $\widehat\Delta_x^{(B)}$ is an unbiased estimator of $\Delta_x$ conditioned on $x$ and is independent of $\widehat\Delta_x^{(A)}$, we have
\[
\mathbb E[Y_\tau(X)\mid X=x]
= \Delta_x\cdot
\mathbb E\!\left[\Soft_{\tau,\xi}\bigl(\widehat\Delta_x^{(A)}\bigr)\,\middle|\,X=x\right].
\]
The function $t\mapsto \Soft_{\tau,\xi}(t)$ is $1/\xi$-Lipschitz, hence on the event
$|\widehat\Delta_x^{(A)}-\Delta_x|\le \alpha$ we have
\[
\left|\Soft_{\tau,\xi}(\widehat\Delta_x^{(A)})-\Soft_{\tau,\xi}(\Delta_x)\right|
\le \alpha/\xi.
\]
This event fails with probability at most $\nu$, and $\Soft_{\tau,\xi}(\cdot)\in[0,1]$. Therefore
\[
\left|
\mathbb E\!\left[\Soft_{\tau,\xi}(\widehat\Delta_x^{(A)})\mid X=x\right]
-\Soft_{\tau,\xi}(\Delta_x)
\right|
\le \alpha/\xi + \nu.
\]
Multiplying by $|\Delta_x|\le 1$ and averaging over $x\sim\mathcal D$ yields the claim.
\end{proof}

\paragraph{Choosing $\tau$.}
Choose a grid step
\[
\rho:=\xi\varepsilon/800
\]
and define a grid
\[
\mathcal T
:= \{j\rho : j=0,1,\ldots,N\}\cup \{1+\xi\},
\qquad
N:=\left\lfloor\frac{1+\xi}{\rho}\right\rfloor,
\]
so that $\mathcal T\subseteq[0,1+\xi]$. For each $\tau\in\mathcal T$, estimate $\mu_\tau$ by sampling $m= \poly(\secp,1/\varepsilon)$  independent points from $\mathsf{Samp}$, computing $Y_\tau(\cdot)$ for each point independently, and averaging:
\[
\widehat F(\tau)
:= \frac1m\sum_{i=1}^m Y_\tau(x_i).
\]

\noindent Output the smallest $\tau\in\mathcal T$ such that
\begin{equation}\label{eq:select_rule}
\widehat F(\tau)\in\Bigl[\frac{\varepsilon}{100},\frac{\varepsilon}{20}\Bigr].
\end{equation}

Hoeffding's bound (since each $Y_\tau\in[-1,1]$) implies that there exists $m= \poly(\secp,1/\varepsilon)$ such that
$|\widehat F(\tau)-\mu_\tau|\le \varepsilon/800$ except with probability $\nu$. Union bounding over $\tau\in\mathcal T$, and using Lemma \ref{lem:mean_correct}, we obtain that except with probability $\negl(\secp)$,
\begin{equation}\label{eq:uniform_estimation}
|\widehat F(\tau)-F(\tau)|\le |\widehat F(\tau)-\mu_\tau|+|\mu_\tau-F(\tau)|\le \varepsilon/400+\nu
\quad\text{for all }\tau\in\mathcal T.
\end{equation}

Denote this \emph{good} event by $\mathcal G$. We next argue that \eqref{eq:select_rule} has at least one solution.
By \eqref{eq:target_window} there exists $\tau^\star$ with $F(\tau^\star)=\varepsilon/25$.
Since for each fixed $t$ the map $\tau\mapsto \Soft_{\tau,\xi}(t)$ is $1/\xi$-Lipschitz, we have
\[
|F(\tau)-F(\tau')|\le |\tau-\tau'|/\xi
\quad\text{for all }\tau,\tau'.
\]
Hence there exists $\widetilde\tau\in\mathcal T$ with
\[
|F(\widetilde\tau)-F(\tau^\star)|\le \rho/\xi=\varepsilon/800.
\]
Together with \eqref{eq:uniform_estimation}, this implies
\[
\widehat F(\widetilde\tau)\in[\varepsilon/100,\varepsilon/20].
\]
\noindent 
For any $\tau\in\mathcal T$ satisfying \eqref{eq:select_rule}, on $\mathcal G$,
\[
F(\tau)\ge \widehat F(\tau)-\varepsilon/400-\nu\ge \varepsilon/100-\varepsilon/400-\nu=\frac{3\varepsilon}{400}-\nu\ge \frac{\varepsilon}{150}-\nu,\]
\[
F(\tau)\le \widehat F(\tau)+\varepsilon/400+\nu\le \varepsilon/20+\varepsilon/400+\nu\le \frac{\varepsilon}{10}+\nu.
\]

\noindent Therefore, except with probability $\negl(\secp)$, the $\tau$ output by $\mathsf{Gen}^{\mathcal C}$ satisfies \eqref{eq:good_tau}:

\begin{equation}\label{eq:good_tau}
\frac{\varepsilon}{150}-\nu\ \le\ F(\tau)\ \le\ \frac{\varepsilon}{10} +\nu
\qquad\text{and}\qquad
\tau\ \ge\ \frac{\varepsilon}{150}-\nu.
\end{equation}
The lower bound on $\tau$ follows since on the support of $w_\tau$ we have $\Delta_x\le \tau$, hence
\[
F(\tau)=\mathbb E[\Delta_x w_\tau(x)]
\le \tau\cdot \mathbb E[w_\tau(X)]
\le \tau.
\]

\paragraph{Defining $S$ and outputting $\delta$.} We define the circuit $S$:\\

\noindent\underline{A quantum circuit S}

The circuit takes an input $(x,b)\in\bits^{n(\secp)+1}$.
\begin{enumerate}
\item estimates $\Delta_x^{(b)}:=p_x(b)-p_x(1-b)$ within $\pm \alpha$ except with probability $\negl(\secp)$,
obtaining $\widehat\Delta_x^{(b)}$;
\item outputs $1$ with probability $\Soft_{\tau,\xi}(\widehat\Delta_x^{(b)})$.
\end{enumerate}

\noindent Let
\[
w(x):=\Soft_{\tau,\xi}(\Delta_x)
\quad\text{and}\quad
\delta_\tau:=\mathbb E[w_\tau(X)].
\]
Because $\Soft_{\tau,\xi}$ is $1/\xi$-Lipschitz, for every fixed $x$:
\begin{equation}
\label{eq:S_is_good}
\Pr[S(x,P(x))=1\mid x]
= \Soft_{\tau,\xi}(\Delta_x)
\pm (\alpha/\xi+\negl(\secp))
= w_\tau(x)\pm(\varepsilon/800+\negl(\secp)).
\end{equation}
Independently sample $m_\delta=\poly(n,1/\varepsilon)$ points
$x_1,\dots,x_{m_\delta}\sim\mathcal D$ via $\mathsf{Samp}$.
For each $x_i$, estimate $\Delta_{x_i}$ within $\pm\alpha$
(except with probability $\negl(\secp)$),
obtaining $\widehat\Delta_{x_i}$, and form
\[
Z_i := \Soft_{\tau,\xi}(\widehat\Delta_{x_i})\in[0,1].
\]
Let
\[
\widehat\delta := \frac1{m_\delta}\sum_{i=1}^{m_\delta} Z_i.
\]
Recall $\delta_\tau=\E_{X\sim\mathcal D}[w_\tau(X)]$. Conditioning on $x_i=x$, on the good event $|\widehat\Delta_x-\Delta_x|\le \alpha$ the $1/\xi$-Lipschitz property gives
$\bigl|\Soft_{\tau,\xi}(\widehat\Delta_x)-\Soft_{\tau,\xi}(\Delta_x)\bigr|\le \alpha/\xi$,
whereas on the bad event we have $\Soft_{\tau,\xi}(\cdot)\in[0,1]$.
Hence
\[
\left|\E[Z_i\mid x_i=x]-w(x)\right|
=
\left|\E[\Soft_{\tau,\xi}(\widehat\Delta_x)\mid x_i=x]-\Soft_{\tau,\xi}(\Delta_x)\right|
\le \alpha/\xi+\nu,
\]
and averaging over $x\sim\mathcal D$ yields
\[
\bigl|\E[Z_i]-\delta_\tau\bigr|\le \alpha/\xi+\nu.
\]
Moreover, since $Z_i\in[0,1]$ are i.i.d., Hoeffding's inequality implies that there exists $m_\delta=\poly(n,1/\varepsilon)$
such that, except with probability $\nu$,
\[
\left|\widehat\delta-\E[Z_i]\right|\le \varepsilon/800.
\]
Thus, by the triangle inequality,
except with probability $\negl(\secp)$,
\[
|\widehat\delta-\delta_\tau|
\le
|\widehat\delta-\E[Z_i]|+|\E[Z_i]-\delta_\tau|
\le
\varepsilon/800+\alpha/\xi+\nu
\le
\varepsilon/400+\nu,
\]
where we used $\alpha/\xi=\varepsilon/800$. Define
\[
\delta := \max\{0,\ \widehat\delta-\varepsilon/400\}.
\]
Then (except with probability $\negl(\secp)$) we have $\delta\le \delta_\tau$ and
\[
0\le \delta_\tau-\delta\le \varepsilon/200,
\]
hence also $|\delta-\delta_\tau|\le \varepsilon/200$.

\paragraph{Indistinguishability.}
Define
\[
P'(x):=P(x)\oplus S(x,P(x))
\]

\noindent By definition of $P'$, for a fixed $x$,
\[
\Pr[\mathcal C(x,P(x))=1]
-\Pr[\mathcal C(x,P'(x))=1]
= \Pr[S(x,P(x))=1\mid x]\cdot \Delta_x.
\]
Averaging over $x\sim\mathcal D$ gives
\[
\Pr_{x\sim\mathcal D}[\mathcal C(x,P(x))=1]
-\Pr_{x\sim\mathcal D}[\mathcal C(x,P'(x))=1]
= \mathbb E_{x\sim\mathcal D}
\bigl[\Delta_x\cdot \Pr[S(x,P(x))=1\mid x]\bigr].
\]
Due to \eqref{eq:S_is_good}, replacing
$\Pr[S(x,P(x))=1\mid x]$ by $w(x)=\Soft_{\tau,\xi}(\Delta_x)$
incurs an error smaller than $\varepsilon/200+\negl(\secp)$,
hence the right-hand side equals
$F(\tau)\pm(\varepsilon/200+\negl(\secp))$.
Using \eqref{eq:good_tau}, its absolute value is at most
\[
\varepsilon/10+\varepsilon/200+\negl(\secp)\le \varepsilon.
\]

\paragraph{Predictability.}
Define
\[
g_x := p_x(1)-p_x(0).
\]
Then for every $x$,
\[
g_x=(2P(x)-1)\Delta_x,
\qquad\text{equivalently}\qquad
(2P(x)-1)\,g_x=\Delta_x,
\qquad |g_x|=|\Delta_x|.
\]

\noindent\underline{A quantum circuit Q}

The circuit takes an input $x\in\bits^{n(\secp)}$.
\begin{enumerate}
\item Estimates $g_x$ to within $\pm\alpha$ with failure probability at most
$\nu\le \negl(\secp)$, obtaining a random estimate $\widehat g_x$ satisfying
\begin{equation}\label{eq:E_x_def}
\Pr\bigl[\,|\widehat g_x-g_x|\le \alpha \,\bigm|\, x\bigr]\ \ge\ 1-\nu.
\end{equation}
\item Outputs a bit according to the following randomized rule (depending only on $\widehat g_x$):
\[
Q(x)
:= \begin{cases}
1, & \widehat g_x\ge \tau,\\[1mm]
0, & \widehat g_x\le -\tau,\\[1mm]
1\ \text{with prob.}\ \frac12\Bigl(1+\frac{\widehat g_x}{\tau}\Bigr),
& \text{otherwise}.
\end{cases}
\]
\end{enumerate}

\ifsubmission
\emph{A Lipschitz bound in the estimate.}
\else
\paragraph{ A Lipschitz bound.}
\fi
Fix $x$ and $b\in\{0,1\}$. Define
\[
\phi_b(t)\ :=\ \Pr\bigl[Q(x)=b\ \bigm|\ \widehat g_x=t\bigr].
\]
By inspection, $\phi_b$ is piecewise linear with slope in
$\bigl[-\tfrac{1}{2\tau},\,\tfrac{1}{2\tau}\bigr]$; hence it is $1/(2\tau)$-Lipschitz :
\begin{equation}\label{eq:phi_lip}
|\phi_b(t)-\phi_b(t')|\ \le\ \frac{|t-t'|}{2\tau}\qquad\forall\,t,t'\in\mathbb R.
\end{equation}
Let $\mathcal E_x:=\{|\widehat g_x-g_x|\le \alpha\}$. On the event $\mathcal E_x$,
\[
\phi_{P(x)}(\widehat g_x)
\ \ge\ \phi_{P(x)}(g_x)-\frac{\alpha}{2\tau},
\]
using \eqref{eq:phi_lip}. Taking expectation conditioned on $x$ and $\mathcal E_x$ gives
\[
\Pr\bigl[Q(x)=P(x)\mid x,\mathcal E_x\bigr]
= \E\bigl[\phi_{P(x)}(\widehat g_x)\mid x,\mathcal E_x\bigr]
\ \ge\ \phi_{P(x)}(g_x)-\frac{\alpha}{2\tau}.
\]
Removing the conditioning on $\mathcal E_x$ incurs a loss of at most $\nu$:
\begin{equation}\label{eq:succ_lip}
\Pr\bigl[Q(x)=P(x)\mid x\bigr]
\ \ge\ \phi_{P(x)}(g_x)-\frac{\alpha}{2\tau}-\nu.
\end{equation}

\ifsubmission
\emph{Exact success when $\widehat g_x=g_x$.}
\else
\paragraph{Exact success when $\widehat g_x=g_x$.}
\fi
Define the function $s_\tau:\mathbb R\to[0,1]$ by
\[
s_\tau(u)\ :=\
\begin{cases}
1, & u\ge \tau,\\[1mm]
\frac12\left(1+\frac{u}{\tau}\right), & |u|<\tau,\\[1mm]
0, & u\le -\tau.
\end{cases}
\]

\noindent Observe that
\begin{equation}\label{eq:phi_equals_s}
\phi_{P(x)}(g_x)\ =\ s\bigl((2P(x)-1)g_x\bigr)\ =\ s_\tau(\Delta_x),
\end{equation}
using $(2P(x)-1)g_x=\Delta_x$.
Combining \eqref{eq:succ_lip} and \eqref{eq:phi_equals_s} yields, for every $x$,
\begin{equation}\label{eq:succ_vs_s_fixed}
\Pr\bigl[Q(x)=P(x)\mid x\bigr]\ \ge\ s_\tau(\Delta_x)-\frac{\alpha}{2\tau}-\nu.
\end{equation}

\ifsubmission
\emph{A pointwise comparison using the soft weight.}
\else
\paragraph{A pointwise comparison using the soft weight.}
\fi
Assume, as ensured by our parameter choices in Step~3, that $\tau\ge \xi$.
We claim that for every $\Delta_x\in[-1,1]$,
\begin{equation}\label{eq:key_pointwise_rewrite}
s_\tau(\Delta_x)\ \ge\
1-\frac{\Soft_{\tau,\xi}(\Delta_x)}{2}
+\frac{\Delta_x\,\Soft_{\tau,\xi}(\Delta_x)}{2\tau}
-\frac{\xi}{8\tau}.
\end{equation}

\begin{proof}[Proof of \eqref{eq:key_pointwise_rewrite}]
We consider three cases.

\emph{Case 1: $\Delta_x\ge \tau$.}
Then $\Soft_{\tau,\xi}(\Delta_x)=0$ and the right-hand side equals $1-\xi/(8\tau)\le 1=s_\tau(\Delta_x)$.

\emph{Case 2: $\Delta_x\le \tau-\xi$.}
Then $\Soft_{\tau,\xi}(\Delta_x)=1$ and the right-hand side is
\[
\frac12\left(1+\frac{\Delta_x}{\tau}\right)-\frac{\xi}{8\tau}.
\]
If $\Delta_x\le -\tau$ then $s_\tau(\Delta_x)=0$ and the inequality holds since the displayed quantity is $\le 0$.
If $-\tau<\Delta_x\le\tau-\xi$ then $s_\tau(\Delta_x)=\frac12(1+\Delta_x/\tau)$ and the inequality holds with slack $\xi/(8\tau)$.

\emph{Case 3: $\Delta_x\in(\tau-\xi,\tau)$.}
Write $u:=\tau-\Delta_x\in(0,\xi)$. Then $\Soft_{\tau,\xi}(\Delta_x)=u/\xi$ and the right-hand side of
\eqref{eq:key_pointwise_rewrite} (without the $-\xi/(8\tau)$ term) equals
\[
1-\frac{u}{2\xi}+\frac{\Delta_x}{2\tau}\cdot\frac{u}{\xi}
=
1-\frac{u}{2\xi}+\frac{\tau-u}{2\tau}\cdot\frac{u}{\xi}
=
1-\frac{u^2}{2\tau\xi}.
\]

\noindent Since $\tau\ge \xi$ and $u\in(0,\xi)$, we have $\Delta_x=\tau-u>0$, hence $\Delta_x\in(-\tau,\tau)$ and therefore
\[
s_\tau(\Delta_x)=\frac12\left(1+\frac{\Delta_x}{\tau}\right)=1-\frac{u}{2\tau}.
\]
Thus,
\[
s_\tau(\Delta_x)-\left(1-\frac{u^2}{2\tau\xi}\right)
=
\frac{u}{2\tau}\left(\frac{u}{\xi}-1\right)
\ge -\frac{\xi}{8\tau},
\]
because the minimum over $u\in[0,\xi]$ occurs at $u=\xi/2$ and equals $-\xi/(8\tau)$.
This proves \eqref{eq:key_pointwise_rewrite} in all cases.
\end{proof}

\noindent Combining \eqref{eq:succ_vs_s_fixed} and \eqref{eq:key_pointwise_rewrite} gives,
for every $x$,
\begin{equation}\label{eq:predict_pointwise_fixed_rewrite}
\Pr\bigl[Q(x)=P(x)\mid x\bigr]
\ge
1-\frac{w_\tau(x)}{2}+\frac{\Delta_x w_\tau(x)}{2\tau}
-\frac{\xi}{8\tau}
-\frac{\alpha}{2\tau}
-\nu.
\end{equation}

\ifsubmission
\emph{Averaging and concluding.}
\else
\paragraph{Averaging and concluding.}
\fi
Taking expectation of \eqref{eq:predict_pointwise_fixed_rewrite} over $x\leftarrow\mathcal D$ yields
\[
\Pr_{x\sim\mathcal D}\bigl[Q(x)=P(x)\bigr]
\ge
1-\frac{\delta_\tau}{2}
+\frac{F(\tau)}{2\tau}
-\frac{\xi}{8\tau}
-\frac{\alpha}{2\tau}
-\nu,
\]
since $\delta_\tau=\E[w(X)]$ and $F(\tau)=\E[\Delta_X w(X)]$.

Using the concrete choice $\alpha=\xi\varepsilon/800$ and $\varepsilon\le 1$ without loss of generality, we have $\alpha\le \xi/800$, and hence
\begin{equation}\label{eq:slack_bound}
\frac{\xi}{8\tau}+\frac{\alpha}{2\tau}
\ \le\
\frac{\xi}{8\tau}+\frac{\xi}{1600\tau}
\ <\
\frac{\xi}{2\tau}.
\end{equation}
Therefore,
\[
\Pr_{x\sim\mathcal D}\bigl[Q(x)=P(x)\bigr]
\ge
1-\frac{\delta_\tau}{2}
+\frac{F(\tau)-\xi}{2\tau}
-\nu.
\]
Finally, by Step~4 we have $0 \le \delta_\tau-\delta\le \varepsilon/200$ except with probability $\negl(\secp)$,
and by Step~3 we have $F(\tau)\ge \varepsilon/150$ and $\tau\le 1+\xi$.
Since $\xi=\varepsilon^2/1000$, for all sufficiently large $\secp$,
\[
\frac{F(\tau)-\xi}{2\tau}
\ \ge\
\frac{\varepsilon/150-\varepsilon^2/1000}{2(1+\xi)}
\ \ge\
\frac{\varepsilon}{400}
\ \ge\
\frac{\delta_\tau-\delta}{2}.
\]
Plugging this in gives
\[
\Pr_{x\sim\mathcal D}\bigl[Q(x)=P(x)\bigr]
\ge
1-\frac{\delta_\tau}{2}+\frac{\delta_\tau-\delta}{2}-\nu
=
1-\frac{\delta}{2}-\nu,
\]
which completes the predictability claim.

\paragraph{Obliviousness.} The obliviousness property stems directly from the definition of $Q$.

\paragraph{Efficiency.}
All estimation steps use $\poly(\secp,1/\varepsilon)$ oracle calls to
$\mathcal C$, and the grid has $\poly(\secp,1/\varepsilon)$ points.
Thus $\mathsf{Gen}$ runs in $\poly(\secp,1/\varepsilon)$ time. By inspection, $S$ and $Q$ are obtained by composing a $\poly(\secp,1/\varepsilon)$-number
of invocations of $\mathcal C$ arising only from the estimation subroutines, with
$\poly(\secp,1/\varepsilon)$ additional non-oracle computation. Hence there exist fixed
polynomials $q(\cdot,\cdot)$ and $r(\cdot,\cdot)$ such that $S$ and $Q$ make at most
$q(\secp,1/\varepsilon)$ oracle calls to $\mathcal C$ and have non-oracle size at most
$r(\secp,1/\varepsilon)$. Therefore, the compiled circuits satisfy
$|\widehat Q|\le q(\secp,1/\varepsilon)\cdot|\mathcal C|+r(\secp,1/\varepsilon)$, and
similarly for $\widehat S$.

\end{proof}

\noindent Next, we state and prove a corollary of Theorem \ref{thm:Q3}.

\begin{corollary}
\label{cor:soft-hardcore-random-labels}
Let $\lambda\in\N$ and let $n=n(\lambda)=\poly(\lambda)$.  Let $D_\lambda$ be a distribution over
$\cX_\lambda\times\{0,1\}$ where $\cX_\lambda:=\{0,1\}^{n}$, and assume $D_\lambda$ is samplable in $\poly(\lambda)$ time.
Write $(X,B)\gets D_\lambda$. Assume the following min-entropy condition:
\[
m(\lambda)\;:=\;\max_{x\in\cX_\lambda}\Pr_{D_\secp}[X=x]\;=\;\negl(\lambda).
\]
Let $\delta=\delta(\lambda)\in(0,1)$ and let $s(\lambda)$ be a circuit size bound.  Suppose that for every uniform (resp. non-uniform) quantum circuit
$Q:\cX_\lambda\to\{0,1\}$ of size at most $s(\lambda)$,
\[
\Pr_{(X,B)\gets D_\lambda}\bigl[Q(X)=B\bigr]\ \le\ 1-\frac{\delta}{2}.
\]
Fix any $\epsilon=\epsilon(\lambda)=1/\poly(\lambda)$ with $\epsilon\le \delta/2$.  Let $\varepsilon:=\delta\epsilon$. In the uniform case, assume that the sampler for $D_\lambda$
is uniform and that $\delta(\lambda)$ and $\epsilon(\lambda)$
are computable in polynomial time. Let $C:\cX_\lambda\to\{0,1\}$ be a uniform (resp. non-uniform) quantum circuit such that
\begin{equation}
\label{eq:size-loss-condition}
q\!\Bigl(\lambda,\tfrac{1}{\varepsilon}\Bigr)\cdot\bigl(|C|+n\bigr)
\;+\;
r\!\Bigl(\lambda,\tfrac{1}{\varepsilon}\Bigr)
\ \le\ s(\lambda),
\end{equation}
where $r(\cdot,\cdot),q(\cdot,\cdot)$ are the polynomials defined in the statement of Theorem \ref{thm:Q3}. Then, there exists a measure $M_\lambda\preceq D_\lambda$ with total mass
$\mu(M_\lambda)\ge \delta-\epsilon$ such that
\[
\Pr_{(X,B)\leftarrow M_\lambda}\bigl[C(X)=B\bigr]\ \le\ \frac12+2\epsilon,
\]
where $(X,B)\leftarrow M_\lambda$ denotes sampling from the normalized
distribution $M_\lambda/\mu(M_\lambda)$.
\end{corollary}

\begin{proof}
Fix $\lambda$ and suppress the subscript $\lambda$.

\paragraph{Define a deterministic predicate.}
Define the lifted domain $\cX':=\cX\times\{0,1\}$ and define the deterministic predicate $P:\cX'\to\{0,1\}$ by
\[
P(x,b_0)\ :=\ b_0.
\]
A sampler for $(X',P(X'))$ is immediate: sample $(x,b_0)\gets D$ and output $\bigl((x,b_0),b_0\bigr)$.

\paragraph{Build an oracle circuit for Theorem~\ref{thm:Q3}.}
Given the predictor $C:\cX\to\{0,1\}$, define a quantum circuit

\[
\mathcal C:\cX'\times \bits\to\{0,1\}
\]
that on input $\bigl((x,b_0),b\bigr)$ outputs $1$ iff $C(x)=b$.
Crucially, $\mathcal C$ ignores $b_0$.

Invoke $\mathsf{Gen}^{\mathcal C}(\varepsilon,1^\lambda)$ from Theorem~\ref{thm:Q3} on the distribution $D$ with the accuracy parameter set to $\varepsilon := \delta\epsilon$, obtaining
$\xi,\tau, \sigma$, circuits $S,Q$ with oracle access to $\mathcal C$, and an output parameter $\delta_{\sf out}$. Importantly, $D$ meets the min-entropy condition in the statement of Theorem \ref{thm:Q3}. With probability $1-\negl(\secp)$, the properties correct mass, indistinguishability, predictability, and obliviousness all hold. 

As in Theorem~\ref{thm:Q3},
let the pointwise gaps  be
\[
p_{x'}(b):=\Pr[\mathcal C(x',b)=1],
\qquad
\Delta_{x'}:=p_{x'}(P(x'))-p_{x'}(1-P(x')).
\]
For $x'=(x,b_0)$ we have
\[
p_{(x,b_0)}(b)=\Pr[C(x)=b],
\qquad
\Delta_{(x,b_0)} = \Pr[C(x)=b_0]-\Pr[C(x)=1-b_0] = 2\Pr[C(x)=b_0]-1.
\]
Define the soft weight and its mean:

\[
    w_\tau(x')
    :=\operatorname{Soft}_{\tau,\xi}(\sigma\Delta_{x'}),
    \qquad
    \delta_\tau
    :=\mathbb E_{X'\sim\mathcal D}[w_\tau(X')].
\]

\paragraph{Define a measure $M$.}
Note that for any $x'$, $0 \le w_\tau(x')\le 1$. Define the measure $M\preceq D$ by
\begin{equation}
\label{eq:dist_and_meas}
M(x,b)\ :=\ D(x,b)\cdot w_\tau(x,b).
\end{equation}
Then, the total mass of $M$ is as follows

\begin{equation}
\label{eq:measure_mass}
\mu(M)=\sum_{x,b}M(x,b)=\E_{X'\sim \mathcal D}[w_\tau(X')]=\delta_\tau. 
\end{equation}

\paragraph{Lower bound the total mass $\mu(M)$.}
We claim $\delta_\tau\ge \delta-\epsilon$. By the \textit{correct mass} property from Theorem~\ref{thm:Q3}, except with negligible probability,
\[
\delta_{\sf out}\ \le\ \delta_\tau
\qquad\text{and}\qquad
|\delta_{\sf out}-\delta_\tau|\le \varepsilon/200.
\]
Hence it suffices to show $\delta_{\sf out}\ge \delta-\epsilon$.

Assume toward contradiction that $\delta_{\sf out}<\delta-\epsilon$ and that the Theorem~\ref{thm:Q3} guarantees hold.
By the \textit{predictability} property from Theorem \ref{thm:Q3},
\[
\Pr_{X'\sim \cD}\bigl[Q(X')=P(X')\bigr]\ \ge\ 1-\frac{\delta_{\sf out}}{2}-\negl(\lambda)
\ >\ 1-\frac{\delta}{2}+\frac{\epsilon}{2}-\negl(\lambda).
\]
Define $\widetilde Q:\cX\to\bits$ by $\widetilde Q(x):=Q((x,0))$.
We claim that for every $(x,b_0)\in\cX'$,
\begin{equation}\label{eq:oblivious_freeze_b0}
\Pr\bigl[Q((x,b_0))=b_0\bigr]\ =\ \Pr\bigl[Q((x,0))=b_0\bigr].
\end{equation}

Indeed, by \emph{obliviousness} of $Q$ (Theorem~\ref{thm:Q3}), on any input $x'=(x,b_0)$ the circuit $Q$ utilizes $x$ only within a fixed, $x$-independent procedure of oracle queries to $C(x,\cdot)$.

Since $\mathcal C((x,b_0),\cdot)$ ignores $b_0$, the joint distribution of oracle answers returned to $Q$
is identical for inputs $(x,0)$ and $(x,1)$.
Therefore the output distribution of $Q((x,b_0))$ is the same as that of $Q((x,0))$, establishing
\eqref{eq:oblivious_freeze_b0}. Averaging \eqref{eq:oblivious_freeze_b0} over $(X,B)\sim D$ yields
\[
\Pr_{(X,B)\sim D}\bigl[Q((X,B))=B\bigr]
\;=\;
\Pr_{(X,B)\sim D}\bigl[Q((X,0))=B\bigr]
\;=\;
\Pr_{(X,B)\sim D}\bigl[\widetilde Q(X)=B\bigr].
\]
Consequently,
\[
\Pr_{(X,B)\sim D}\bigl[\widetilde Q(X)=B\bigr]
\;=\;
\Pr_{X'\sim\cD}\bigl[Q(X')=P(X')\bigr]
\ >\ 1-\frac{\delta}{2}+\frac{\epsilon}{2}-\negl(\lambda).
\]

\paragraph{Size of the compiled predictor.}
We now argue that $\widetilde Q$ can be realized by a (non-oracle) quantum circuit of size at most $s(\lambda)$.
Recall that $Q$ is produced by invoking Theorem~\ref{thm:Q3} on the oracle circuit $\mathcal C$ with accuracy parameter
$\varepsilon=\delta\epsilon$.

By Theorem~\ref{thm:Q3}, the circuit $Q$ makes at most
$q(\lambda,1/\varepsilon)$ oracle calls to $\mathcal C$ on any input and has non-oracle size at most
$r(\lambda,1/\varepsilon)$. By hardwiring $\mathcal C$ and replacing each oracle gate by a copy of $\mathcal C$,
we obtain an explicit (non-oracle) circuit $\widehat Q$ of size at most
\[
|\widehat Q|\ \le\ q(\lambda,1/\varepsilon)\cdot|\mathcal C| \;+\; r(\lambda,1/\varepsilon).
\]
In our construction, $\mathcal C$ on input $\bigl((x,b_0),b\bigr)$ evaluates $C(x)$ and compares it to $b$ while
ignoring $b_0$, and hence
\[
|\mathcal C|\ \le\ |C|+n.
\]
Therefore, the size condition~\eqref{eq:size-loss-condition} implies $|\widehat Q|\le s(\lambda)$. This contradicts the assumed hardness of $D$ against size-$s(\lambda)$ circuits. Therefore $\delta_{\sf out}\ge \delta-\epsilon$, and since $\delta_{\sf out}\le \delta_\tau$ we get
\begin{equation}
\label{eq:delta_tau_and_delta}
\mu(M)=\delta_\tau\ \ge\ \delta-\epsilon.
\end{equation}

In the uniform case, we instead show that there exists an
output of $\operatorname{Gen}$ satisfying all four guarantees of \cref{thm:Q3}
and $\delta_{\mathrm{out}}\geq\delta-\epsilon$.
Suppose no such output exists for infinitely many $\lambda$.
Consider the uniform predictor $A$ that, on input $x$, runs
$\operatorname{Gen}^{\mathcal{C}}(\varepsilon,1^\lambda)$
with fresh randomness and outputs $Q((x,0))$, where $Q$ is
the generated circuit.

For every output satisfying the guarantees, the predictability
and obliviousness calculations above apply. By \cref{thm:Q3}, all four guarantees hold simultaneously
with probability $1-\operatorname{negl}(\lambda)$ over
the generator's execution. averaging
over the generator's outputs gives, for these $\lambda$,
\[
    \Pr_{(X,B)\leftarrow D}[A(X)=B]
    \geq 1-\frac{\delta}{2}+\frac{\epsilon}{2}
           -\operatorname{negl}(\lambda)
    > 1-\frac{\delta}{2}
\]
for sufficiently large $\lambda$.
Our choice of $q,r$ and the size condition \eqref{eq:size-loss-condition} in this corollary's statement
ensure that 
\[
    |A|
    \leq q(\lambda,1/\varepsilon)(|C|+n)
         +r(\lambda,1/\varepsilon)
    \leq s(\lambda),
\]
contradicting uniform hardness.

Thus, for all sufficiently large $\lambda$, we may fix an
output of $\operatorname{Gen}$ satisfying all four guarantees of \cref{thm:Q3} and
$\delta_{\mathrm{out}}\geq\delta-\epsilon$.
The \textit{correct mass} guarantee then gives
$\mu(M)=\delta_\tau\geq\delta-\epsilon$
for its associated measure $M$.
Use this fixed output for the remainder of the proof.

\paragraph{An upper bound on the success of $C$ on $M$.}
We analyze $C$ under the normalized distribution $M/\mu(M)$. Using $\Delta_{(x,b)} = 2\Pr[C(x)=b]-1$, we have
\begin{equation}
\label{eq:success_over_meas}
\begin{aligned}
\Pr_{(X,B)\leftarrow M}\bigl[C(X)=B\bigr]
=
\frac12+\frac12\cdot \E_{(X,B)\leftarrow M}\bigl[\Delta_{(X,B)}\bigr]
=
\frac12+\frac{1}{2\delta_\tau}\cdot \E_{(X,B)\sim D}\bigl[\Delta_{(X,B)}\,w_\tau(X,B)\bigr].
\end{aligned}
\end{equation}

where the last equality stems from \eqref{eq:dist_and_meas} and \eqref{eq:measure_mass}. Let $F(\tau):=\E_{X'\sim\cD}[\Delta_{X'}\,w_\tau(X')]$, as in the proof of Theorem~\ref{thm:Q3}. Let $P'(x'):=P(x')\oplus S(x',P(x'))$.  Fixing $x'$, we get
\begin{equation*}
\begin{aligned}
\Pr[\mathcal C(x',P(x'))=1 \mid x']-\Pr[\mathcal C(x',P'(x'))=1 \mid x']&=\\\Pr[P'(x')=1-P(x')]\cdot\left(\Pr[\mathcal C(x',P(x'))=1 \mid x']-\Pr[\mathcal C(x',1-P(x'))=1 \mid x']\right)
&=\\
\Pr[S(x',P(x'))=1\mid x']\cdot \Delta_{x'}.
\end{aligned}
\end{equation*}
Averaging over $X'\sim\cD$ and using the \textit{indistinguishability} property from Theorem \ref{thm:Q3} gives
\begin{equation}
\label{eq:indist-gap}
A:=\left|\E_{X'\sim\cD}\Bigl[\Delta_{X'}\cdot \Pr[S(X',P(X'))=1\mid X']\Bigr]\right|
\ \le\ \varepsilon.
\end{equation}
By the \textit{correct mass} property from Theorem \ref{thm:Q3}, for each fixed $x'$ we have
\[\Pr[S(x',P(x'))=1\mid x']=w_\tau(x')\pm(\varepsilon/200+\negl(\lambda)).\]
Combining with \eqref{eq:indist-gap} and observing that $|\Delta_{x'}|\le 1$, 
\[
|F(\tau)| \ = \  |\E_{X'\sim\cD}[\Delta_{X'}\,w_\tau(X')]|\ \le \ A + \varepsilon/200+\negl(\secp)   \ \le\ \varepsilon+\varepsilon/200+\negl(\lambda) \  \le\ 2\varepsilon
\]
for all sufficiently large $\lambda$. From \eqref{eq:success_over_meas} we have the identity
\[
\Pr_{(X,B)\leftarrow M}\bigl[C(X)=B\bigr]
=
\frac12+\frac{1}{2\delta_\tau}\cdot F(\tau),
\]
Because $\delta_\tau>0$ this results in the first inequality below:

\[
\Pr_{(X,B)\leftarrow M}\bigl[C(X)=B\bigr]
\ \le\
\frac12+\frac{|F(\tau)|}{2\delta_\tau}
\ \le\
\frac12+\frac{2\varepsilon}{2(\delta-\epsilon)}
\ =\
\frac12+\frac{\delta\epsilon}{\delta-\epsilon}
\ \le\ \frac12+2\epsilon.
\]
Here, the second inequality relies on \eqref{eq:delta_tau_and_delta}, and the last inequality uses $\epsilon\le \delta/2$, so $\delta/(\delta-\epsilon)\le 2$. This completes the proof.
\end{proof}

\noindent In the next proof, we make use of the following lemma.

\begin{lemma}[Sampling lemma, following {\cite[Section~5.4]{10.1007/978-3-540-45198-3_18}}] \label{thm:convex-combo}

Let $n \geq 1$, let $F$ be a set of functions
$f:\{0,1\}^n \to [-1,1]$, and let $\bar f$ be a convex
combination of functions from $F$.
For every $\epsilon \in (0,1)$, setting
\[
    k :=
    2^{\left\lceil \log_2(4n/\epsilon^2) \right\rceil}
    \leq \frac{8n}{\epsilon^2},
\]
there exist functions $f_1,\ldots,f_k \in F$, with repetitions
allowed, such that
\[
    \max_{x \in \{0,1\}^n}
    \left|
        \bar f(x)-\frac{1}{k}\sum_{i=1}^k f_i(x)
    \right|
    \leq \epsilon.
\]
\end{lemma}

\begin{proof}
Let $\nu$ be a distribution on $F$ representing the convex
combination, so that
$\bar f(x)=\mathbb{E}_{f\sim\nu}[f(x)]$.
Sample $f_1,\ldots,f_k$ independently from $\nu$.
For each fixed $x$, the random variables $f_i(x)$ lie in
$[-1,1]$. Hence Hoeffding's inequality gives
\[
    \Pr\left[
        \left|
            \frac{1}{k}\sum_{i=1}^k f_i(x)-\bar f(x)
        \right|>\epsilon
    \right]
    \leq 2\exp\left(-\frac{k\epsilon^2}{2}\right).
\]
Taking a union bound over all $2^n$ inputs and using
$k\geq 4n/\epsilon^2$, we obtain
\[
    \Pr\left[
        \max_{x\in\{0,1\}^n}
        \left|
            \frac{1}{k}\sum_{i=1}^k f_i(x)-\bar f(x)
        \right|>\epsilon
    \right]
    \leq 2^{n+1}e^{-2n}<1.
\]
Therefore some choice of $f_1,\ldots,f_k$ satisfies the
claimed approximation.
\end{proof}

\begin{remark} We note that Lemma~\ref{thm:convex-combo} does not guarantee that the functions $f_1,\ldots,f_k$ can be computed efficiently given $\overline{f}$. Therefore, the following theorems we prove in this section, which all rely on Lemma~\ref{thm:convex-combo}, are stated in the non-uniform security model.\end{remark}

\begin{lemma}[Minimax fixing lemma]
\label{lem:minimax-fix}
Let $\{D_\secp\}_{\secp\in\bbN}$ be a family of distributions over
$\{0,1\}^{n(\secp)}\times\{0,1\}$, where $n(\secp)\ge 1$.
Fix functions $\alpha(\secp)\in(0,1]$ and $\gamma(\secp)\in(0,1)$, and a size bound
$t(\secp)$.

Assume that for all sufficiently large $\secp$, for every measure
$M_\secp \preceq D_\secp$ with $\mu(M_\secp)\ge \alpha(\secp)$, there exists a
quantum circuit $Q'_\secp$ of size at most $t(\secp)$ such that
\[
\Pr_{(X,B)\leftarrow M_\secp}\bigl[\,Q'(X)=B\,\bigr]
\ >\ \frac12+\gamma(\secp).
\]
Define

\[
    t''(\lambda)
    :=
    \frac{8n(\lambda)}{\gamma(\lambda)^2}\,t(\lambda)
    +
    \log_2\left(\frac{8n(\lambda)}{\gamma(\lambda)^2}\right).
\]

Then for all sufficiently large $\secp$, there exists a quantum circuit
$Q''_\secp$ of size at most $t''(\secp)$ such that for every measure
$M_\secp \preceq D_\secp$ with $\mu(M_\secp)\ge \alpha(\secp)$,
\[
\Pr_{(X,B)\leftarrow M_\secp}\bigl[\,Q''_\secp(X)=B\,\bigr]
\ >\ \frac12+\frac{\gamma(\secp)}{2}.
\]
\end{lemma}

\begin{proof}
Fix $\secp\in\mathbb{N}$ and suppress it from the notation. 

\paragraph{Mass-normalization.}
Throughout this proof, the notation $(X,B)\leftarrow M$ denotes sampling from the
\emph{normalized} distribution $M/\mu(M)$, so $\E_{(x,b)\leftarrow M}[\cdot]$ is not linear in $M$.
To apply minimax, we may w.l.o.g.\ restrict Bob to measures of \emph{exact} mass $\alpha$.
Indeed, for any measure $M\preceq D$ with $\mu(M)\ge \alpha$, define the scaled measure
\[
\widetilde M \ :=\ \frac{\alpha}{\mu(M)}\cdot M.
\]
Then $\widetilde M\preceq D$, $\mu(\widetilde M)=\alpha$, and the normalized distributions coincide:
\[
\frac{\widetilde M}{\mu(\widetilde M)} \ =\ \frac{M}{\mu(M)}.
\]
Hence, for every circuit $Q$,
\begin{equation}
\label{eq:equivalence_of_measures}
\Pr_{(X,B)\leftarrow M}[Q(X)=B]\ =\ \Pr_{(X,B)\leftarrow \widetilde M}[Q(X)=B].
\end{equation}
Therefore, the lemma’s assumption and desired conclusion are unchanged if we quantify Bob only over
\[
R_\alpha \ :=\ \{\,M\preceq D \mid \mu(M)=\alpha\,\}.
\]
Note that $R_\alpha$ is nonempty (since $\alpha D\in R_\alpha$), convex, and compact when identified
with a subset of $\R^{\{0,1\}^n\times\{0,1\}}$.

\paragraph{The game.}
For any circuit $Q'$, define the function $g_{Q'}:\{0,1\}^n\times\{0,1\}\to[-1,1]$ by
\[
g_{Q'}(x,b)\ :=\ 2\Pr[Q'(x)=b]-1 \ =\ \Pr[Q'(x)=b]-\Pr[Q'(x)=1-b].
\]
Then for any measure $M$,
\[
\E_{(x,b)\leftarrow M}[g_{Q'}(x,b)] \ =\ 2\Pr_{(x,b)\leftarrow M}[Q'(x)=b]-1.
\]

Consider the following zero-sum game between Alice and Bob.
Alice chooses a circuit $Q'$ of size at most $t$;
Bob chooses a measure $M\in R_\alpha$;
and Alice's payoff is
\[
\Pi(Q',M)\ :=\ \frac{1}{\alpha}\sum_{x,b} M(x,b)\cdot g_{Q'}(x,b)
\qquad\text{for } M\in R_\alpha.
\]
For $M\in R_\alpha$ we have $\Pi(Q',M)=\E_{(x,b)\leftarrow M}[g_{Q'}(x,b)]$, and $\Pi(Q',\cdot)$ is
linear in $M$ on $R_\alpha$. By the lemma assumption, for every $M\in R_\alpha$ there exists a circuit $Q'$ of size at most $t$
such that
\[
\Pr_{(X,B)\leftarrow M}[Q'(X)=B]\ >\ \frac12+\gamma,
\]
and hence
\[
\Pi(Q',M)\ =\ 2\Pr_{(X,B)\leftarrow M}[Q'(X)=B]-1\ >\ 2\gamma.
\]
Equivalently,
\[
\min_{M\in R_\alpha}\ \max_{Q':\,|Q'|\le t}\ \Pi(Q',M)\ >\ 2\gamma.
\]

\paragraph{Minimax.}
Apply von Neumann's minimax theorem with Alice's finite action set being the set of size-$t$ circuits (with a finite universal gate set),
Bob's action set $Y=R_\alpha$ (compact convex), and payoff $u(Q',M)=\Pi(Q',M)$
(affine in $M$).
We obtain a mixed strategy for Alice, i.e.\ a distribution over size-$t$ circuits which we denote
$\overline Q$, such that for all $M\in R_\alpha$,
\[
\E_{Q'\sim \overline Q}\bigl[\Pi(Q',M)\bigr]\ >\ 2\gamma.
\]
Define the (pointwise) averaged payoff function
\[
g_{\overline Q}(x,b)\ :=\ \E_{Q'\sim \overline Q}\bigl[g_{Q'}(x,b)\bigr].
\]
By linearity,
\[
\E_{Q'\sim \overline Q}\bigl[\Pi(Q',M)\bigr]
\;=\; \frac{1}{\alpha}\sum_{x,b} M(x,b)\,g_{\overline Q}(x,b)
\;=\; \E_{(x,b)\leftarrow M}\bigl[g_{\overline Q}(x,b)\bigr],
\]
so we have for every $M\in R_\alpha$,
\begin{equation}
\label{eq:minimax_mixed}
\E_{(x,b)\leftarrow M}\bigl[g_{\overline Q}(x,b)\bigr]\ >\ 2\gamma.
\end{equation}

\paragraph{Sparsifying the mixed strategy.}
Note that $g_{Q'}(x,1)=-g_{Q'}(x,0)$ for every $x$, since $\Pr[Q'(x)=0]+\Pr[Q'(x)=1]=1$.
Define the one-argument function $h_{Q'}:\{0,1\}^n\to[-1,1]$ by
\[
h_{Q'}(x)\ :=\ g_{Q'}(x,0),
\]
so that for all $(x,b)$,
\[
g_{Q'}(x,b)\ =\ (1-2b)\,h_{Q'}(x).
\]
Define $h_{\overline Q}(x):=\E_{Q'\sim\overline Q}[h_{Q'}(x)]$; then
$g_{\overline Q}(x,b)=(1-2b)\,h_{\overline Q}(x)$.

Now apply Lemma~\ref{thm:convex-combo} to the convex combination $h_{\overline Q}$
over the domain $\{0,1\}^n$:
there exist circuits $Q'_1,\ldots,Q'_k$ of size at most $t$ with

\[
    k :=
    2^{\left\lceil \log_2(4n/\gamma^2) \right\rceil}
    \leq \frac{8n}{\gamma^2}.
\]

such that
\[
\max_{x\in\{0,1\}^n}
\left|
h_{\overline{Q}}(x)\ -\ \frac1k\sum_{i=1}^k h_{Q'_i}(x)
\right|
\ \le\ \gamma.
\]
Multiplying by $(1-2b)\in\{\pm1\}$ shows that, for all $(x,b)$,
\[
\left|
g_{\overline{Q}}(x,b)\ -\ \frac1k\sum_{i=1}^k g_{Q'_i}(x,b)
\right|
\ \le\ \gamma.
\]
Define $Q''$ as follows: on input $x$, sample $i\leftarrow [k]$ uniformly and output $Q'_i(x)$.
Then for every $(x,b)$ we have
\[
g_{Q''}(x,b)\ =\ \frac1k\sum_{i=1}^k g_{Q'_i}(x,b),
\]
and hence for every $M\in R_\alpha$,
\[
\E_{(x,b)\leftarrow M}[g_{Q''}(x,b)]
\ \ge\
\E_{(x,b)\leftarrow M}[g_{\overline{Q}}(x,b)]\ -\ \gamma
\ >\ 2\gamma-\gamma\ =\ \gamma,
\]
where the strict inequality uses \eqref{eq:minimax_mixed}.

\paragraph{Extending back to all $\mu(M)\ge\alpha$.}
Let
\[
R \ :=\ \{\,M\preceq D \mid \mu(M)\ge\alpha\,\}.
\]
For any $M\in R$, let $\widetilde M=\frac{\alpha}{\mu(M)}M\in R_\alpha$.
By \eqref{eq:equivalence_of_measures}, sampling from $M$ or $\widetilde M$ yields the same normalized
distribution, so the same success probability for $Q''$.
Therefore, for every $M\in R$,
\[
\Pr_{(x,b)\leftarrow M}[Q''(x)=b]
\ =\ \frac12+\frac12\,\E_{(x,b)\leftarrow M}[g_{Q''}(x,b)]
\ >\ \frac12+\frac{\gamma}{2}.
\]

\paragraph{Size bound.}
By construction, $Q''$ can be implemented by wiring in the $k$ circuits $Q'_1,\dots,Q'_k$
and using $\log k$ bits to sample $i$; thus we obtain
\[
    |Q''|
    \leq kt+\log_2 k
    \leq
    \frac{8n}{\gamma^2}\,t
    +\log_2\left(\frac{8n}{\gamma^2}\right)
    =t'',
\]
as required.

\end{proof}

\begin{theorem}[Post-quantum hard-core measure]
\label{thm:hardcore-Q3}
Let $\{D_\secp\}_{\secp\in\bbN}$ be a family of distributions over
$\{0,1\}^{n(\secp)}\times\{0,1\}$ that are sampleable in time $\poly(\secp)$, and assume the min-entropy condition
\[
\max_{(x,b)\in\{0,1\}^{n(\secp)}\times\{0,1\}} D_\secp(x,b)\ \le\ \negl(\secp).
\]
Fix a function $\delta(\secp)\in(0,1)$ and a size bound $s(\secp)$.
Let $q(\cdot,\cdot)$ and $r(\cdot,\cdot)$ be the polynomials from Theorem~\ref{thm:Q3}. Assume that for all sufficiently large $\secp$, for every quantum circuit
$Q:\{0,1\}^{n(\secp)}\to\{0,1\}$ of size at most $s(\secp)$,
\[
\Pr_{(X,B)\gets D_\secp}\bigl[\,Q(X)=B\,\bigr]
\ \le\ 1-\frac{\delta(\secp)}{2}.
\]

Let $\gamma(\secp)$ be any function such that $\gamma(\secp)\in [1/\poly(\secp),\delta(\secp)]$, and define
\[
\epsilon(\secp)\ :=\ \frac{\gamma(\secp)}{4},
\qquad
\varepsilon(\secp)\ :=\ \delta(\secp)\cdot \epsilon(\secp)
\ =\ \frac{\delta(\secp)\gamma(\secp)}{4},
\]
and
\[
s''(\secp)\ :=\ \left\lfloor
\frac{s(\secp)-r\!\bigl(\secp,1/\varepsilon(\secp)\bigr)}
     {q\!\bigl(\secp,1/\varepsilon(\secp)\bigr)}
\right\rfloor
-
n(\secp),
\qquad
 s'(\lambda)
    :=
    \left\lfloor
        \frac{\gamma(\lambda)^2}{16n(\lambda)}
        \,s''(\lambda)
    \right\rfloor.
\]
Then for all sufficiently large $\secp$, there exists a measure
$M_\secp \preceq D_\secp$ such that
\[
\mu(M_\secp)\ \ge\ \delta(\secp)-\epsilon(\secp)
\ =\ \delta(\secp)-\frac{\gamma(\secp)}{4},
\]
and for every quantum circuit
$Q':\{0,1\}^{n(\secp)}\to\{0,1\}$ of size at most $s'(\secp)$,
\[
\Pr_{(X,B)\leftarrow M_\secp}\bigl[\,Q'(X)=B\,\bigr]
\ \le\ \frac12+\gamma(\secp).
\]
\end{theorem}

\begin{proof}
Fix a sufficiently large security parameter $\secp$ and suppress it from the notation.
Write
$n:=n(\secp)$, $\delta:=\delta(\secp)$, $\gamma:=\gamma(\secp)$,
$\epsilon:=\gamma/4$, and $\varepsilon:=\delta\epsilon=\delta\gamma/4$.
Note that $0<\gamma<\delta<1$ and also $n\ge 1$: indeed, if $n=0$ then the domain has
size $2$ and hence $\max_{x,b}D(x,b)\ge 1/2$, contradicting
$\max_{x,b}D(x,b)\le \negl(\secp)$. We first dispose of the degenerate case $s'\le 0$.

\medskip\noindent
\textbf{Case 1: $s'<0$.}
Then there is no circuit of size at most $s'$, so the prediction condition is vacuous.
Take $M:=(\delta-\epsilon)\cdot D$. Then $M\preceq D$ and $\mu(M)=\delta-\epsilon$, as required.

\medskip\noindent
\textbf{Case 2: $s'=0$.}
Let $p_b:=\Pr_{(X,B)\leftarrow D}[B=b]$ for $b\in\{0,1\}$.
By the hypothesis of the theorem (applied to the constant circuits of size $0$),
\[
p_b=\Pr_{(X,B)\leftarrow D}[Q_b(X)=B]\ \le\ 1-\frac{\delta}{2}
\qquad\text{for}\qquad
Q_b(x)\equiv b.
\]
Hence $p_{1-b}=1-p_b\ge \delta/2$, so in fact $p_0,p_1\ge\delta/2$.
Define a measure $M\preceq D$ by reweighting labels:
\[
M(x,b)\ :=\ \frac{\delta-\epsilon}{2p_b}\,D(x,b)\qquad b\in\{0,1\}.
\]
Since $p_b\ge\delta/2\ge(\delta-\epsilon)/2$, we have $(\delta-\epsilon)/(2p_b)\le 1$,
and therefore $M\preceq D$. Moreover,
\[
\mu(M)=\sum_{b\in\{0,1\}}\sum_x M(x,b)
=\sum_{b\in\{0,1\}}\frac{\delta-\epsilon}{2p_b}\sum_x D(x,b)
=\sum_{b\in\{0,1\}}\frac{\delta-\epsilon}{2p_b}\,p_b
=\delta-\epsilon.
\]
Finally, under $M$ the label is perfectly balanced:
\[
\Pr_{(X,B)\leftarrow M}[B=b]
=\frac{\sum_x M(x,b)}{\mu(M)}
=\frac{(\delta-\epsilon)/2}{\delta-\epsilon}
=\frac12.
\]
Thus any size-$0$ circuit is constant and succeeds with probability exactly $1/2$,
which is at most $1/2+\gamma$. This establishes the theorem when $s'=0$.

\medskip\noindent
\textbf{Case 3: $s'\ge 1$.}
Assume towards contradiction that for every measure $M\preceq D$ with
$\mu(M)\ge \delta-\epsilon$ there exists a quantum circuit $Q'$ of size at most $s'$
such that
\begin{equation}\label{eq:assume-no-hardcore-fixed}
\Pr_{(X,B)\leftarrow M}\bigl[Q'(X)=B\bigr]\ >\ \frac12+\gamma.
\end{equation}
Apply Lemma~\ref{lem:minimax-fix} with parameters
$\alpha:=\delta-\epsilon$, $t:=s'$, and $\gamma$ (as above).
We obtain a single quantum circuit $Q''$ such that for all measures $M\preceq D$
with $\mu(M)\ge\delta-\epsilon$,
\begin{equation}\label{eq:fixed-good-fixed}
\Pr_{(X,B)\leftarrow M}\bigl[Q''(X)=B\bigr]\ >\ \frac12+\frac{\gamma}{2},
\end{equation}
and whose size satisfies

\begin{equation}\label{eq:Qpp-size-start}
|Q''|
    \leq
    \frac{8n}{\gamma^2}\,s'
    +\log_2\left(\frac{8n}{\gamma^2}\right).
\end{equation}

We are in the case $s'\geq 1$. Since $n\geq 1$ and
$0<\gamma<1$, we have
\[
    \log_2\left(\frac{8n}{\gamma^2}\right)
    \leq \frac{8n}{\gamma^2}
    \leq \frac{8n}{\gamma^2}\,s'.
\]

Consequently, by the definition of $s'$,

\begin{equation}\label{eq:Qpp_size_bound_fixed}
|Q''|
    \leq \frac{16n}{\gamma^2}\,s'
    \leq
    \frac{16n}{\gamma^2}
    \cdot \frac{\gamma^2}{16n}\,s''
    =s''.
\end{equation}

\medskip\noindent
\textbf{Applying Corollary~\ref{cor:soft-hardcore-random-labels}.}
We verify that the min-entropy premise of Corollary~\ref{cor:soft-hardcore-random-labels}
holds. Indeed, for every $x\in\{0,1\}^n$,
\[
\Pr[X=x]=\sum_{b\in\{0,1\}} D(x,b)\ \le\ 2\cdot \max_{(x,b)}D(x,b)=\negl(\secp).
\]
Apply Corollary~\ref{cor:soft-hardcore-random-labels} to the circuit
$C:=Q''$ with parameter $\epsilon:=\gamma/4$.
We verify the corollary's size condition~\eqref{eq:size-loss-condition}.
By \eqref{eq:Qpp_size_bound_fixed} and the definition
$s''=\left\lfloor\frac{s-r(\secp,1/\varepsilon)}{q(\secp,1/\varepsilon)}\right\rfloor-n$,
we have
\[
|C|+n\ \le\ s''+n\ \le\ \left\lfloor\frac{s-r(\secp,1/\varepsilon)}{q(\secp,1/\varepsilon)}\right\rfloor.
\]
Let $q_\secp:=q(\secp,1/\varepsilon)$ and $r_\secp:=r(\secp,1/\varepsilon)$.
Since $q_\secp\ge 1$ (interpreting $q_\secp$ as the integer oracle-call bound from Theorem~\ref{thm:Q3}$)$,
multiplying by $q_\secp$ and adding $r_\secp$ yields
\[
q_\secp\cdot(|C|+n)+r_\secp
\ \le\
q_\secp\cdot
\left\lfloor\frac{s-r_\secp}{q_\secp}\right\rfloor
+r_\secp\le s.
\]
which is exactly the needed size condition. Therefore the corollary yields a measure $M^\star\preceq D$ with
$\mu(M^\star)\ge\delta-\epsilon$ such that
\[
\Pr_{(X,B)\leftarrow M^\star}\bigl[Q''(X)=B\bigr]
\ \le\ \frac12+2\epsilon
\ =\ \frac12+\frac{\gamma}{2}.
\]
This contradicts \eqref{eq:fixed-good-fixed} applied to $M^\star$,
and hence the assumption \eqref{eq:assume-no-hardcore-fixed} was false.
Thus there exists a measure $M\preceq D$ with $\mu(M)\ge\delta-\epsilon$
such that every size-$s'$ circuit predicts $B$ from $X$ under $M$ with probability
at most $1/2+\gamma$, completing the proof.
\end{proof}
To conclude this subsection, we state a corollary that follows immediately from \ref{thm:hardcore-Q3} by taking the contrapositive.

\begin{corollary}[Contrapositive of Theorem~\ref{thm:hardcore-Q3}]
\label{cor:hardcore-contrap}
Let $\{D_\secp\}_{\secp\in\bbN}$ be a family of distributions over
$\{0,1\}^{n(\secp)}\times\{0,1\}$ that are sampleable in time $\poly(\secp)$, and assume
\[
\max_{(x,b)\in\{0,1\}^{n(\secp)}\times\{0,1\}} D_\secp(x,b)\ \le\ \negl(\secp).
\]
Fix functions $\delta(\secp)\in(0,1)$ and $s(\secp)\in\bbN$, and let
$q(\cdot,\cdot)$ and $r(\cdot,\cdot)$ be the polynomials from Theorem~\ref{thm:Q3}.
Let $\gamma(\secp)$ be any function such that $\gamma(\secp)\in [1/\poly(\secp),\delta(\secp)]$ and define
\[
\epsilon(\secp)\ :=\ \frac{\gamma(\secp)}{4},
\qquad
\varepsilon(\secp)\ :=\ \delta(\secp)\cdot \epsilon(\secp)
\ =\ \frac{\delta(\secp)\gamma(\secp)}{4},
\]
and
\[
s''(\secp)\ :=\ \left\lfloor
\frac{s(\secp)-r\!\bigl(\secp,1/\varepsilon(\secp)\bigr)}
     {q\!\bigl(\secp,1/\varepsilon(\secp)\bigr)}
\right\rfloor
-
n(\secp),
\qquad
 s'(\lambda)
    :=
    \left\lfloor
        \frac{\gamma(\lambda)^2}{16n(\lambda)}
        \,s''(\lambda)
    \right\rfloor.
\]
Suppose that for infinitely many $\secp$, for \emph{every} measure $M_\secp\preceq D_\secp$
satisfying
\[
\mu(M_\secp)\ \ge\ \delta(\secp)-\epsilon(\secp)
\ =\ \delta(\secp)-\frac{\gamma(\secp)}{4},
\]
there exists a quantum circuit $Q'_\secp:\{0,1\}^{n(\secp)}\to\{0,1\}$ of size at most
$s'(\secp)$ such that
\[
\Pr_{(X,B)\leftarrow M_\secp}\bigl[\,Q'_\secp(X)=B\,\bigr]\ >\ \frac12+\gamma(\secp).
\]
Then for infinitely many $\secp$, there exists a quantum circuit
$Q_\secp:\{0,1\}^{n(\secp)}\to\{0,1\}$ of size at most $s(\secp)$ such that
\[
\Pr_{(X,B)\gets D_\secp}\bigl[\,Q_\secp(X)=B\,\bigr]\ >\ 1-\frac{\delta(\secp)}{2}.
\]
\end{corollary}

\subsection{Key agreement amplification}\label{subsec:KA-amp}
In this section, we prove a theorem which enables generically amplifying a weak bit agreement protocol to a full-fledged key agreement protocol, in a manner that preserves post-quantum security.

\begin{theorem}
    \label{thm:ba-to-ka-params}
    Assuming there exists a post-quantum $(\epsilon,\delta)$-BA protocol secure against adversaries with non-uniform classical advice for constants $\epsilon,\delta\in[0,1]$ such that \[\delta < \frac{2\epsilon}{1+\epsilon},\] there exists a $(1,0)$-BA protocol secure against adversaries with non-uniform classical advice.
\end{theorem}
A proof of this theorem is presented at the end of this section, after developing the necessary tools. We begin by defining the notion of a $(\epsilon,\delta)$-secure random variable, and of an information theoretic key-agreement amplification for a set of such independent variables.

\begin{definition}[$(\epsilon,\delta)$-secure random variable \cite{Holenstein}] \label{def:secure-random-variable}
    A triple $X \times Y \times Z$ of random variables over $\{0,1\} \times \{0,1\} \times \{0,1\}^*$ is $(\epsilon,\delta)$-secure for $\epsilon = \epsilon(\secp)$ and $\delta = \delta(\secp)$ if 
    \begin{itemize}
        \item $\Pr[X=0]=\Pr[X=1]=\Pr[Y=0]=\Pr[Y=1] = 1/2$.
        \item $\Pr[X=Y] \geq 1/2 + \epsilon/2$.
        \item There exists an event $E$ which implies $X=Y$ such that \[
    \Pr[E\mid X=Y]\geq\delta,
    \qquad
    \Pr[X=0\mid E]=\frac12,
    \qquad
    I(X;Z\mid E)=0.
\]
    \end{itemize}
\end{definition}

\begin{definition}[Information-theoretic key agreement protocol for $(\epsilon,\delta)$-secure random variables \cite{Holenstein}]\label{def:info_protocol}
    Let $\epsilon = \epsilon(\secp),\delta = \delta(\secp)$, $n = n(\secp)$, and $\{(X_i,Y_i,Z_i)\}_{i \in [\ell]}$ be $\ell$ independent $(\epsilon,\delta)$-secure random variables. Let $X = \{X_i\}_{i \in [\ell]}$, $Y = \{Y_i\}_{i \in [\ell]}$, $Z = \{Z_i\}_{i \in [\ell]}$. An information-theoretic key agreement protocol for $(\epsilon,\delta)$-secure random variables is a protocol where $A$ has input $(1^\secp,X)$, $B$ has input $(1^\secp,Y)$, and they communicate (classically) to produce a transcript $T$, an output $k_A$ for $A$, and an output $k_B$ for $B$. The guarantee is that $(k_A,k_B,(T,Z))$ is an $(1-2^{-\secp},1-2^{-\secp})$-secure random variable, and the protocol is called \emph{efficient} if the running time for both $A$ and $B$ is $\mathsf{poly}(\secp)$.
\end{definition}

\noindent We will rely on the following claim:

\begin{claim} [\cite{Holenstein}] \label{claim:info_key_agreement_params}
    For any constants $\epsilon,\delta$ such that $\delta > \frac{1-\epsilon}{1+\epsilon}$, there exists an efficient information-theoretic key agreement protocol for $(\epsilon,\delta)$-secure random variables.
\end{claim}

\noindent Next, we state and prove an amplification theorem of weak bit-agreement to key agreement.

\begin{theorem}[Adaptation of Theorem 2.17 from \cite{Holenstein}]\label{thm:wba-to-ka}
    Let $\epsilon(\secp),\delta(\secp) : \bbN \to (0,1)$. If there exists a post-quantum $(\epsilon(\secp),\delta(\secp))$-BA protocol secure against adversaries with non-uniform classical advice, and an efficient information-theoretic key agreement protocol for
$(\epsilon(\secp),1-\delta(\secp)-\tfrac{1}{\poly(\secp)})$-secure random variables, then there exists a $(1,0)$-WBA, and thus a standard (computationally-secure) key agreement protocol.
\end{theorem}

\begin{proof}
Let $\ell = \ell(\secp)$ be the number of invocations of the WBA protocol required by the information-theoretic key agreement protocol, and let $n = n(\secp)$ be the number of bits in the transcript of each WBA protocol. Let $D = \{D_\secp\}_{\secp \in \bbN}$ be the distribution over transcripts and output bits $(Z,X) \in \{0,1\}^n \times \{0,1\}$ in the WBA protocol, conditioned on $A$ and $B$ obtaining the same output bit $X$. Let
\[
    \delta_{\sf hc}(\secp)\ :=\ 1-\delta(\secp).
\]
By the theorem's assumption, there exists an inverse-polynomial function $\kappa(\secp)=1/\poly(\secp)$ such that the information-theoretic key agreement protocol works for
\[
    (\epsilon(\secp),\delta_{\sf hc}(\secp)-\kappa(\secp))
\]
-secure random variables. By decreasing $\kappa$ if necessary, we assume that $\kappa(\secp)\le \delta_{\sf hc}(\secp)/2$ for all sufficiently large $\secp$.

Now consider the following key agreement protocol. $A$ and $B$ first use the WBA protocol $\ell$ times to obtain outputs
\[
    X_1,\dots,X_\ell
    \qquad\text{and}\qquad
    Y_1,\dots,Y_\ell
\]
along with transcripts $Z_1,\dots,Z_\ell$. Then, for every $i\in[\ell]$, party $A$ publicly samples and sends a uniformly random bit $G_i\gets\bits$, and the parties define
\[
    \overline X_i := X_i\oplus G_i,\qquad
    \overline Y_i := Y_i\oplus G_i,\qquad
    \overline Z_i := (Z_i,G_i).
\]
They use $\overline X_1,\dots,\overline X_\ell$ and $\overline Y_1,\dots,\overline Y_\ell$ as inputs to the information-theoretic key agreement protocol and run it, producing a transcript $T$. Their outputs $k_A$ and $k_B$ are the outputs of the information-theoretic key agreement protocol.

The public masks preserve correctness, since $\overline X_i=\overline Y_i$ if and only if $X_i=Y_i$. They also preserve the WBA security of each coordinate: any predictor for $\overline X_i=X_i\oplus G_i$ given $(Z_i,G_i)$ can be converted into a predictor for $X_i$ given $Z_i$ by XORing the predictor's output with $G_i$. Suppose for contradiction that there exists a QPT adversary $\{\Adv_\secp\}_{\secp \in \bbN}$ such that
\[
    \Pr\left[
        \Adv_\secp(\overline Z_1,\dots,\overline Z_\ell,T) = k_A
        \mid k_A=k_B
    \right]
    \ge
    \frac12+\nonnegl(\secp).
\]
Define its conditional key-guessing advantage by
\[
    \beta(\secp)
    :=
    \Pr\left[
        \Adv_\secp(\overline Z_1,\dots,\overline Z_\ell,T) = k_A
        \mid k_A=k_B
    \right]
    -
    \frac12.
\]
By assumption, $\beta(\secp)=\nonnegl(\secp)$. Our goal is to obtain a QPT adversary $\{\Adv'_\secp\}_{\secp\in\bbN}$ such that, for infinitely many $\secp$,
\begin{equation}
    \label{eq:adv_assumption}
    \Pr_{(Z,X)\gets D_\secp}
    \left[
        \Adv'_\secp(Z)=X
    \right]
    \ge
    \frac12+\frac{\delta(\secp)}2+\nonnegl(\secp).
\end{equation}

\paragraph{Min-entropy alignment.}
We denote
\[
    \eta(\secp)
    :=
    \max\{
        \Pr_{(Z,X)\gets D_\secp}[Z=z]
        \mid z\in\{0,1\}^{n(\secp)}
    \}.
\]
If it is not the case that $\eta(\secp)\le\negl(\secp)$, we modify $\Pi_{\mathsf{WBA}}$ so that, at the end of the protocol, one of the parties sends a uniformly random bitstring $U\leftarrow\{0,1\}^{\secp}$, which is then ignored. We include $U$ as part of the WBA transcript. Since $U$ is independent of the parties' outputs, this leaves correctness and security unchanged and ensures that
\[
    \max_{(z,x)}D_\secp(z,x)\le \negl(\secp).
\]
Henceforth, we assume without loss of generality that this min-entropy condition holds.

\paragraph{Parameter translation.} The WBA security guarantee can be rewritten as \[ \Pr\bigl[\widetilde E(\tau)=k_A \mid k_A=k_B\bigr] \le \frac12+\frac{\delta(\secp)}2+\negl(\secp) = 1-\frac{\delta_{\sf hc}(\secp)}2+\negl(\secp). \] 
Define \[ \delta_*(\secp) := \delta_{\sf hc}(\secp)-\frac{\kappa(\secp)}2. \]
Earlier, we defined $\kappa(\secp)\le \delta_{\sf hc}(\secp)/2$ for all sufficiently large $\secp$. Hence $\delta_*(\secp)\in(0,1)$ for all sufficiently large $\secp$.

\paragraph{Fixing the hard-core parameter.} Since $\beta(\secp)=\nonnegl(\secp)$, there exists a polynomial $p(\cdot)$ and infinitely many security parameters $\secp$ such that \[ \beta(\secp)\ge \frac1{p(\secp)}. \] In the remainder of the proof, we restrict attention to this infinite set of security parameters. Define \begin{equation} \label{eq:gamma-def} \gamma(\secp) := \min\left\{ \frac{\delta_*(\secp)}8,\, \frac{\kappa(\secp)}4,\, \frac{1}{8\ell(\secp)p(\secp)} \right\}. \end{equation} Then $\gamma(\secp)=\Omega(1/\poly(\secp))$ on this infinite set, and for all sufficiently large $\secp$ we have \[ 0<\gamma(\secp)\le \delta_*(\secp), \] as required in Corollary~\ref{cor:hardcore-contrap}.

\paragraph{Invoking the hard-core contrapositive.} We will apply Corollary~\ref{cor:hardcore-contrap} to the distribution $D_\secp$ over $(Z,X)$ with density parameter $\delta_*(\secp)$ and hard-core parameter $\gamma(\secp)$. It suffices to show that for every measure $M_\secp\preceq D_\secp$ satisfying \begin{equation} \label{eq:measure-mass-threshold} \mu(M_\secp) \ge \delta_*(\secp)-\frac{\gamma(\secp)}4, \end{equation} there exists a non-uniform QPT predictor $\Adv''_\secp$ such that \begin{equation} \label{eq:measure-predictor-goal} \Pr_{(Z,X)\leftarrow M_\secp} \left[ \Adv''_\secp(Z)=X \right] > \frac12+\gamma(\secp), \end{equation} where $(Z,X)\leftarrow M_\secp$ denotes sampling from the normalized distribution $M_\secp/\mu(M_\secp)$. Indeed, once this is shown, Corollary~\ref{cor:hardcore-contrap} yields a non-uniform QPT predictor $\Adv'_\secp$ such that, for infinitely many $\secp$,

\[
    \Pr_{(Z,X)\gets D_\secp}
    \left[
        \Adv'_\secp(Z)=X
    \right]
    >
    1-\frac{\delta_{*}(\secp)}2
    =
    \frac12+\frac{\delta(\secp)}2+\frac{\kappa(\secp)}4.
\]
contradicting WBA security.

\paragraph{A hybrid argument.}
Fix any such measure $M:=M_\secp\preceq D:=D_\secp$, and define
\[
    \rho(z,x)
    :=
    \begin{cases}
        \frac{M(z,x)}{D(z,x)}, & D(z,x)>0,\\
        0, & D(z,x)=0.
    \end{cases}
\]
Since $M\preceq D$, we have $\rho(z,x)\in[0,1]$ for every $(z,x)$. For $i\in\{0,\dots,\ell\}$, define the hybrid $\cH_i$ as follows. First, run the WBA protocol $\ell$ times to produce
\[
    X_1,\dots,X_\ell,\qquad
    Y_1,\dots,Y_\ell,\qquad
    Z_1,\dots,Z_\ell.
\]
Set initially
\[
    (\widehat X_j,\widehat Y_j):=(X_j,Y_j)
    \qquad\text{for every }j\in[\ell].
\]
For every $j\in[\ell]$, if $X_j=Y_j$, sample a flag bit $B_j$ according to
\[
    \Pr[B_j=1\mid X_j,Y_j,Z_j]=\rho(Z_j,X_j).
\]
If $X_j\ne Y_j$, set $B_j:=0$. For every $j\le i$ such that $B_j=1$, replace $   (\widehat X_j,\widehat Y_j)$ by $(R_j,R_j)$,
where $R_j\gets\bits$ is freshly sampled and uniform. After this replacement step, sample independent public bits
\[
    G_1,\dots,G_\ell\gets\bits,
\]
and define
\[
    \overline X_j:=\widehat X_j\oplus G_j,\qquad
    \overline Y_j:=\widehat Y_j\oplus G_j,\qquad
    \overline Z_j:=(Z_j,G_j).
\]
Run the information-theoretic key agreement protocol on
\[
    \overline X_1,\dots,\overline X_\ell
    \qquad\text{and}\qquad
    \overline Y_1,\dots,\overline Y_\ell
\]
to obtain transcript $T$ and output bits $(k_A,k_B)$. Run
\[
    \Adv_\secp(\overline Z_1,\dots,\overline Z_\ell,T)
\]
to obtain $c$. The hybrid outputs $1$ if and only if
\[
    c=k_A
    \qquad\text{and}\qquad
    k_A=k_B.
\]
We first note that the joint distribution of the input bits
\[
    (\overline X_1,\overline Y_1),\dots,(\overline X_\ell,\overline Y_\ell)
\]
to the information-theoretic key agreement protocol is identical in all hybrids. Indeed, the only replacement occurs when $X_j=Y_j$, and after the public mask $G_j$ is applied, both
\[
    (X_j\oplus G_j,Y_j\oplus G_j)
    \qquad\text{and}\qquad
    (R_j\oplus G_j,R_j\oplus G_j)
\]
are distributed as a uniformly random equal pair. Therefore, if
\[
    p_{\sf agr}:=\Pr[k_A=k_B],
\]
then $p_{\sf agr}$ is the same in all hybrids. The hybrid $\cH_0$ is exactly the real attack experiment for the amplified protocol. Hence, by the contradiction assumption,
\begin{equation}
    \label{eq:h0-prob}
    \Pr[\cH_0=1]
    =
    p_{\sf agr}\cdot
    \left(\frac12+\beta(\secp)\right).
\end{equation}

We next argue that in $\cH_\ell$, the variables
\[
    (\overline X_j,\overline Y_j,\overline Z_j)_{j\in[\ell]}
\]
are independent $  (\epsilon(\secp),\delta_{\sf hc}(\secp)-\kappa(\secp))$-secure random variables.

Fix any $j\in[\ell]$. Since $G_j$ is uniform and independent, we have
\[
    \Pr[\overline X_j=0]
    =
    \Pr[\overline X_j=1]
    =
    \Pr[\overline Y_j=0]
    =
    \Pr[\overline Y_j=1]
    =
    \frac12.
\]
Moreover, equality is preserved by the mask, and the replacement step only occurs when $X_j=Y_j$. Therefore
\[
    \Pr[\overline X_j=\overline Y_j]
    =
    \Pr[X_j=Y_j]
    \ge
    \frac12+\frac{\epsilon(\secp)}2-\negl(\secp).
\]
Define the event
\[
    F_j:=\{X_j=Y_j\ \wedge\ B_j=1\}.
\]
Then $F_j$ implies $\overline X_j=\overline Y_j$. Furthermore, because $\overline X_j=\overline Y_j$ if and only if $X_j=Y_j$, and conditioned on $X_j=Y_j$ the pair $(Z_j,X_j)$ is distributed according to $D$, we have
\[
    \Pr[F_j\mid \overline X_j=\overline Y_j]
    =
    \E_{(Z,X)\gets D}[\rho(Z,X)]
    =
    \mu(M).
\]

By \eqref{eq:measure-mass-threshold} and the definitions of $\delta_*$ and $\gamma$, we have \[ \mu(M) \ge \delta_*(\secp)-\frac{\gamma(\secp)}4 = \delta_{\sf hc}(\secp)-\frac{\kappa(\secp)}2-\frac{\gamma(\secp)}4. \] Since $\gamma(\secp)\le \kappa(\secp)/4$, this implies \[ \mu(M) \ge \delta_{\sf hc}(\secp)-\frac{\kappa(\secp)}2-\frac{\kappa(\secp)}{16} \ge \delta_{\sf hc}(\secp)-\kappa(\secp). \] Thus the event $F_j$ has sufficiently large conditional probability to witness that the variables in $\cH_\ell$ are $(\epsilon(\secp),\delta_{\sf hc}(\secp)-\kappa(\secp))$ -secure random variables. Finally, conditioned on $F_j$, in $\mathcal H_\ell$ we have
\[
    \overline X_j=R_j\oplus G_j,
    \qquad
    \overline Z_j=(Z_j,G_j).
\]
The fresh bit $R_j$ remains uniform and independent of $(Z_j,G_j)$
under this conditioning. Consequently,
\[
    \Pr[\overline X_j=0\mid F_j]=\frac12,
    \qquad
    I(\overline X_j;\overline Z_j\mid F_j)=0.
\]

Thus the variables in $\cH_\ell$ satisfy the input condition for the information-theoretic key agreement protocol. Hence the output triple
\[
    (k_A,k_B,(T,\overline Z_1,\dots,\overline Z_\ell))
\]
is an $(1-2^{-\secp},1-2^{-\secp})$-secure random variable. In particular,
\[
    p_{\sf agr}=\Pr[k_A=k_B]\ge 1-\negl(\secp),
\]
and any predictor of $k_A$ from $(T,\overline Z_1,\dots,\overline Z_\ell)$ has success probability at most $1/2+\negl(\secp)$ conditioned on $k_A=k_B$. Therefore
\begin{equation}
    \label{eq:hl-prob}
    \Pr[\cH_\ell=1]
    =
    p_{\sf agr}\cdot
    \left(\frac12+\negl(\secp)\right).
\end{equation}

Combining \eqref{eq:h0-prob} and \eqref{eq:hl-prob}, we obtain
\[
    \tau(\secp)
    :=
    \Pr[\cH_0=1]-\Pr[\cH_\ell=1]
    =
    p_{\sf agr}\cdot
    \left(\beta(\secp)-\negl(\secp)\right).
\]
Since $p_{\sf agr}\ge 1-\negl(\secp)$, for all sufficiently large $\secp$ in the infinite set under consideration,
\begin{equation}
    \label{eq:tau-lower-bound}
    \tau(\secp)\ge \frac{\beta(\secp)}2.
\end{equation}

Therefore, there exists an index $i=i(\secp)\in[\ell]$ such that
\begin{equation}
    \label{eq:one-step-gap}
    \Pr[\cH_{i-1}=1]-\Pr[\cH_i=1]
    \ge
    \alpha(\secp),
    \qquad
    \text{where }
    \alpha(\secp):=\frac{\beta(\secp)}{2\ell(\secp)}.
\end{equation}

Let $W$ denote the tuple of post-randomization values in all coordinates except $i$:
\[
    W
    :=
    \bigl((\overline X_j,\overline Y_j,\overline Z_j)\bigr)_{j\ne i}.
\]
The distribution of $W$ is the same in $\cH_{i-1}$ and $\cH_i$. By averaging, there exists a value $w=w(\secp)$ in the support of $W$ such that
\begin{equation}
    \label{eq:fixed-w-gap}
    \Pr[\cH_{i-1}=1\mid W=w]
    -
    \Pr[\cH_i=1\mid W=w]
    \ge
    \alpha(\secp).
\end{equation}

\paragraph{Defining $\Adv''_\secp$.}
Hard-code $i=i(\secp)$ and $w=w(\secp)$ into $\Adv''_\secp$. Given input a WBA transcript $z$, the predictor $\Adv''_\secp$ operates as follows.

\begin{enumerate}
    \item Sample independent uniform bits $R_i,G_i\gets\bits$, and set
    \[
        \overline x_i=\overline y_i:=R_i\oplus G_i,
        \qquad
        \overline z_i:=(z,G_i).
    \]
    Combine this coordinate with the hard-coded values $w$ in all coordinates $j\ne i$, and run the information-theoretic key agreement protocol to obtain transcript $T$ and outputs $(k_A,k_B)$.

    \item Query
    \[
        \Adv_\secp(\overline z_1,\dots,\overline z_\ell,T)
    \]
    to obtain $c$.

    \item If $c=k_A$ and $k_A=k_B$, output $R_i$. Otherwise, output $1\oplus R_i$.
\end{enumerate}

Let $s_{\Adv}(\secp)$ be an upper bound on the circuit size of $\Adv_\secp$, and let $s_{\sf KA}(\secp)$ upper bound the circuit size needed to execute the efficient information-theoretic key-agreement protocol once on $\ell(\secp)$ input bits, producing $(T,k_A,k_B)$. Both $s_{\Adv}$ and $s_{\sf KA}$ are polynomial in $\secp$. Hard-coding $w,i$ contributes
\[
    O(\ell(\secp)\cdot(n(\secp)+3))
\]
bits of advice and hence $O(\ell(\secp)\cdot n(\secp))$ circuit-size overhead. Therefore the constructed non-uniform predictor $\Adv''_\secp$ can be implemented by a quantum circuit of size at most
\begin{equation}
    \label{eq:adv-circuit-size}
    t(\secp)
    :=
    s_{\Adv}(\secp)+s_{\sf KA}(\secp)+O(\ell(\secp)\cdot n(\secp))+\poly(\secp)
    =
    \poly(\secp).
\end{equation}

Recall
\[
    F_i=\{X_i=Y_i\ \wedge\ B_i=1\}.
\]
By construction, conditioned on $W=w$, the hybrids $\cH_{i-1}$ and $\cH_i$ are identical conditioned on $\neg F_i$. Conditioned on $F_i$, they differ only in the value fed to the information-theoretic protocol in coordinate $i$: in $\cH_{i-1}$ it is $X_i\oplus G_i$, while in $\cH_i$ it is $R_i\oplus G_i$ for an independent uniform $R_i$. Therefore
\[
\begin{aligned}
    &\Pr[\cH_{i-1}=1\mid W=w]
    -
    \Pr[\cH_i=1\mid W=w] \\
    &\quad =
    \Pr[F_i\mid W=w]\cdot
    \Bigl(
        \Pr[\cH_{i-1}=1\mid W=w,F_i]
        -
        \Pr[\cH_i=1\mid W=w,F_i]
    \Bigr).
\end{aligned}
\]
Combining this identity with \eqref{eq:fixed-w-gap}, and using $\Pr[F_i\mid W=w]\le 1$, yields
\begin{equation}
    \label{eq:conditional-gap}
    \Pr[\cH_{i-1}=1\mid W=w,F_i]
    -
    \Pr[\cH_i=1\mid W=w,F_i]
    \ge
    \alpha(\secp).
\end{equation}

Let $G_{\sf out}$ denote the event that the hybrid outputs $1$, namely
\[
    G_{\sf out}:=\{c=k_A\ \wedge\ k_A=k_B\}.
\]
Then \eqref{eq:conditional-gap} says
\[
    \Pr[G_{\sf out}\mid W=w,F_i,\overline X_i=X_i\oplus G_i]
    -
    \Pr[G_{\sf out}\mid W=w,F_i,\overline X_i\gets R_i\oplus G_i]
    \ge
    \alpha(\secp).
\]
Since $R_i$ is uniform and independent of $X_i$, the second term is the average of the two cases
\[
    \overline X_i=X_i\oplus G_i
    \qquad\text{and}\qquad
    \overline X_i=(1-X_i)\oplus G_i.
\]
Hence
\begin{equation}
    \label{eq:G-gap}
\begin{aligned}
    \Delta_G
    &:=
    \Pr[G_{\sf out}\mid W=w,F_i,\overline X_i=X_i\oplus G_i] \\
    &\quad -
    \Pr[G_{\sf out}\mid W=w,F_i,\overline X_i=(1-X_i)\oplus G_i]
    \ge
    2\alpha(\secp).
\end{aligned}
\end{equation}

Conditioned on $F_i$, the pair $(Z_i,X_i)$ is distributed according to the normalized measure $M$. Indeed, for every $(z,x)$ in the support of $D$,
\[
\begin{aligned}
    \Pr[(Z_i,X_i)=(z,x)\mid F_i]
    &=
    \frac{
        D(z,x)\rho(z,x)
    }{
        \E_{(Z,X)\gets D}[\rho(Z,X)]
    } \\
    &=
    \frac{M(z,x)}{\mu(M)}.
\end{aligned}
\]

Therefore, by construction of $\Adv''_\secp$,
\[
\begin{aligned}
    \Pr_{(Z_i,X_i)\leftarrow M}
    \left[
        \Adv''_\secp(Z_i)=X_i
    \right]
    &=
    \Pr[G_{\sf out}\wedge R_i=X_i\mid W=w,F_i] \\
    &\quad+
    \Pr[\neg G_{\sf out}\wedge R_i=1-X_i\mid W=w,F_i] \\
    &=
    \frac12
    \Bigl(
        \Pr[G_{\sf out}\mid W=w,F_i,R_i=X_i] \\
    &\qquad\qquad+
        \Pr[\neg G_{\sf out}\mid W=w,F_i,R_i=1-X_i]
    \Bigr) \\
    &=
    \frac12(1+\Delta_G) \\
    &\ge
    \frac12+\alpha(\secp).
\end{aligned}
\]
On the infinite set under consideration, $\beta(\secp)\ge 1/p(\secp)$, and hence
\[
    \alpha(\secp)
    =
    \frac{\beta(\secp)}{2\ell(\secp)}
    \ge
    \frac{1}{2\ell(\secp)p(\secp)}
    >
    \gamma(\secp),
\]
by \eqref{eq:gamma-def}. Thus
\[
    \Pr_{(Z_i,X_i)\leftarrow M_\secp}
    \left[
        \Adv''_\secp(Z_i)=X_i
    \right]
    >
    \frac12+\gamma(\secp),
\]
which establishes \eqref{eq:measure-predictor-goal}.

Finally, choose the size bound $s(\secp)$ in Corollary~\ref{cor:hardcore-contrap} so that for all sufficiently large $\secp$, 
\[
    s(\lambda)
    \geq
    q\bigl(\lambda,1/\varepsilon(\lambda)\bigr)
    \left(
        n(\lambda)
        +
        \left\lceil
            \frac{16n(\lambda)}{\gamma(\lambda)^2}
            \bigl(t(\lambda)+1\bigr)
        \right\rceil
    \right)
    +
    r\bigl(\lambda,1/\varepsilon(\lambda)\bigr).
\]
where \[ \varepsilon(\secp) = \frac{\delta_*(\secp)\gamma(\secp)}4. \] Then the resulting parameter $s'(\secp)$ from Corollary~\ref{cor:hardcore-contrap} satisfies $s'(\secp)\ge t(\secp)$, so $\Adv''_\secp$ is a valid witness circuit in the premise of the corollary. Therefore Corollary~\ref{cor:hardcore-contrap}, invoked with density parameter $\delta_*(\secp)$ and hard-core parameter $\gamma(\secp)$, gives a predictor $\Adv'_\secp$ such that, for infinitely many $\secp$, \[ \Pr_{(Z,X)\gets D_\secp} \left[ \Adv'_\secp(Z)=X \right] > 1-\frac{\delta_*(\secp)}2. \] By definition of $\delta_*(\secp)$, \[ 1-\frac{\delta_*(\secp)}2 = 1-\frac{\delta_{\sf hc}(\secp)}2+\frac{\kappa(\secp)}4 = \frac12+\frac{\delta(\secp)}2+\frac{\kappa(\secp)}4. \] Since $\kappa(\secp)=1/\poly(\secp)$, this exceeds the WBA security threshold \[ \frac12+\frac{\delta(\secp)}2+\negl(\secp) \] by an inverse-polynomial amount for all sufficiently large $\secp$ in the infinite set under consideration. This contradicts WBA security.

This contradiction proves that the amplified protocol is computationally secure. Its correctness and lack of bias follow from the guarantee of the information-theoretic key agreement protocol, and hence the protocol is a fully secure bit-agreement protocol. This completes the proof.

\end{proof}

We conclude this section by relating our key-agreement amplification result to non-commutation tests, and prove that non-commutation tests in a certain parameter regime imply key agreement. We also conclude that \emph{any} non-commutation test with robust completeness implies key agreement.

\begin{corollary}
\label{cor:ToNc_to_BA_params_epsdelta}
Given any normal-form $(\epsilon,\delta)$-ToNC with constant parameters
$\epsilon,\delta\in[0,1]$ satisfying
\[
\epsilon \;>\; \frac{1+4\delta}{5},
\]
there exists classical-communication key agreement.
\end{corollary}

\begin{proof}
By slightly decreasing $\epsilon$ if necessary, we may assume
$0<\epsilon<1$ while preserving the strict inequality in the hypothesis.
By \cref{thm:wba_from_test}, the ToNC implies an $(\epsilon,\gamma)$-WBA with
\[
    \gamma := \frac{1+4\delta-3\epsilon}{1+\epsilon}.
\]
\cref{remark:advantage} gives $\delta\ge\epsilon/2$, and hence
\[
    \gamma=\frac{1+4\delta-3\epsilon}{1+\epsilon}
    \ge \frac{1-\epsilon}{1+\epsilon}>0,
\]
where the last inequality uses $\epsilon<1$.
On the other hand, the hypothesis $4\delta<5\epsilon-1$ gives
\begin{equation}
    \label{eq:gamma_inequality}
    \gamma<\frac{2\epsilon}{1+\epsilon}<1.
\end{equation}

Equation \eqref{eq:gamma_inequality} therefore allows us to set
\[
    \eta:=\frac12\left(\frac{2\epsilon}{1+\epsilon}-\gamma\right)>0,
\]
which is a constant satisfying
\[
    1-\gamma-\eta
    =\frac{1-\epsilon}{1+\epsilon}+\eta
    >\frac{1-\epsilon}{1+\epsilon}.
\]

By \cref{claim:info_key_agreement_params}, there exists an efficient information-theoretic
key-agreement protocol for
$(\epsilon,1-\gamma-\eta)$-secure random variables.
For all sufficiently large $\lambda$, we have $1/\lambda\le\eta$,
so any $(\epsilon,1-\gamma-1/\lambda)$-secure random variable
also satisfies this constant-parameter guarantee.
The same protocol therefore meets the efficiency assumption of
\cref{thm:wba-to-ka} with inverse-polynomial slack $1/\lambda$.
Applying that theorem proves the claim.
\end{proof}

\begin{corollary}
    Given any normal-form $(\epsilon,\delta)$-ToNC with robust completeness and constant parameters $\epsilon>\delta \in[0,1]$, there exists classical-communication key agreement.
\end{corollary}

\begin{proof}
By \cref{thm:wba_from_test}, a normal-form $(\epsilon,\delta)$-ToNC with robust completeness implies $(\epsilon,\gamma)$-WBA with $\gamma = 2\delta - \epsilon.$
By \cref{remark:advantage}, we have $\delta\ge\epsilon/2$,
so $\gamma\ge0$. Moreover, the hypothesis $\delta<\epsilon$
and the bound $\epsilon\le1$ give
\[
    \gamma=2\delta-\epsilon
    <\epsilon
    \le\frac{2\epsilon}{1+\epsilon}.
\]
\cref{thm:ba-to-ka-params} therefore yields a $(1,0)$-BA protocol,
and hence classical-communication key agreement.
\end{proof}

\subsection{Post-quantum interactive XOR lemma}\label{subsec:XOR-lemma}

The multiplicative Isolation Lemma---introduced and proved by Levin---is used in his work to prove Yao's XOR lemma~\cite{levinIsolation85}. Halevi and Rabin~\cite[Lemma~3.1]{HR08} extend Levin's proof of Yao's XOR lemma, which relies on the Isolation Lemma, to the setting of classical interactive protocols. In particular, they show that if every computationally bounded adversary can predict the protocol's single secret bit with correlation advantage at most $\varepsilon$, then the $t$-fold sequential XOR-composition amplifies secrecy so that every efficient adversary's advantage drops to roughly $\varepsilon^{t}$ (up to negligible terms).

We extend their result to the post-quantum setting, by proving an analogous theorem which enables amplifying security against computationally-bounded \emph{quantum} adversaries. The main challenge in proving such a theorem, is the reliance of the classical proof on rewinding the adversary, which is non-trivial and at times potentially impossible when the adversary is quantum, due to the no-cloning principle. In order to overcome it, we utilize an approach for quantum rewinding developed by Marriott and Watrous~\cite{marriott2005} and further analyzed by \cite{ChiesaRewinding} (see our Preliminaries section for an introduction to this approach).

\begin{theorem}\label{thm:sequential-rep}
Let $\Gamma_1,\Gamma_2$ be two protocols such that for any $i\in\{1,2\}$, for any QPT adversary $A$ and state $\ket{\psi}$,

\[\E_{\substack{r \gets \{0,1\}^*, \\ b_A, \ \xi \gets \Gamma_i\langle A(\ket{\psi}),V_i(r)\rangle}}\left[b_A \cdot P_i(r, \xi)\right] \leq \delta_i,\] where $V_i$ is the verifier in $\Gamma_i$, $r_i$ are its random bits,\footnote{While this theorem and proof are stated for purely classical verifiers that compute the predicate $P_i$ using their random tape, it is straightforward to extend it to the setting of (efficient) \emph{quantum} verifiers that maintain a quantum state and then compute the predicate from the (classical) transcript and their final state.} $\xi$ is the interaction transcript, $P_i(r, \xi) \in \{-1,+1\}$ is a $\poly(\secp)$-time computable predicate, and $b_A \in \{-1,+1\}$ is the adversary's output.

%, and $\delta_i= 1/\poly(\lambda)$. 

Then for any QPT adversary $A$ and state $\ket{\psi}$,

\[\E_{\substack{r_1,r_2 \gets \{0,1\}^*, \\ \ket{\psi'},\ \xi_1 \gets \Gamma_1\langle A_1(\ket{\psi}),V_1(r_1)\rangle, \\ b_A,\ \xi_2 \gets \Gamma_2\langle A_2(\ket{\psi'}),V_2(r_2)\rangle}}\left[b_A \cdot P_1(r_1,\xi_1) \cdot P_2(r_2,\xi_2)\right] \leq \delta_1\cdot\delta_2 + \negl(\secp),\]

where we view $A$ as a sequential adversary $A=(A_1,A_2)$: first $A_1$ interacts with $V_1(r_1)$ on input $\ket{\psi}$, producing a transcript $\xi_1$ and leaving a residual state $\ket{\psi'}$; then $A_2$, initialized with $\ket{\psi'}$, interacts with $V_2(r_2)$ producing a transcript $\xi_2$ and output bit $b_A$.
\end{theorem}

\begin{proof}
Suppose for contradiction that there exist a QPT adversary $A$ and a quantum state $\ket{\psi}$ such that \[\E_{\substack{r_1,r_2 \gets \{0,1\}^*, \\ \ket{\psi'},\ \xi_1 \gets \Gamma_1\langle A_1(\ket{\psi}),V_1(r_1)\rangle, \\ b_A,\ \xi_2 \gets \Gamma_2\langle A_2(\ket{\psi'}),V_2(r_2)\rangle}}\left[b_A \cdot P_1(r_1,\xi_1) \cdot P_2(r_2,\xi_2)\right] > \delta_1\cdot\delta_2 + \kappa(\secp).\] 

for a function $\kappa = \Omega(1/\poly(\secp))$. We can write this as \begin{equation}\label{eq:expIn}\E_{\substack{r_1 \gets \{0,1\}^* \\ \ket{\psi'}, \ \xi_1 \gets \Gamma_1\langle A_1(\ket{\psi}),V_1(r_1)\rangle}}\left[P_1(r_1,\xi_1) \cdot \E_{\substack{r_2 \gets \{0,1\}^* \\b_A,\ \xi_2 \gets \Gamma_2\langle A_2(\ket{\psi'}),V_2(r_2)\rangle}}\left[b_A \cdot P_2(r_2, \xi_2) \right]\right] > \delta_1\cdot\delta_2 + \kappa(\secp),\end{equation}

Now, for any $\ket{\psi'}$, let \[T(\ket{\psi'}) \coloneqq \E_{\substack{r_2 \gets \{0,1\}^* \\b_A,\ \xi_2 \gets \Gamma_2\langle A_2(\ket{\psi'}),V_2(r_2)\rangle}}\left[b_A \cdot P_2(r_2, \xi_2) \right].\] 

Due to the security guarantee of $\Gamma_2$, for any state $\ket{\psi'}$

\begin{equation} \label{eq:Tbound}
|T(\ket{\psi'})|\le \delta_2.
\end{equation}

Note that if $T(\ket{\psi'})< -\delta_2$ using an adversary $A_2$, we obtain an adversary which breaks the security by flipping the sign of the predictions of $A_2$. This justifies the bound on the \emph{absolute value} of $T(\ket{\psi'})$.

Due to \eqref{eq:expIn}, we have that 
\begin{equation}
\label{eq:advantage}
\E_{\substack{r_1 \gets \{0,1\}^* \\ \ket{\psi'},\ \xi \gets \Gamma_1\langle A_1(\ket{\psi}),V_1(r_1)\rangle}}\left[P_1(r_1, \xi_1) \cdot \frac{T(\ket{\psi'})}{\delta_2} \right] > \delta_1 + \frac{\kappa(\secp)}{\delta_2}.
\end{equation}

To complete the proof, it suffices to give a procedure $B$ that, given $\ket{\psi'}$, outputs a real number $\alpha \in [-\delta_2,\delta_2]$ so that for any $\ket{\psi'}$,
\[\tau(\secp):=\left|\E[B(\ket{\psi'})] - T(\ket{\psi'})\right|\] 

such that 

\begin{equation}
\label{eq:kappa_tau}    
\frac{\kappa-\tau}{\delta_2}=\Omega(\frac{1}{\poly(\lambda)}).\end{equation}
Indeed, given such a $B$, an adversary $A$ could perform the following steps:

\begin{enumerate}
    \item Run $A_1$ in the interaction with $v_1$ and keep the intermediate state $\ket{\psi'}$.
    \item Run $B(\ket{\psi'})$ and obtain $\alpha$.
    \item Output $b_A\in\{1,-1\}$ sampled as a Bernoulli random variable which takes the value $1$ with probability $\frac{\alpha/\delta_2+1}{2}$.
\end{enumerate}
In that case, we would have

\begin{equation}
\begin{aligned}
&\E_{\substack{r \gets \{0,1\}^*, \\ b_A,\ \xi \gets \Gamma_1\langle A ,V_1(r)\rangle}}
  \left[b_A \cdot P(r, \xi)\right] \\
&\quad= \E_{\substack{r \gets \{0,1\}^*, \\ \ket{\psi'},\ \xi \gets \Gamma_1\langle A_1,V_1(r)\rangle}}
  \left[\E_{\substack{\alpha\gets B(\ket{\psi'}),\\ b_A\gets \mathsf{Ber}(\frac{\alpha/\delta_2+1}{2})}}
  [b_A] \cdot P_1(r, \xi)\right] \\
&\quad= \E_{\substack{r_1 \gets \{0,1\}^* \\ \ket{\psi'},\ \xi_1 \gets \Gamma_1\langle A_1,V_1(r_1)\rangle}}
  \left[\frac{\E[B(\ket{\psi'})]}{\delta_2}\cdot P_1(r_1,\xi_1)\right] \\
&\quad= \E_{\substack{r_1 \gets \{0,1\}^* \\ \ket{\psi'},\ \xi_1 \gets \Gamma_1\langle A_1,V_1(r_1)\rangle}}
  \left[\left(\frac{T(\ket{\psi'})}{\delta_2}+\frac{\E[B(\ket{\psi'})]-T(\ket{\psi'})}{\delta_2}\right)\cdot
  P_1(r_1,\xi_1)\right] \\
&\quad\ge \E_{\substack{r_1 \gets \{0,1\}^* \\ \ket{\psi'},\ \xi_1 \gets \Gamma_1\langle A_1,V_1(r_1)\rangle}}
  \left[\frac{T(\ket{\psi'})}{\delta_2}\cdot P_1(r_1,\xi_1)\right]
  \;-\; \frac{1}{\delta_2}\,
  \E_{\substack{r_1 \gets \{0,1\}^* \\ \ket{\psi'},\ \xi_1 \gets \Gamma_1\langle A_1,V_1(r_1)\rangle}}
  \left[\left|\E[B(\ket{\psi'})]-T(\ket{\psi'})\right|\right] \\
&\quad\ge \E_{\substack{r_1 \gets \{0,1\}^* \\ \ket{\psi'},\ \xi_1 \gets \Gamma_1\langle A_1,V_1(r_1)\rangle}}
  \left[\frac{T(\ket{\psi'})}{\delta_2}\cdot P_1(r_1,\xi_1)\right]
  \;-\; \frac{\tau(\secp)}{\delta_2}.
  \end{aligned}
\end{equation}

a contradiction with the hardness of $\Gamma_1$. In the second equality, we rely on the following identity
\[\E_{\substack{\alpha\gets B(\ket{\psi'}),\\ b_A\gets \mathsf{Ber}(\frac{\alpha/\delta_2+1}{2})}}[b_A]=\E_{\substack{\alpha\gets B(\ket{\psi'})}}\left[2\E[\mathsf{Ber}(\frac{\alpha/\delta_2+1}{2})]-1]\right]=\E_{\substack{\alpha\gets B(\ket{\psi'})}}\left[\alpha/\delta_2\right]=\frac{\E[B(\ket{\psi'})]}{\delta_2}\]

Combined with \eqref{eq:advantage},
we get:

\begin{equation}
\E\left[b_A\cdot P_1(r_1,\xi_1)\right]
\;=\;
\E\left[\frac{\E[B(\ket{\psi'})]}{\delta_2}\cdot P_1(r_1,\xi_1)\right]
\;>\;
\delta_1 + \frac{\kappa(\secp)}{\delta_2} - \frac{\tau(\secp)}{\delta_2}=\delta_1+\Omega(\frac{1}{\poly(\lambda)}),
\end{equation}
contradicting the security of $\Gamma_1$. The last equality is due to \eqref{eq:kappa_tau}.

We next define $B$. Let $X$ be the adversary register holding the intermediate state $|\psi'\rangle$ at the start of the
\emph{second} execution. Let $Y$ be an auxiliary register that contains the ancilla qubits to be used by both parties for computations and communication. It includes one designated qubit $Z$
that indicates whether the final outputs match. By purification and deferred measurement, there exists a unitary
\[
U : X\otimes Y \to X\otimes Y
\]
that coherently implements the full second execution of $\Gamma_2$ between $A_2$ and $V_2$.

Define the match projector by
\[
\Pi_{\mathsf{match}} \;:=\; I_{X,Y\setminus Z}\otimes |0\rangle\!\langle 0|_Z,
\]
and define
\[
\Pi_1 \;:=\; U^\dagger \Pi_{\mathsf{match}} U.
\]
Also define the projector onto the clean start state of $Y$:
\[
\Pi_0 \;:=\; I_X \otimes |0\rangle\!\langle 0|_Y.
\]

Fix $\varepsilon,\eta>0$ and let

\[
    n := \left\lceil
        \frac{\ln(2/\eta)}{2\varepsilon^2}
    \right\rceil.
\]

On input $|\psi'\rangle_X$, algorithm $B$:
\begin{enumerate}
  \item Prepares $Y$ in $|0\rangle_Y$, so the joint state is
        \(
        |\Psi_0\rangle = |\psi'\rangle_X|0\rangle_Y \in \mathrm{im}(\Pi_0).
        \)
  \item Performs $n$ alternating binary projective measurements, starting with $\Pi_1$:
        \[
        (\Pi_1,I-\Pi_1),\;(\Pi_0,I-\Pi_0),\;(\Pi_1,I-\Pi_1),\;(\Pi_0,I-\Pi_0),\ldots
        \]
        Let $b_i\in\{0,1\}$ be the outcome bit of the $i$-th measurement.
  \item Sets
        \[
        \widetilde p := \NReps(1,b_1,\ldots,b_{n}),
        \qquad
        \widetilde\alpha := 2\widetilde p - 1.
        \]
  \item Outputs the clipped value
        \[
        \alpha \;:=\; \max\{-\delta_2,\;\min\{\delta_2,\;\widetilde\alpha\}\}\in[-\delta_2,\delta_2].
        \]
\end{enumerate}

Let $q(\psi') := \Pr[b_A=P(r_2, \xi_2)\,:\, \Gamma_2\langle A_2(|\psi'\rangle),V(r_2)\rangle]$ be the
probability in the second execution, so that
\[
T(|\psi'\rangle)\;=\;\mathbb{E}[b_A\cdot P_2(r_2, \xi_2)]\;=\;2q(\psi')-1.
\]
By Proposition \ref{prop:mwDist} applied with $B=(\Pi_0,\mathbb{I}-\Pi_0)$ and $A=(\Pi_1,\mathbb{I}-\Pi_1)$,
the estimator $\widetilde p$ satisfies $\mathbb{E}[\widetilde p]=q(\psi')$, hence
\[
\mathbb{E}[\widetilde\alpha] \;=\; 2\mathbb{E}[\widetilde p]-1 \;=\; 2q(\psi')-1 \;=\; T(|\psi'\rangle).
\]
Furthermore, the hardness of $\Gamma_2$ implies that all eigenvalues of $\Pi_0\Pi_1\Pi_0$ restricted to $\mathrm{im}(\Pi_0)$
lie in $\big[\,\tfrac12-\tfrac{\delta_2} {2},\tfrac12+\tfrac{\delta_2}{2}\,\big]$. Let us show this.  Fix any unit vector $|\phi\rangle \in \mathrm{im}(\Pi_0)$. Since
$\Pi_0 = I_X \otimes |0\rangle\!\langle 0|_Y$, every such vector is of the form
$|\phi\rangle = |\psi'\rangle_X|0\rangle_Y$ for some unit $|\psi'\rangle_X$, and hence
\[
\langle \phi|\Pi_0\Pi_1\Pi_0|\phi\rangle \;=\; \langle \phi|\Pi_1|\phi\rangle.
\]
By construction of $\Pi_1 := U^\dagger \Pi_{\mathsf{match}} U$, the quantity $q(\psi')$ is equal to $\langle \psi'0|\Pi_1|\psi'0\rangle$. Moreover, due to \eqref{eq:Tbound}, $|T(|\psi'\rangle)|\le \delta_2$, and recall that $T(|\psi'\rangle)=2q(\psi')-1$. Hence, $q(\psi')\in[\tfrac12-\tfrac{\delta_2}{2},\tfrac12+\tfrac{\delta_2}{2}]$.
Therefore, for every unit $|\phi\rangle\in\mathrm{im}(\Pi_0)$,
\[
\tfrac12-\tfrac{\delta_2}{2} \;\le\; \langle \phi|\Pi_0\Pi_1\Pi_0|\phi\rangle \;\le\; \tfrac12+\tfrac{\delta_2}{2},
\]
which implies the stated restriction of the eigenvalues.

Therefore, Corollary \ref{cor:jordan-param-eig} implies that every Jordan parameter $p_j$
lies in this interval, so $2p_j-1\in[-\delta_2,\delta_2]$. By Chernoff's bound,
$\Pr[|\widetilde p-p_j|\le \varepsilon]\ge 1-\eta$ in each Jordan block, hence
$|\widetilde\alpha-(2p_j-1)|\le 2\varepsilon$ except with probability $\eta$.
Therefore, the clipping step changes $\widetilde\alpha$ by at most $2\varepsilon$ on the
good event and by at most $2$ always, yielding the expectation bound
\[
\big|\mathbb{E}[\alpha]-\mathbb{E}[\widetilde\alpha]\big|
\;\le\; 2\varepsilon + 2\eta.
\]
Combining the last two displays,
\[
\tau(\secp):=\big|\mathbb{E}[B(|\psi'\rangle)] - T(|\psi'\rangle)\big|
\;=\;
\big|\mathbb{E}[\alpha]-T(|\psi'\rangle)\big|
\;\le\; 2\varepsilon + 2\eta.
\]
Choosing, for example, $\varepsilon(\secp),\eta(\secp):=\frac{\kappa(\secp)}{8}$, implies $\frac{\kappa-\tau}{\delta_2}=\Omega(\frac{1}{\poly(\secp)})$ as required.
\end{proof}

\begin{remark}
    We stress that our Theorem holds in the non-uniform computational model, where the security of the protocol holds against any QPT adversary with non-uniform \emph{quantum} advice. Levin's Isolation Lemma is also proven in the non-uniform model, and Goldreich et al.\ extend it to the uniform model under an assumption on sampling~\cite{GoldreichOnYao2011}. We leave the question of extending our result to the uniform model for future work.
\end{remark}

Using the estimator construction from the proof of
Theorem~6.16, we obtain the following sequential composition
bound.

\begin{corollary}\label{cor:xor-amplification}
Let $\{\Gamma_i\}_{i \in [\ell]}$ be $\ell = \poly(\secp)$ many protocols such that for each $i\in [\ell]$, any QPT adversary $A$, and any state $\ket{\psi}$,

\[\E_{\substack{r \gets \{0,1\}^*, \\ b_A, \ \xi \gets \Gamma_i\langle A(\ket{\psi}),V_i(r)\rangle}}\left[b_A \cdot P_i(r, \xi)\right] \leq \delta_i,\] where $V_i$ is the verifier in $\Gamma_i$, $r$ are its random bits, $\xi$ is the interaction transcript, $P_i(r, \xi) \in \{-1,+1\}$ is a $\poly(\secp)$-time computable predicate, and $b_A \in \{-1,+1\}$ is the adversary's output.

%, and $\delta_i= 1/\poly(\lambda)$. 

Then for any QPT adversary $A$ and state $\ket{\psi}$,

\[\E_{\substack{\{r_i \gets \{0,1\}^*\}_{i \in [\ell]}, \\ \ket{\psi'},\ \xi_1 \gets \Gamma_1\langle A_1(\ket{\psi}),V_1(r_1)\rangle, \\ \vdots \\ b_A,\ \xi_\ell \gets \Gamma_\ell\langle A_\ell(\ket{\psi'}),V_\ell(r_\ell)\rangle}}\left[b_A \cdot P_1(r_1,\xi_1) \cdot \dots \cdot P_\ell(r_\ell,\xi_\ell)\right] \leq \prod_{i \in [\ell]}\delta_i + \negl(\secp),\]

where we view $A$ as a sequential adversary $A=(A_1,\dots,A_\ell)$: first $A_1$ interacts with $V_1(r_1)$ on input $\ket{\psi}$, producing transcript $\xi_1$ and leaving a residual state $\ket{\psi'}$; this repeats $\ell$ times sequentially until $A_\ell$, initialized with $\ket{\psi'}$, interacts with $V_\ell(r_\ell)$ producing transcript $\xi_\ell$ and final output bit $b_A$.

\end{corollary}

\begin{proof}(Sketch)
    Fix a non-uniform QPT adversary adversary $A$, and let $a_i$ denote its maximum absolute correlation with $\prod_{j=i}^{\ell}P_j$ over all input states at the start of execution $i$, with $a_{\ell+1} \coloneqq 1$. If $a_1>\prod_i\delta_i+\kappa$ for non-negligible $\kappa$, a hybrid argument identifies an index $i$ with $a_i>\delta_i a_{i+1}+\mathsf{non}\text{-}\mathsf{negl}(\kappa)$. Applying the estimation-and-rescaling argument from the proof of \cref{thm:sequential-rep} at index $i$ then yields a correlation $> \delta_i + \mathsf{non}\text{-}\mathsf{negl}(\kappa)$ with $P_i$, contradicting the security of $\Gamma_i$.
\end{proof}

\subsection{Oblivious transfer amplification}\label{subsec:OT-amp}

First, we give an OT amplification theorem that follows fairly immediately from \cref{cor:xor-amplification}.

\begin{theorem}
    Given any $(1,\delta,0)$-OT or $(1,0,\delta)$-OT secure against adversaries with non-uniform quantum advice (\cref{def:standard-OT}) for $\delta \leq 1-1/\poly(\secp)$, there exists OT secure against adversaries with non-uniform quantum advice. 
\end{theorem}

\begin{proof}
    Note that $(1,0,\delta)$-weak OT implies $(1,\delta,0)$-weak OT by standard OT reversal. So let $S',R'$ be an $(1,\delta,0)$-weak OT. We perform standard sequential repetition to achieve the standard OT:\\
    
    \noindent $\underline{(1,0,0) \ \OT}$
    \begin{itemize}
        \item Let $\ell = \lceil\secp/(1-\delta)\rceil$. For $i \in [\ell]$, the parties run \[(b_i,r_i),(r_{0,i},r_{1,i}) \gets \langle R'(1^\secp),S'(1^\secp)\rangle.\]
        \item $R$ samples $b \gets \{0,1\}$, sends $\{b \oplus b_i\}_{i \in [\ell]}$ to $S$, and outputs $r_b \coloneqq \bigoplus_{i \in [\ell]} r_i$.
        \item $S$ receives $\{b'_i\}_{i \in [\ell]}$ and outputs \[r_0 \coloneqq \bigoplus_{i \in [\ell]} r_{b'_i,i}, \quad r_1 \coloneqq \bigoplus_{i \in [\ell]} r_{1-b'_i,i}.\]
    \end{itemize}

    Correctness follows by inspection, and receiver security follows from a standard hybrid argument. Sender security follows from \cref{cor:xor-amplification}. Indeed, we fix each of $\Pi_1,\dots,\Pi_\ell$ to be the interactive protocol between honest sender $S'$ and adversarial receiver, where the predicate $P$ outputs the XOR of the sender outputs $r_{0,i} \oplus r_{1,i}$. Then note that $\delta^{\ell} + \negl(\secp) \leq e^{-\secp} + \negl(\secp) =  \negl(\secp)$.

\end{proof}

Next, we amplify our committed-bit OT to the point where we can apply the above theorem. As discussed in \cref{subsec:OT-amp-overview}, our strategy is inspired by a strategy from \cite{10.1007/978-3-540-72540-4_32} in the classical semi-honest case.

\begin{lemma}
    Given any $(\epsilon,\delta,0)$ committed-bit OT (\cref{def:committed-bit}) for constants $\epsilon,\delta$ such that $\delta < \epsilon^2$, there exists a $(1,1/2,0)$ committed-bit OT.
\end{lemma}

\begin{proof}
    Let $\Setup',\OT',\Ver',\TrapGen',\Ext'$ be the committed-bit OT. Set $\ell = \lceil\log_{\epsilon^2 / \delta}(2\secp)\rceil$, and let $t$ be the closest odd integer to $\frac{1}{4\delta^\ell}$. Note that $\ell = O(\secp)$ and $t = \poly(\secp)$. Our protocol between receiver $R$ and sender $S$ is defined as follows.\\

    \noindent\underline{$(1,1/2,0)$ committed-bit OT}

    \begin{itemize}
        \item $\Setup\langle R(1^\secp),S(1^\secp)\rangle$:
        \begin{itemize}
            \item For $i \in [\secp], j \in [t], k \in [\ell]$, the parties run $\init_{R,i,j,k},\tau_{\Setup,i,j,k} \gets \Setup'\langle R'(1^\secp),S'(1^\secp)\rangle$.
            \item Set $\init_R \coloneqq \{\init_{R,i,j,k}\}_{i,j,k}$.
        \end{itemize}

        \item $\OT\langle R(b,\init_R),S(\tau_{\mathsf{Setup}})\rangle$:
        \begin{itemize}
            \item For $i \in [\secp]$:
            \begin{itemize}
                \item $R$ samples $b_i \gets \{0,1\}$.
                \item For $j \in [t],k \in [\ell]$, the parties run \[(r_{i,j,k},\state_{R,i,j,k}),(r_{i,j,k,0},r_{i,j,k,1},\state_{S,i,j,k}),\tau_{i,j,k} \gets \OT'\langle R'(b_i,\init_{R,i,j,k}),S'(\tau_{\Setup,i,j,k})\rangle.\]
            \end{itemize}
            \item $S$ samples $i^* \gets [\secp]$. For $i \in [\secp] \setminus \{i^*\}$:
            \begin{itemize}
                \item $R$ sends $b_i$.
                \item For $j \in [t]$, $k \in [\ell]$, the parties run \[\Ver'\langle R'(\state_{R,i,j,k}),S'(b_i,\state_{S,i,j,k})\rangle,\] and $S$ aborts if any of the results are $\bot$.
            \end{itemize}
            \item $R$ sends $b' \coloneqq b \oplus b_{i^*}$.
            \item $S$ samples $r_0,r_1 \gets \{0,1\}$, sends \[\left\{r'_{j,b'} \coloneqq r_0 \oplus \bigoplus_{k \in [\ell]}r_{i^*,j,k,b'}, \quad r'_{j,1 \oplus b'} \coloneqq r_1 \oplus \bigoplus_{k \in [\ell]}r_{i^*,j,k,1\oplus b'}\right\}_{j \in [t]}\] to $R$ and outputs $r_0,r_1$ and $\state_S \coloneqq b',i^*,\{\state_{S_{i^*,j,k}}\}_{j,k}$.
            \item $R$ outputs
        \[r \coloneqq \maj\left\{r'_{j,b_{i^*}} \oplus \bigoplus_{k \in [\ell]} r_{i^*,j,k}\right\}_{j \in [t]}\] and $\state_R \coloneqq \{\state_{R_{i^*,j,k}}\}_{j,k}$.
        \end{itemize}
        
    \item $\Ver\langle R(\state_R),S(b,\state_S)\rangle$:
    \begin{itemize}
        \item For $j \in [t], k \in [\ell]$, the parties run \[\Ver'\langle R'(\state_{R,i^*,j,k}),S'(b \oplus b',\state_{S,i^*,j,k})\rangle,\] and output $\top$ only if all are $\top$.
    \end{itemize}

    \item $\TrapGen(\tau_\Setup)$: For $i \in [\secp], j \in [t], k \in [\ell]$, run $\td_{i,j,k} \coloneqq \TrapGen'(\tau_{\Setup,i,j,k})$ and output $\td \coloneqq \{\td_{i,j,k}\}_{i,j,k}$.

    \item $\Ext(\td,\tau)$: For $j \in [t], k \in [\ell]$, run $b_{j,k} \coloneqq \Ext'(\td_{i^*,j,k},\tau_{i^*,j,k})$. If there exists $b^*$ such that $b_{j,k} = b^*$ for all $j,k$, then output $b^* \oplus b'$, and otherwise output 0. 
    
    \end{itemize}

\paragraph{Correctness.} Fix $i^*$ and a value of
$b_{i^*}$. All probabilities and expectations below are taken
under this conditioning. Define
\[
E_{j,k}:=r_{i^*,j,k}\oplus r_{i^*,j,k,b_{i^*}},
\qquad
E_j:=\bigoplus_{k\in[\ell]}E_{j,k}.
\]
By correctness of $\mathsf{OT}'$,
\[
\mathbb{E}[(-1)^{E_{j,k}}]
\ge \epsilon-\operatorname{negl}(\lambda).
\]
The honest executions are independent under this conditioning,
so
\[
\Pr[E_j=0]
=
\frac12+\frac12
\prod_{k\in[\ell]}
\mathbb{E}[(-1)^{E_{j,k}}]
\ge
\frac12+\frac{\epsilon^\ell}{2}
-\operatorname{negl}(\lambda).
\]
Moreover, the variables $E_j$ are independent. Thus,
Hoeffding's inequality gives
\[
\Pr\!\left[
  \operatorname{maj}_{j\in[t]} E_j=1
\right]
\le
\exp\!\left(-\frac{t\epsilon^{2\ell}}{2}\right)
+\operatorname{negl}(\lambda)
=
\operatorname{negl}(\lambda),
\]
where the last equality follows from
\[
t\epsilon^{2\ell}
=
\Omega\!\left((\epsilon^2/\delta)^\ell\right)
=
\Omega(\lambda).
\]
By construction, the $j$'th candidate bit computed by the
receiver equals $r_b\oplus E_j$, and their majority equals $r_b$ whenever
$\operatorname{maj}_{j\in[t]}E_j=0$.
Finally, the preliminary verification checks reject with
negligible probability by completeness of $\mathsf{Ver}'$.
Averaging over $i^*$ and $b_{i^*}$, and accounting for this
rejection probability, proves that $r=r_b$ with probability
$1-\operatorname{negl}(\lambda)$.

    % First, we argue correctness. By correctness of $\OT'$, we have that for any fixed $i^*,j,k$,
        %\[\Pr[\forall k \in [\ell], \ r_{i^*,j,k} = r_{i^*,j,k,b_i}] \geq \epsilon^\ell - \negl(\secp).\]
     % Thus, by a Hoeffding inequality, we have that 
     %\[\Pr\left[\maj\left\{\bigoplus_{k \in [\ell]}r_{i^*,j,k} \oplus \bigoplus_{k \in [\ell]}r_{i^*,j,k,b_i}\right\} = 0\right] \geq 1-e^{-2(\epsilon^\ell - \negl(\secp))^2 t} \geq 1-e^{-(\epsilon^2/\delta)^\ell}-\negl(\secp) = 1-\negl(\secp).\] By construction, this event implies that $r = r_{b}$, which completes the proof.

     \paragraph{Completeness.} Follows directly from the completeness of $\OT'$.
     
     \paragraph{Receiver security.} Follows directly from the receiver security of $\OT'$ via a standard hybrid argument. 

    \paragraph{Sender security: Privacy.} We define $\Hyb_\iota$ for each $\iota = 0,\dots,t$, as follows. Here, $(\widetilde{R}_\init,R)$ is the adversarial receiver, and we describe their interaction with a challenger running a variant of the honest sender algorithms. To avoid notational clutter, we don't explicitly notate the adversary's state as it evolves throughout this process. \\

     \noindent\underline{$\Hyb_\iota$}

     \begin{itemize}
         \item For $i \in [\secp], j \in [t], k \in [\ell]$, run $\tau_{\Setup,i,j,k} \gets \Setup'\langle \widetilde{R}_\init,S'(1^\secp)\rangle$ and $\td_{i,j,k} \coloneqq \TrapGen'(\tau_{\Setup,i,j,k})$.
        \item For $i \in [\secp], j \in [t], k \in [\ell]$, run $(r_{i,j,k,0},r_{i,j,k,1},\state_{S,i,j,k}),\tau_{i,j,k} \gets \OT'\langle \widetilde{R},S'(\tau_{\Setup,i,j,k})\rangle$ and $b_{i,j,k} \coloneqq \Ext'(\td_{i,j,k},\tau_{i,j,k})$.
        \item Sample $i^* \gets [\secp]$. For $i \in [\secp] \setminus \{i^*\}$:
        \begin{itemize}
            \item Receive $b_i$ from $\widetilde{R}$.
            \item For $j \in [t]$, $k \in [\ell]$, run $\Ver'\langle \widetilde{R},S'(b_i,\state_{S,i,j,k})\rangle,$ and abort if any of the results are $\bot$.
        \end{itemize}
        \item Receive $b'$ from $\widetilde{R}$.
        \item For each $j \in [t],k \in [\ell]$, run $b_{j,k} \coloneqq \Ext'(\td_{i^*,j,k},\tau_{i^*,j,k})$. If there does not exist $b^*$ such that $b_{j,k} = b^*$ for all $j,k$, then abort (and the adversary wins with probability exactly 1/2, i.e. 0 advantage). Otherwise, set $b = b' \oplus b^*$.
        \item Sample $r_b \gets \{0,1\}$, set $r_{b,j} \coloneqq r_b$ for all $j \in [t]$, set $r_{1\oplus b,j} \coloneqq 1$ for all $j \leq \iota$, and set $r_{1\oplus b,j} \coloneqq 0$ for all $j > \iota$.
        \item Send \[\left\{r'_{j,b'} \coloneqq r_{0,j} \oplus \bigoplus_{k \in [\ell]}r_{i^*,j,k,b'}, \quad r'_{j,1 \oplus b'} \coloneqq r_{1,j} \oplus \bigoplus_{k \in [\ell]}r_{i^*,j,k,1\oplus b'}\right\}_{j \in [t]}\] and $r_b$ to $R$.
     \end{itemize}

     %Now, note that $\Hyb_0$ is the same as the sender security game for $(\widetilde{R}_\init,\widetilde{R})$ conditioned on:
     %\begin{itemize}
         %\item $r_b = 0$, and 
         %\item there exists $b^*$ such that $b_{j,k} = b^*$ for all $j,k$ (and the sender did not abort earlier).
     %\end{itemize}

     %Moreover, $\Hyb_t$ is the same as the sender security game for $(\widetilde{R}_\init,\widetilde{R})$ conditioned on:
     %\begin{itemize}
         %\item $r_b = 1$, and 
         %\item there exists $b^*$ such that $b_{j,k} = b^*$ for all $j,k$ (and the sender did not abort earlier).
     %\end{itemize}

     Consider the modified sender security game that additionally aborts (with zero advantage) whenever the extracted bits in the unopened block are not all equal. Then $\mathsf{Hyb}_0$ and $\mathsf{Hyb}_t$ are exactly this modified game with the hidden bit $r_{1\oplus b}$ fixed to $0$ and $1$, respectively.

     Note that the event that there does \emph{not} exist $b^*$ such that $b_{j,k} = b^*$ for all $j,k$, but the sender did not abort earlier occurs with probability at most $1/\secp + \negl(\secp)$ due to the sender security (in particular the second part) of $\OT'$. Thus, letting $p_i$ be the probability that $(\widetilde{R}_\init,\widetilde{R})$ outputs 0 in $\Hyb_i$, we have that the receiver's advantage in the sender security game is bounded by $|p_0 - p_t| + o(1)$.

     To complete the proof, we argue that for any $\iota \in [t]$, $|p_{\iota-1} - p_\iota| \leq \delta^\ell + \negl(\secp)$, as this implies that $|p_0 - p_t| \leq t \cdot \delta^\ell + \negl(\secp) \leq 1/4 + o(1)$. This follows by a reduction to \cref{cor:xor-amplification}. Indeed, we note that the only difference between $\Hyb_{\iota-1}$ and $\Hyb_\iota$ is the bit $r_{1\oplus b,\iota}$, which is masked by $\bigoplus_{k \in [\ell]} r_{i^*,\iota,k,1 \oplus b^*}$ since $b=b' \oplus b^*$. Thus, it suffices to argue that an appropriate $\ell$-wise sequential repetition of the sender security game of $\OT'$ yields advantage $\delta^\ell + \negl(\secp)$. In particular, it suffices to show that adversary has $\delta^\ell + \negl(\secp)$ advantage in the following game.

     \begin{itemize}
         \item For $k \in [\ell]$, run $\tau_{\Setup,k} \gets \Setup'\langle \widetilde{R}_\init,S'(1^\secp)\rangle$ and $\td_k \gets \TrapGen'(\tau_{\Setup,k})$.
         \item For $k \in [\ell]$, run $(r_{k,0},r_{k,1}),\tau_k \gets \OT'\langle\widetilde{R},S'(\tau_{\Setup,k})\rangle$, $b_k \coloneqq \Ext'(\td_k,\tau_k)$, and send $r_{k,b_k}$ to $\widetilde{R}$.
         \item $\widetilde{R}$ outputs a bit $r^*$ and wins if $r^* = \bigoplus_{k \in [\ell]} r_{k,1\oplus b_k}$.
     \end{itemize}

     By the sender security of $\OT'$, we know that with $1-\negl(\secp)$ probability over \[\left\{\tau_{\Setup,k} \gets \Setup'\langle \widetilde{R}_\init,S'(1^\secp)\rangle, \quad \td_k \coloneqq \TrapGen'(\tau_{\Setup,k})\right\}_{k \in [\ell]},\] $\widetilde{R}$ has advantage at most $\delta + \negl(\secp)$ in guessing $r_{k,1\oplus b_k}$ for any fixed $k \in [\ell]$, even given arbitrary inefficiently-computable advice about $\tau_{\Setup,1},\dots,\tau_{\Setup,\ell}$. Thus, it suffices to show that for any such $\tau_{\Setup,1},\dots,\tau_{\Setup,\ell}$, $\widetilde{R}$ has advantage at most $\delta^\ell + \negl(\secp)$ in guessing $\bigoplus_{k \in [\ell]} r_{k,1\oplus b_k}$. This follows directly from \cref{cor:xor-amplification}, by setting each $\Gamma_k$ to run $(r_{k,0},r_{k,1}),\tau_k \gets \OT'\langle \widetilde{R},S'(\tau_{\Setup,k})\rangle$, set $b_k \coloneqq \Ext'(\td_k,\tau_k)$, and send $r_{k,b_k}$, and defining the predicate $P_k = (-1)^{r_{k,1\oplus b_k}}$.

     \paragraph{Sender security: Binding.} Follows directly from the binding of $\OT'$.

\end{proof}

\appendix

\section{Generalizing the ToNC Definition}
\label{sec:generalized-tonc}

The qubit-certification protocol of Brakerski et al.\
does not satisfy Definition~\ref{def:ToNC}: its
final-round responses are polynomial-length strings, whereas
Definition~\ref{def:ToNC} requires binary-outcome measurements.
In this section we introduce a multi-output generalization of Definition \ref{def:ToNC}
that still remains useful for building bit-agreement.

For one distinguished challenge, which we label $c=1$, the prover's
answer has the form $(d,b)$, where $d$ is a polynomial-length string
and $b$ is a single bit.  We require that $d$ can be efficiently simulated from the
public transcript. With overwhelming probability over the honest
preparation and response conditioned on $c=1$, the verifier's state
and $d$ determine a unique accepting bit. If no unique accepting bit
exists, the derived bit-agreement protocol aborts the attempt.
We impose no analogous restriction on the response to $c=0$. Notice that normal-form test of non-commutation is the special case where \emph{both} $c=0$ and $c=1$ are structured and the string $d$ is empty (so trivially simulatable). 

We show below that this primitive implies bit agreement with explicit
parameters, and hence implies key agreement whenever these parameters
fall in the range of our bit-agreement amplification theorem.  In
particular, the qubit certification protocol of Brakerski et al.\
satisfies the required condition with substantial slack. We leave an investigation of the relationship between this generalized notion of ToNC and oblivious transfer to future work.

\begin{definition}[Generalized test of non-commutation]
\label{def:generalized-ToNC}
An $(\alpha,\beta)$ classical generalized test of non-commutation
is a two-stage interactive protocol between a QPT uniform prover $P$
and a PPT uniform verifier $V$.

The first stage is defined by interactive strategies $P_\prep$ and
$V_\prep$, which we denote by
\[
  \bigl(\ket{\psi},(c,\state),\tau\bigr)
  \gets
  \left\langle
    P_\prep(1^\secp),V_\prep(1^\secp)
  \right\rangle .
\]

Here $(c,\state)$ is the verifier's output, with $c\in\bits$,
and $\ket{\psi}$ is the prover's state immediately before $c$
is sent. We write $\tau=(\tau_{\mathrm{pre}},c)$, where
$\tau_{\mathrm{pre}}$ is the preceding interaction transcript.
All prover computation after receiving $c$ belongs to $P_c$,
acting on the pre-challenge prover space; this convention
also applies to cheating strategies.

There exist polynomials $s,t$ such that the answer alphabets are
\[
  \mathcal A_{0,\secp}=\{0,1\}^{s(\secp)}
  \qquad\text{and}\qquad
  \mathcal A_{1,\secp}=\{0,1\}^{t(\secp)}\times\bits .
\]
The second stage is defined by efficiently implementable projective
measurements
\[
  P_c=\bigl\{P_c^a\bigr\}_{a\in\mathcal A_{c,\secp}},
  \qquad c\in\bits,
\]
and by a verification predicate
$V_\out(\state,a)\to\bits$. All private randomness used in
verification is included in $\state$, so that this predicate is
deterministic given its inputs. The following properties are required.

\begin{itemize}
  \item \textbf{Completeness.}
  It holds that
  \[
    \Pr\left[
      V_\out(\state,a)=1 :
      \begin{array}{l}
        \ket{\psi},(c,\state),\tau
        \gets
        \langle P_\prep(1^\secp),V_\prep(1^\secp)\rangle,\\
        a\gets P_c(\ket{\psi})
      \end{array}
    \right]
    \geq \beta .
  \]

  \item \textbf{Soundness.}
  For every non-uniform QPT strategy
  \[
    \rho,\widetilde P_\prep,
    \widetilde P_0,\widetilde P_1,
  \]
  where
  $\widetilde P_c
    =\{\widetilde P_c^a\}_{a\in\mathcal A_{c,\secp}}$
  and
  \[
    \widetilde P_0^{a_0}\widetilde P_1^{a_1}
    =
    \widetilde P_1^{a_1}\widetilde P_0^{a_0}
    \qquad
    \text{for all }
    a_0\in\mathcal A_{0,\secp},
    \ a_1\in\mathcal A_{1,\secp},
  \]
  it holds that
  \[
    \Pr\left[
      V_\out(\state,a)=1 :
      \begin{array}{l}
        \ket{\psi},(c,\state),\tau
        \gets
        \langle
          \widetilde P_\prep(1^\secp,\rho),
          V_\prep(1^\secp)
        \rangle,\\
        a\gets\widetilde P_c(\ket{\psi})
      \end{array}
    \right]
    \leq \alpha+\negl(\secp).
  \]

  \item \textbf{Unique accepting bit except with negligible probability.}
  For every verifier state $\state$ that can occur with $c=1$ and
  every $d\in\{0,1\}^{t(\secp)}$, define
  \[
    b^\star(\state,d):=
    \begin{cases}
      b, & \text{if $b$ is the unique bit satisfying
                   $V_\out(\state,(d,b))=1$,}\\
      \bot, & \text{otherwise.}
    \end{cases}
  \]
  This function is efficiently computable by evaluating $V_\out$
  on both candidate bits. For the honest distribution
  $\mathsf D^{\mathrm{real}}_\secp$ defined below, we require
  \[
    \Pr_{(\state,\tau,d)\gets\mathsf D^{\mathrm{real}}_\secp}
    [b^\star(\state,d)=\bot]
    \leq\negl(\secp).
  \]
  No restriction is imposed on the verification outcomes when
  $b^\star(\state,d)=\bot$.

  \item \textbf{Simulatability of $d$.}
  There exists a potentially non-uniform QPT algorithm
  $\mathsf{Sim}$ such that the following two ensembles have negligible
  total-variation distance:
  \[
  \begin{aligned}
    \mathsf D^{\mathrm{real}}_\secp
    &:=
    \left\{
      (\state,\tau,d):
      \begin{array}{l}
        (\ket{\psi},(c,\state),\tau)
        \gets
        \langle P_\prep(1^\secp),V_\prep(1^\secp)\rangle
        \text{ conditioned on }c=1,\\
        (d,b)\gets P_1(\ket{\psi})
      \end{array}
    \right\},\\[1ex]
    \mathsf D^{\mathrm{sim}}_\secp
    &:=
    \left\{
      (\state,\tau,d):
      \begin{array}{l}
        (\ket{\psi},(c,\state),\tau)
        \gets
        \langle P_\prep(1^\secp),V_\prep(1^\secp)\rangle
        \text{ conditioned on }c=1,\\
        d\gets\mathsf{Sim}(1^\secp,\tau)
      \end{array}
    \right\}.
  \end{aligned}
  \]
  Namely,
  \[
    \Delta_{\mathrm{TV}}
    \bigl(
      \mathsf D^{\mathrm{real}}_\secp,
      \mathsf D^{\mathrm{sim}}_\secp
    \bigr)
    \leq \negl(\secp).
  \]
  Thus, the simulator sees the public transcript $\tau$, but not the
  verifier's private state $\state$.  The simulator may use the same
  type of non-uniform advice allowed in the soundness condition.
  By total-variation closeness, the event $b^\star(\state,d)=\bot$
  also has negligible probability under $\mathsf D^{\mathrm{sim}}_\secp$.
\end{itemize}
\end{definition}

\paragraph{Bit agreement from generalized ToNC.}
Fix the honest strategy in
Definition~\ref{def:generalized-ToNC}, and write
\[
  q:=\Pr[c=1]
\]
and
\[
  p_c:=
  \Pr\left[
    V_\out(\state,a)=1
    \,\middle|\,
    c
  \right],
  \qquad c\in\bits,
\]
where $a$ is obtained by applying the honest measurement $P_c$.
We suppress the dependence of these quantities on the security
parameter.

The bit-agreement protocol is as follows.

\begin{enumerate}
  \item Alice and Bob repeatedly execute the preparation stage, with
  Alice playing the verifier and Bob playing the prover, using
  fresh independent preparations until an execution with challenge
  $c=1$ occurs. They stop after a polynomially large cutoff;
  as shown below, its failure probability is negligible whenever
  $\beta-\alpha\geq 1/\poly(\secp)$.

  \item In the selected execution, Bob measures
  \[
    (d,b)\gets P_1(\ket{\psi})
  \]
  and sends only $d$ to Alice.  Alice computes
  \[
    b^\star:=b^\star(\state,d).
  \]
  If $b^\star=\bot$, she sends an abort flag and they use the
  fallback specified below. Otherwise, she sends a continue flag.

  \item Alice samples a uniformly random bit $r$, sends it to Bob,
  and outputs
  \[
    k_A:=b^\star\oplus r.
  \]
  Bob outputs
  \[
    k_B:=b\oplus r.
  \]
\end{enumerate}

If the cutoff is reached or Alice aborts the selected attempt,
Alice samples and sends a uniformly random bit, which both parties
output. This fallback ensures that the protocol always outputs bits;
it occurs only with negligible probability.
\begin{proposition}
\label{prop:generalized-ToNC-to-BA}
Let $g:=\beta-\alpha$ and assume $g\geq 1/\poly(\secp)$. Then
\[
  q\geq 2g-\negl(\secp).
\]
Moreover, the protocol above gives weak bit agreement with correctness
advantage at least
\[
  \epsilon_{\mathrm{BA}}=2p_1-1
\]
and eavesdropper advantage at most
\[
  \delta_{\mathrm{BA}}
  \leq
  1-\frac{2g}{q p_1}+\negl(\secp).
\]
Consequently, for constant parameters, the bit-agreement amplification
theorem from Section~6.2 gives key agreement whenever
\begin{equation}
  \frac{2(\beta-\alpha)}{q}>1-p_1.
  \label{eq:generalized-ToNC-amplification-condition}
\end{equation}

\end{proposition}

\begin{proof}
We first lower-bound $q$. Consider a prover that uses the honest
preparation strategy and the honest measurement $P_0$ on challenge
$0$, but on challenge $1$ samples
$d\gets\mathsf{Sim}(1^\secp,\tau)$ and an independent uniform bit $b$.
The challenge-$1$ strategy ignores the prover's quantum state and
therefore commutes with $P_0$. Under the simulated distribution,
$b^\star(\state,d)=\bot$ has negligible probability. On its
complement, a uniform bit is accepted with probability exactly $1/2$.
Thus the conditional success probability is
$1/2\pm\negl(\secp)$, and soundness gives
\begin{equation}
  (1-q)p_0+\frac q2
  \leq \alpha+\negl(\secp).
  \label{eq:generalized-ToNC-soundness-random-bit}
\end{equation}
On the other hand, completeness gives
\begin{equation}
  (1-q)p_0+q p_1\geq\beta.
  \label{eq:generalized-ToNC-completeness}
\end{equation}
Subtracting~\eqref{eq:generalized-ToNC-soundness-random-bit} from~\eqref{eq:generalized-ToNC-completeness}, we obtain
\begin{equation}
  q\left(p_1-\frac12\right)
  \geq g-\negl(\secp),
  \label{eq:generalized-ToNC-gap}
\end{equation}
and hence $q\geq 2g-\negl(\secp)$.  Thus a polynomial number
of repetitions suffices to obtain a challenge-$1$ execution except
with negligible probability.

For correctness, let
\[
  \nu:=\Pr[b^\star(\state,d)=\bot\mid c=1]
  \leq\negl(\secp)
\]
in the honest execution. On a selected execution with
$b^\star\ne\bot$, agreement is equivalent to verifier acceptance.
On $b^\star=\bot$, the fallback guarantees agreement. Hence,
conditioned on reaching a selected execution, the agreement
probability lies between $p_1$ and $p_1+\nu$. The cutoff fallback
also guarantees agreement, so the overall agreement probability is
at least $p_1$. The uniform mask $r$, or the uniform fallback bit,
makes each party's output exactly unbiased.

It remains to prove security.  Fix a QPT eavesdropper
$\widetilde E$ that observes the complete bit-agreement transcript.
Given its guess for $k_A$, define its corresponding guess for the
unmasked verifier bit by
\[
  \widehat b
  :=
  \widetilde E(\tau_{\mathrm{BA}})\oplus r.
\]
Let
\[
  e:=\Pr[\widehat b=b^\star(\state,d)]
\]
in a selected challenge-$1$ execution, where the transcript includes
Alice's continue or abort flag. This event is false when
$b^\star(\state,d)=\bot$, since $\widehat b$ is a bit. On an aborted
attempt, take $r$ to be the public fallback bit for purposes of this
notation.

We construct a commuting prover for the generalized ToNC.  On
challenge $0$, it applies the honest measurement $P_0$.  On challenge
$1$, it samples
\[
  d\gets\mathsf{Sim}(1^\secp,\tau),
\]
simulates the public bit-agreement transcript for $\widetilde E$,
and outputs $(d,\widehat b)$. In this simulation it always supplies
a continue flag and a fresh uniform mask $r$. The preceding rejected
preparations are sampled independently according to their distribution
before the selected execution, with the same cutoff. The reduction
does not need to determine whether $b^\star(\state,d)=\bot$:
this event is negligible in both the honest and simulated distributions.
Consequently, by simulatability, its challenge-$1$ success probability
is $e\pm\negl(\secp)$. The challenge-$1$ strategy does not act on the
prover's quantum state and therefore commutes with the challenge-$0$
measurement. Soundness therefore implies
\begin{equation}
  (1-q)p_0+q e
  \leq \alpha+\negl(\secp).
  \label{eq:generalized-ToNC-eavesdropper-soundness}
\end{equation}
Combining~\eqref{eq:generalized-ToNC-completeness} and~\eqref{eq:generalized-ToNC-eavesdropper-soundness} gives
\begin{equation}
  e\leq p_1-\frac{g}{q}+\negl(\secp).
  \label{eq:generalized-ToNC-eavesdropper-bound}
\end{equation}

Finally, a successful eavesdropper guess on an agreeing execution
without fallback is counted in $e$. Both fallback events have
negligible probability, and the overall agreement probability is at
least $p_1$. Therefore,
\[
\begin{aligned}
  \Pr\bigl[
    \widetilde E(\tau_{\mathrm{BA}})=k_A
    \,\bigm|\,
    k_A=k_B
  \bigr]
  &\leq
  \frac{e+\negl(\secp)}{p_1}\\
  &\leq
  1-\frac{g}{q p_1}+\negl(\secp).
\end{aligned}
\]
This is the asserted bound on $\delta_{\mathrm{BA}}$. The bit-agreement amplification theorem applies when
\[
  \delta_{\mathrm{BA}}
  <
  \frac{2\epsilon_{\mathrm{BA}}}
       {1+\epsilon_{\mathrm{BA}}}.
\]
Substituting the bounds above, this condition becomes
\[
  1-\frac{2g}{q p_1}
  <
  \frac{2(2p_1-1)}{2p_1},
\]
which is equivalent to
\[
  \frac{2g}{q}>1-p_1.
\]
\end{proof}

\paragraph{The qubit-certification protocol of Brakerski et al.}
We next explain how the qubit-certification protocol of Brakerski
et al.~\cite{BCMVV} fits
Definition~\ref{def:generalized-ToNC}.  We recall the relevant
single-round protocol at a high-level.

The protocol is based on a noisy trapdoor claw-free (NTCF) family.
A public key $k$ specifies two distribution-valued functions
$f_{k,0}$ and $f_{k,1}$ with domain $\mathcal X$.
For each relevant image $y$, the verifier can use its trapdoor
to recover a matched pair of preimages $(x_0,x_1)$, whereas
finding both preimages without the trapdoor is computationally
hard. The efficient honest sampling procedure produces a joint
image--prover state negligibly close in trace distance to the
corresponding ideal state. Below we describe the honest execution
using this idealization, which changes the joint outcome
distributions by at most a negligible amount.

The NTCF family also specifies an efficiently computable injective
encoding
\[
    J:\mathcal X\longrightarrow\bits^w,
    \qquad w=\poly(\secp),
\]
and sets of good vectors
$G_{k,e,x}\subseteq\bits^w$ for $e\in\{0,1\}$ and $x\in\mathcal X$.
Membership is efficiently checkable given the trapdoor, and
\[
    \Pr_{d\leftarrow\bits^w}[d\notin G_{k,e,x}]
    \leq \operatorname{negl}(\secp).
\]
For the matched preimages associated with $y$, write
\[
    \widehat{G}_y := G_{k,0,x_0}\cap G_{k,1,x_1}.
\]
The adaptive hardcore-bit guarantee applies to equation vectors
$d\in\widehat{G}_y$.

In the preparation stage, the verifier samples a public key $k$
together with its trapdoor and sends $k$ to the prover.  The honest
prover evaluates the two functions in superposition, measures the
image register, and sends the resulting image $y$ to the verifier.
Conditioned on this value of $y$, the prover holds the state
\[
    \ket{\psi_y}
    =
    \frac{1}{\sqrt 2}
    \bigl(
       \ket{0}\ket{x_0}
       +
       \ket{1}\ket{x_1}
    \bigr),
\]
while the verifier uses its trapdoor to recover the matched
preimages $x_0,x_1$.

There are two possible challenges.

\begin{itemize}
    \item If $c=0$, the prover is given the \emph{preimage challenge}.
    The honest prover measures $\ket{\psi_y}$ in the computational
    basis and returns $(e,x_e)$ for a uniformly random
    $e\in\bits$.  The verifier checks that $x_e$ is a valid preimage
    of $y$ under $f_{k,e}$.

    \item If $c=1$, the prover is given the \emph{equation challenge}.
    The honest prover applies $J$ to its preimage register, measures
    all registers in the Hadamard basis, and returns
    \[
        (d,u)\in\bits^w\times\bits.
    \]

    If $d\in\widehat{G}_y$, the verifier accepts if and only if
    \begin{equation}
      u=d\cdot\bigl(J(x_0)\oplus J(x_1)\bigr),
    \end{equation}
    where the inner product is over $\mathbb{F}_2$.
    If $d\notin\widehat{G}_y$, it accepts with probability $1/2$
    using an independent coin, as in~\cite{BCMVV}.
    In our formalism, this private verification coin is included
    in $\state$.

\end{itemize}

For every fixed verifier state and every $d\in\widehat{G}_y$,
there is exactly one accepting bit,
\[
  u^\star(\state,d)=d\cdot\bigl(J(x_0)\oplus J(x_1)\bigr).
\]
For $d\notin\widehat{G}_y$, fixing the verifier's acceptance coin
makes both bits accepted or both rejected, so $u^\star(\state,d)=\bot$.
The density guarantees for the good-vector sets and the honest
sampling approximation imply that this exceptional event has
negligible probability. Thus the original test satisfies the relaxed
unique-accepting-bit condition without changing its verification rule;
the derived bit-agreement protocol aborts the attempt on an exceptional
vector.

It remains to verify that $d$ is simulatable from the public
transcript.  Applying the honest equation measurement to
$\ket{\psi_y}$ produces, up to a reordering of registers, the state
\[
    2^{-w/2}
    \sum_{d\in\bits^w}
       (-1)^{d\cdot J(x_0)}
       \ket{d}
       \ket{
          d\cdot
          \bigl(J(x_0)\oplus J(x_1)\bigr)
       }.
\]
Consequently, $d$ is uniformly distributed over $\bits^w$, even
conditioned on the preparation transcript and the verifier's private
state.  The simulator required by
Definition~\ref{def:generalized-ToNC} therefore simply samples
\[
    d^\star\gets\bits^w
\]
uniformly at random.  The resulting simulated and real distributions
are identical, up to the negligible approximation inherent in the
honest NTCF sampling procedure.

The completeness is $1-\negl(\secp)$, accounting for both the NTCF
sampling approximation and the negligible probability of an invalid
equation vector. The commuting soundness is at most
$3/4+\negl(\secp)$, by the adaptive hardcore-bit property.
Thus, for some negligible function $\eta$, the protocol is a
generalized $(3/4,1-\eta(\secp))$ test of non-commutation.
For the distinguished challenge,
\[
    q=\Pr[c=1]=\frac12
    \qquad\text{and}\qquad
    p_1=1-\negl(\secp).
\]
These parameters have constant slack in
\eqref{eq:generalized-ToNC-amplification-condition}; one may weaken
the guarantees to fixed constants before invoking amplification.
It therefore satisfies the condition of
Proposition~\ref{prop:generalized-ToNC-to-BA}, and hence implies
classical-communication key agreement.

\section{Disclosure of AI Usage}
The only part of this work where AI was used beyond proofreading and basic rephrasing, is Appendix~\ref{sec:generalized-tonc}. There, the authors interacted with ChatGPT Sol 5.6 while designing the Generalized ToNC primitive and formulating the proof outline, and then AI was further used to formalize it. The authors then revised the proof for clarity.  

\clearpage\fi

\ifsubmission
\bibliographystyle{plain}
\else
\bibliographystyle{alpha}
\fi

\newcommand{\etalchar}[1]{$^{#1}$}

\ifsubmission
\newpage
\appendix

\else
\fi

\end{document}